\documentclass[journal]{IEEEtran}
\ifCLASSINFOpdf
\else
\fi
\usepackage{nicefrac}       % compact symbols for 1/2, etc.
\usepackage{amsfonts,amsopn}
\usepackage{algpseudocode}
\usepackage[Algorithmus]{algorithm}
\usepackage{enumitem,cite,url,bm,amsmath,amssymb,graphicx,subfigure,enumitem,float,xfrac,mathtools}
\usepackage{multirow}

\usepackage{caption}
\usepackage{color}% Include colors for document elements
\usepackage{dcolumn}% Align table columns on decimal point
\usepackage{etoolbox}
\makeatletter
\patchcmd{\maketitle}
 {\def\@makefnmark}
 {\def\@makefnmark{}\def\useless@macro}
 {}{}
\makeatother

\usepackage{tikz}
\usetikzlibrary{bayesnet}
\usetikzlibrary{arrows,shapes,snakes,automata,backgrounds,petri}

\newcommand{\vw}{{\bm w}}
\newcommand{\vx}{{\bm x}}
\newcommand{\vy}{{\bm y}}
\newcommand{\vz}{{\bm z}}
\newcommand{\vr}{{\bm r}}
\newcommand{\vs}{{\bm s}}
\newcommand{\vv}{{\bm v}}
\newcommand{\vq}{{\bm q}}

\newcommand{\vA}{{\bm A}}
\newcommand{\vK}{{\bm K}}

\usepackage{chngcntr}
\usepackage{apptools}
\AtAppendix{\counterwithin{lemma}{section}}

\newtheorem{theorem}{Theorem}
\newtheorem{assumption}{Assumption}
\newtheorem{lemma}{Lemma}
\newtheorem{definition}{Definition}
\newtheorem{corollary}{Corollary}

\DeclareMathOperator{\erfcx}{erfcx}

\newenvironment{proof}[1][Proof]{\begin{trivlist}
\item[\hskip \labelsep {\bfseries #1}]}{\end{trivlist}}

\algnewcommand{\IIf}[1]{\State\algorithmicif\ #1\ \algorithmicthen}
\algnewcommand{\EndIIf}{\unskip\ \algorithmicend\ \algorithmicif}

\begin{document}
\title{Sparse Signal Recovery using Generalized Approximate Message Passing with Built-in Parameter Estimation}

\author{Shuai Huang, ~\IEEEmembership{Student Member,~IEEE,} and Trac D. Tran, ~\IEEEmembership{Fellow,~IEEE}% <-this % stops a space
\thanks{This work is supported by the National Science Foundation under grants NSF-CCF-1117545, NSF-CCF-1422995 and NSF-EECS-1443936.}
\thanks{The authors are with the Department of Electrical and Computer Engineering, Johns Hopkins University, Baltimore, MD, 21218 USA (email: shuaihuang@jhu.edu; trac@jhu.edu).}}% <-this % stops a space

\maketitle

\begin{abstract}
The generalized approximate message passing (GAMP) algorithm under the Bayesian setting shows advantage in recovering under-sampled sparse signals from corrupted observations. Compared to conventional convex optimization methods, it has a much lower complexity and is computationally tractable. In the GAMP framework, the sparse signal and the observation are viewed to be generated according to some pre-specified probability distributions in the input and output channels. However, the parameters of the distributions are usually unknown in practice. In this paper, we propose an extended GAMP algorithm with built-in parameter estimation (PE-GAMP) and present its empirical convergence analysis. PE-GAMP treats the parameters as unknown random variables with simple priors and jointly estimates them with the sparse signals. Compared with Expectation Maximization (EM) based parameter estimation methods, the proposed PE-GAMP could draw information from the prior distributions of the parameters to perform parameter estimation. It is also more robust and much simpler, which enables us to consider more complex signal distributions apart from the usual Bernoulli-Gaussian (BGm) mixture distribution. Specifically, the formulations of Bernoulli-Exponential mixture (BEm) distribution and Laplace distribution are given in this paper. Simulated noiseless sparse signal recovery experiments demonstrate that the performance of the proposed PE-GAMP matches the oracle GAMP algorithm that knows the true parameter values. When noise is present, both the simulated experiments and the real image recovery experiments show that the proposed PE-GAMP is still able to maintain its robustness and outperform EM based parameter estimation method when the sampling ratio is small. Additionally, using the BEm formulation of the proposed PE-GAMP, we can successfully perform non-negative sparse coding of local image patches and provide useful features for the image classification task.
\end{abstract}

\begin{IEEEkeywords}
Sparse signal recovery, approximate message passing, parameter estimation, belief propagation, compressive sensing, non-negative sparse coding, image recovery, image classification. 
\end{IEEEkeywords}

% \fi
\IEEEpeerreviewmaketitle

\section{Introduction}
Sparse signal recovery (SSR) is the key topic in Compressive Sensing (CS) \cite{Decode05,RUP06,CS06,SRRP06}, it lays the foundation for applications such as dictionary learning \cite{DL06}, sparse representation-based classification \cite{SRC09}, etc. Specifically, SSR tries to recover the sparse signal $\vx\in\mathbb{R}^N$ given a $M\times N$ sensing matrix $\vA$ and a measurement vector $\vy=\vA\vx+\vw\in\mathbb{R}^M$, where $M<N$ and $\vw\in\mathbb{R}^M$ is the unknown noise introduced in this process. Although the problem itself is ill-posed, perfect recovery is still possible provided that $\vx$ is sufficiently sparse and $\vA$ is incoherent enough \cite{Decode05}. Lasso \cite{Lasso94}, a.k.a $l_1$-minimization, is one of most popular approaches proposed to solve this problem:
\begin{align}
\label{eq:lasso}
\arg\min_{\vx}\quad\|\vy-\vA\vx\|_2^2+\gamma\|\vx\|_1\,,
\end{align}
where $\|\vy-\vA\vx\|_2^2$ is the data-fidelity term, $\|\vx\|_1$ is the sparsity-promoting term, and $\gamma$ balances the trade-off between them.

From a probabilistic view, Lasso is equivalent to a maximum likelihood (ML) estimation of the signal $\vx$ under the assumption that the entries of $\vx$ are i.i.d. distributed following the Laplace distribution $p(x_j)\propto\exp(-\lambda|x_j|)$, and those of $\vw$ are i.i.d. distributed following the Gaussian distribution $p(w_i)\propto\exp\left(-\sfrac{w_i^2}{2\theta}\right)$. Let $\vz=\vA\vx$, we have $p(y_i|\vx)\propto\exp\left(-\sfrac{(y_i-z_i)^2}{2\theta}\right)$. The ML estimation is then $\arg\max_{\vx}p(\vx,\vy)$, which is essentially the same as (\ref{eq:lasso}). In general, SSR can be described by the Bayesian model from \cite{GAMP11}, as is shown in Fig. \ref{fig:bayesian_model}.
\begin{figure}[tbp]
\centering
\includegraphics[width=3in]{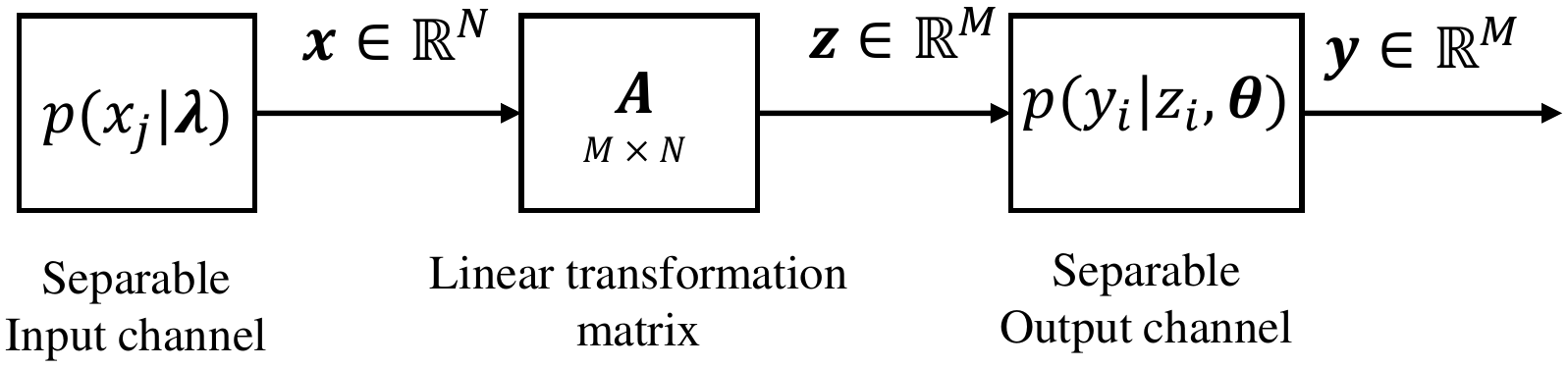}
\caption{A probabilistic view of the sparse signal recovery \cite{GAMP11}: The signal $\vx$ is estimated given the output vector $\vy$, the channel transition probability functions $p(x_j|\boldsymbol\lambda)$, $p(y_i|z_i,\boldsymbol\theta)$ and the transformation matrix $\vA$. $\{\boldsymbol\lambda,\boldsymbol\theta\}$ denote the parameters of the probability models and are usually unknown.}
\label{fig:bayesian_model}
\end{figure}

Under the Bayesian setting it is possible to design efficient iterative algorithms to compute either the maximum a posterior (MAP) or minimum mean square error (MMSE) estimate of the signal $\vx$. Most notable among them are the ``message-passing'' based algorithms \cite{AMP09,AMP10,AMP11}. They perform probabilistic inferences on the corresponding factor graph using Gaussian and/or quadratic approximations of loopy belief propagation (loopy BP), a.k.a. message passing \cite{Graph08}. Based on different inference tasks, loopy BP has two variants: \emph{sum-product} message passing for the MMSE estimate of $\vx$ and \emph{max-sum} message passing for the MAP estimate of $\vx$. \cite{AMP09,AMP10,AMP11} proposed the \emph{approximate} message passing (AMP) algorithm based on a quadratic approximation of \emph{max-sum} message passing. It has low complexity and can be used to find solutions of Lasso accurately. In fact, AMP is able to match the performance of theoretical Lasso in noiseless signal recovery experiments \cite{AMP09}. The asymptotic behavior of the variables in the AMP algorithm can be concisely described by a set of state evolution equations, and their empirical convergences are guaranteed in the large system limit for $\vA$ with i.i.d Gaussian entries \cite{AMP11}. 

\subsection{Prior Work}
Various methods based on the above AMP framework has been proposed to perform sparse signal recovery \cite{SURE_AMP15,D_AMP16,GAMP11}. \cite{SURE_AMP15} treats each AMP iteration as a signal denoising process and introduces the denoiser constructed from the Stein's unbiased risk estimate (SURE) into the AMP algorithm (SURE-AMP). Three Different kernel functions are also proposed to get linear parameterization of the SURE based denoiser, which serves as the objective function to be minimized. \cite{D_AMP16} also provides an extension to the AMP algorithm by including a denoiser within the AMP iterations (D-AMP) and demonstrates its effectiveness in recovering natural images.

In \cite{GAMP11}, a generalized version of the AMP algorithm (GAMP) is proposed to work with essentially arbitrary input and output channel distributions. It can approximate both the \emph{sum-product} and \emph{max-sum} message passings using only scalar estimations and linear transforms. Similar to AMP, GAMP can be described by state evolution equations and its empirical convergence can also be shown using an extension of the analysis in \cite{AMP11}. The parameters $\{\boldsymbol\lambda, \boldsymbol\theta\}$ in the input and output channels are usually unknown, and need to be decided for the AMP/GAMP algorithm. In this paper, we shall propose an extension to the GAMP framework by treating the parameters as unknown random variables with simple prior distributions and estimating them jointly with the signal $\vx$.

The Expectation-Maximization (EM) \cite{EM77} algorithm has been proposed to perform parameter estimation for the GAMP algorithm in \cite{EMBGAMP11,EMGMAMP13, AMP_CPE14,EM_GAMP_Florent12}. Specifically, EM treats $\vx$ as the hidden variable and tries to find the parameters that maximize $p(\vy|\boldsymbol\lambda,\boldsymbol\theta)$ by maximizing $\mathbb{E}[\log p(\vx,\vw;\boldsymbol\lambda,\boldsymbol\theta)|\vy,\hat{\boldsymbol\lambda}^{(t)},\hat{\boldsymbol\theta}^{(t)}]$ iteratively. \cite{EMGMAMP13} assumes the signal is generated according to i.i.d Bernoulli-Gaussian mixture (BGm) distribution and shows that EM-BGm-GAMP is still able to recover the sparse signals successfully even if the assumption is not satisfied. In \cite{AMP_CPE14}, a more generalized EM based parameter estimation method is proposed; the complete state evolution analysis and empirical convergence proofs of said algorithm are also presented as the theoretical work. While \cite{EM_GAMP_Florent12} assumes the signal to be i.i.d. Bernoulli-Gaussian distributed, it also gives the asymptotic state evolution and replica analysis for the EM parameter estimation with respect to their derivation of the GAMP algorithm.

\subsection{Main Contributions}
By treating the parameters $\boldsymbol\lambda,\boldsymbol\theta$ as random variables with simple priors, we can integrate the parameter estimation and signal recovery under the same framework: PE-GAMP. This enables us to compute the posterior distributions of the parameters directly from loopy belief propagation.
\begin{itemize}
\item {\bfseries \emph{Sum-product} message passing:} The ``marginal'' posterior distributions $\left\{p(\boldsymbol\lambda|\vy), p(\boldsymbol\theta|\vy)\right\}$ can be obtained.
\item {\bfseries \emph{Max-sum} message passing:} The ``joint'' posterior distribution $\left\{p(\boldsymbol\lambda,\tilde{\vx}|\vy), p(\boldsymbol\theta,\tilde{\vx}|\vy)\right\}$ can be obtained, where $\tilde{\vx}$ are the values that maximizes the joint posterior distribution. 
\end{itemize}

For the \emph{sum-product} message passing, if the input and output channel distributions $p(\vx|\boldsymbol\lambda),p(\vy|\vz,\boldsymbol\theta)$ are simple enough so that the integration involved in the message passing process can be computed, the parameter estimation will be automatically taken care of and no special treatments are needed. However, in practice the channel distributions are usually complicated, and the integration usually doesn't have closed-form solutions. In this case, we can compute the MMSE or MAP estimates of the parameters $\boldsymbol\lambda,\boldsymbol\theta$ using Dirac delta approximations of the posterior distributions and use them to simplify the message passing process. For the \emph{max-sum} message passing, the maximization problem involving multiple variables can be efficiently solved by using the approximate maximizing parameters. As can be seen from the Appendix \ref{app:sp_message_passing},\ref{app:ms_message_passing}, the MMSE scalar estimates of the parameters involve integration and are often quite difficult to compute; MAP scalar estimates of the parameters are thus preferred since they are much easier to compute and there are many maximization methods we can choose from. 

Following the line of work on the state evolution analysis of the AMP related algorithms in \cite{AMP11,GAMP11,AMP_CPE14}, we can write the state evolution equations for the proposed PE-GAMP and prove the empirical convergences of the involved variables. 

Previous EM based parameter estimation methods can only be used with \emph{sum-product} message passing, Since it relies on the marginal probability $p(\vx|\vy,\hat{\boldsymbol\lambda}^{(t)},\hat{\boldsymbol\theta}^{(t)})$ to compute the expectation. While the proposed PE-GAMP could be applied to both \emph{sum-product} and \emph{max-sum} message passings, which gives MMSE and MAP estimations of the signal respectively. 

Additionally, the proposed PE-GAMP could draw information from the prior distributions of the parameters to perform parameter estimation. It is also more robust and much simpler, which enables us to consider more complex signal distributions apart from the usual Bernoulli-Gaussian mixture distribution. Specifically, in Section \ref{sec:numerical_results} and Appendix \ref{app:map_pe}, input channels with three different distributions are considered: Bernoulli-Gaussian mixture distribution, Bernoulli-Exponential mixture distribution and Laplace distribution; while the output channel assumes the noise is additive white Gaussian noise. Both simulated and real experiments demonstrate the advantage the proposed PE-GAMP has over the previous EM based parameter estimation methods in both robustness and performance when the sampling ratio is small. With more signal distributions incorporated to the framework, the PE-GAMP also enjoys wider applicabilities and provides more possibilities for the sparse signal recovery task.

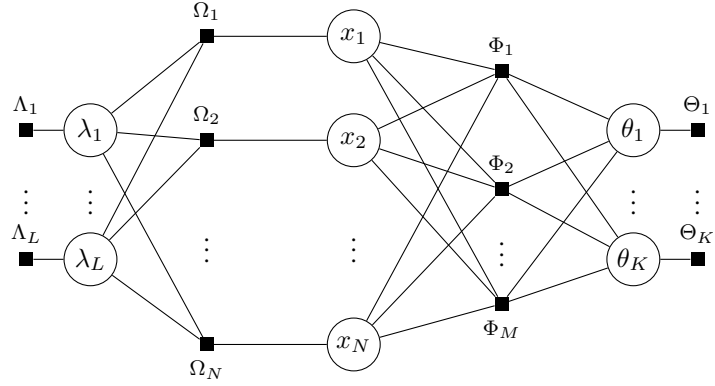
\begin{figure*}[tbp]
\begin{center}
\begin{tabular}{cc}
% model_pca2.tex
%
% Copyright (C) 2010,2011 Laura Dietz
% Copyright (C) 2012 Jaakko Luttinen
%
% This file may be distributed and/or modified
%
% 1. under the LaTeX Project Public License and/or
% 2. under the GNU General Public License.
%
% See the files LICENSE_LPPL and LICENSE_GPL for more details.

%\beginpgfgraphicnamed{model-pca}
\begin{tikzpicture}

  % Define nodes

  \node[latent] (lambda_1) {$\lambda_1$};
  \factor[left=of lambda_1] {Lambda_1-f} {above:$\Lambda_1$} {lambda_1} {};
  \node[latent, below = 1 of lambda_1] (lambda_L) {$\lambda_L$};
  \factor[left=of lambda_L] {Lambda_L-f} {above:$\Lambda_L$} {lambda_L} {};
  
  \path (lambda_1) -- node[auto=false]{\vdots} (lambda_L);
  \path (Lambda_1-f) -- node[auto=false]{\vdots} (Lambda_L-f);

  \node[latent, above=0.533 of lambda_1, xshift=3.5cm] (x_1) {$x_1$};
  \node[latent, below=0.666 of x_1] (x_2) {$x_2$};
  \node[latent, below=2 of x_2] (x_N) {$x_N$};
  
  \path (x_2) -- node[auto=false]{\vdots} (x_N);
  
  \factor[left=0 of x_1, xshift=-1.5cm] {Omega_1-f} {above:$\Omega_1$} {x_1, lambda_1, lambda_L}{};
  \factor[left=0 of x_2, xshift=-1.5cm] {Omega_2-f} {above:$\Omega_2$} {x_2, lambda_1, lambda_L}{};
  \factor[left=0 of x_N, xshift=-1.5cm] {Omega_N-f} {below:$\Omega_N$} {x_N, lambda_1, lambda_L}{};
  \path(Omega_2-f) -- node[auto=false]{\vdots} (Omega_N-f);
  
  \node[latent, right=0 of lambda_1, xshift=6.5cm] (theta_1) {$\theta_1$};
  \node[latent, right=0 of lambda_L, xshift=6.5cm] (theta_K) {$\theta_K$};
  \path(theta_1) -- node[auto=false]{\vdots} (theta_K);
  
  \factor[right=of theta_1] {Theta_1-f} {above:$\Theta_1$}{theta_1}{};
  \factor[right=of theta_K] {Theta_K-f} {above:$\Theta_K$}{theta_K}{};
  \path(Theta_1-f) -- node[auto=false]{\vdots} (Theta_K-f);

  \factor[above=.333 of theta_1, xshift=-1.75cm] {Phi_1-f} {above:$\Phi_1$} {theta_1, theta_K, x_1, x_2, x_N} {};
  \factor[below=.333 of theta_1, xshift=-1.75cm] {Phi_2-f} {above:$\Phi_2$} {theta_1, theta_K, x_1, x_2, x_N} {};
  \factor[below=1.333 of Phi_2-f] {Phi_M-f} {below:$\Phi_M$} {theta_1, theta_K, x_1, x_2, x_N} {};
  
  \path(Phi_2-f) -- node[auto=false]{\vdots} (Phi_M-f);

  % Y
  %\node[obs]          (y)   {$y$}; %
  %\factor[above=of y] {y-f} {left:$\mathcal{N}$} {} {} ; %

  % W and X
  %\node[det, above=of y]            (dot) {dot} ; % 
  %\node[latent, above left=1.2 of dot]  (w)   {$\mathbf{w}$}; %
  %\node[latent, above right=1.2 of dot] (x)   {$\mathbf{x}$}; %

  % W hyperparameters
  %\node[const, above=1.2 of w, xshift=-0.5cm] (mw) {$\mu_w$} ; %
  %\node[const, above=1.2 of w, xshift=0.5cm]  (aw) {$\alpha_w$} ; %

  % X hyperparameters
  %\node[const, above=1.2 of x, xshift=-0.5cm] (mx) {$\mu_x$} ; %
  %\node[const, above=1.2 of x, xshift=0.5cm]  (ax) {$\alpha_x$} ; %

  % noise
  %\node[latent, right=2.5cm of y-f]         (t)   {$\tau$}; %
  %\node[const, above=of t, xshift=-0.5cm] (at)  {$\alpha_\tau$} ; %
  %\node[const, above=of t, xshift=0.5cm]  (bt)  {$\beta_\tau$} ; %

  % Factors
  %\factor[above=of w] {w-f} {left:$\mathcal{N}$} {mw,aw} {w} ; %
  %\factor[above=of x] {x-f} {left:$\mathcal{N}$} {mx,ax} {x} ; %
  %\factor[above=of t] {t-f} {left:$\mathcal{G}$} {at,bt} {t} ; %
  %\factoredge {dot,t} {y-f} {y} ; %

  % Connect w and x to the dot node
  %\edge[-] {w,x} {dot} ;

  % Plates
  %\plate {yx} { %
  %  (y)(y-f)(y-f-caption) %
  %  (x)(x-f)(x-f-caption) %
  %  (dot) %
  %} {$N$} ;
  %\plate {} {%
  %  (y)(y-f)(y-f-caption) %
  %  (w)(w-f)(w-f-caption) %
  %  (dot) %
  %  (yx.north west)(yx.south west) %
  %} {$M$} ;

\end{tikzpicture}
%\endpgfgraphicnamed

%%% Local Variables: 
%%% mode: tex-pdf
%%% TeX-master: "example"
%%% End:  
\end{tabular}
\end{center}
\caption{The factor graph for the proposed PE-GAMP. ``$\blacksquare$'' represents the factor node, and ``$\bigcirc$'' represents the variable node. $\boldsymbol\lambda=\{\lambda_1,\cdots,\lambda_L\}$ and $\boldsymbol\theta=\{\theta_1,\cdots,\theta_K\}$ are the parameters. $\vx=[x_1,\cdots,x_N]^\textrm{T}$ is the sparse signal.}
\label{fig:factor_graph_pegamp}
\end{figure*}

\section{GAMP with Built-in Parameter Estimation}
The generalized factor graph for the proposed PE-GAMP framework that treats the parameters as random variables is shown in Fig. \ref{fig:factor_graph_pegamp}. Inference tasks performed on the factor graph rely on the ``messages'' passed among connected nodes of the graph. Here we adopt the same notations used by \cite{GAMP11}. Take the messages being passed between the factor node $\Phi_m$ and the variable node $x_n$ for example, $\Delta_{\Phi_m\rightarrow x_n}$ is the message from $\Phi_m$ to $x_n$, and $\Delta_{\Phi_m\leftarrow x_n}$ is the message from $x_n$ to $\Phi_m$. Both $\Delta_{\Phi_m\rightarrow x_n}$ and $\Delta_{\Phi_m\leftarrow x_n}$ can be viewed as functions of $x_n$. In the following section \ref{sec:sum_product} and \ref{sec:max_sum}, we give the messages being passed on the generalized factor graph in $\log$ domain for the \emph{sum-product} message passing algorithm and the \emph{max-sum} message passing algorithm respectively.

\subsection{Sum-product Message Passing}
\label{sec:sum_product}
\emph{Sum-product} message passing is used to compute the marginal distributions of the random variables in the graph: $p(\vx|\vy), p(\boldsymbol\lambda|\vy), p(\boldsymbol\theta|\vy)$. In the following, we first present the \emph{sum-product} message updates equations in the $(t+1)$-th iteration.
\begin{subequations}
\label{eq:sp_fv_signal}
\begin{align}
\label{eq:sp_fv_phi_x}
\begin{split}
\Delta^{(t+1)}_{\Phi_m\rightarrow x_n}&=\textrm{const}+\log\int_{\vx\backslash x_n,\boldsymbol\theta}\left[\vphantom{\textstyle\sum_{j\neq n}\Delta^{(t)}_{\Phi_m\leftarrow x_j}}\Phi_m\left(y_m, \vx, \boldsymbol\theta\right)\right.\\
&\quad\quad\left.\times\exp\left(\textstyle\sum_{j\neq n}\Delta^{(t)}_{\Phi_m\leftarrow x_j}+\textstyle\sum_v\Delta^{(t)}_{\Phi_m\leftarrow \theta_v}\right)\right]
\end{split}\\
\label{eq:sp_vf_x_phi}
\Delta^{(t+1)}_{\Phi_m\leftarrow x_n}&=\textrm{const}+\Delta^{(t+1)}_{\Omega_n\rightarrow x_n}+\textstyle\sum_{i\neq m}\Delta^{(t+1)}_{\Phi_i\rightarrow x_n}\\
\label{eq:sp_fv_omega_x}
\Delta^{(t+1)}_{\Omega_n\rightarrow x_n}&=\textrm{const}+\log\int_{\boldsymbol\lambda}\Omega_n(x_n,\boldsymbol\lambda)\cdot\exp\left(\textstyle\sum_u\Delta^{(t)}_{\Omega_n\leftarrow\lambda_u}\right)\\
\label{eq:sp_vf_x_omega}
\Delta^{(t+1)}_{\Omega_n\leftarrow x_n}&=\textrm{const}+\textstyle\sum_i\Delta^{(t+1)}_{\Phi_i\rightarrow x_n}\,,
\end{align}
\end{subequations}
where $\vx\backslash x_n$ denotes the sequence obtained by removing $x_n$ from $\vx$, $\Phi_m(y_m,\vx,\boldsymbol\theta)=p(y_m|\vx,\boldsymbol\theta)$ and $\Omega_n(x_n,\boldsymbol\lambda)= p(x_n|\boldsymbol\lambda)$. Similarly, we can write the message updates involving the variable nodes $\lambda_l,\theta_k$ as follows:
\begin{subequations}
\label{eq:sp_fv_parameter}
\begin{align}
\label{eq:sp_fv_omega_lambda}
\begin{split}
\Delta^{(t+1)}_{\Omega_n\rightarrow \lambda_l}&=\textrm{const}+\log\int_{x_n,\boldsymbol\lambda\backslash\lambda_l}\left[\vphantom{\textstyle\sum_{j\neq n}\Delta^{(t)}_{\Phi_m\leftarrow x_j}}\Omega_n(x_n,\boldsymbol\lambda)\right.\\
&\quad\quad\left.\times\exp\left(\Delta^{(t+1)}_{\Omega_n\leftarrow x_n}+\textstyle\sum_{u\neq l}\Delta^{(t)}_{\Omega_n\leftarrow\lambda_u}\right)\right]
\end{split}\\
\label{eq:sp_vf_lambda_omega}
\Delta^{(t+1)}_{\Omega_n\leftarrow\lambda_l}&=\textrm{const}+\textstyle\sum_{j\neq n}\Delta^{(t+1)}_{\Omega_j\rightarrow\lambda_l}+\log p(\lambda_l)\\
\label{eq:sp_fv_phi_theta}
\begin{split}
\Delta^{(t+1)}_{\Phi_m\rightarrow\theta_k}&=\textrm{const}+\log\int_{\boldsymbol\theta\backslash\theta_k,\vx}\left[\vphantom{\textstyle\sum_{j\neq n}\Delta^{(t)}_{\Phi_m\leftarrow x_j}}\Phi_m\left(y_m, \vx, \boldsymbol\theta\right)\right.\\
&\quad\quad\left.\times\exp\left(\textstyle\sum_j\Delta^{(t)}_{\Phi_m\leftarrow x_j}+\textstyle\sum_{v\neq k}\Delta^{(t)}_{\Phi_m\leftarrow\theta_v}\right)\right]
\end{split}\\
\label{eq:sp_vf_theta_phi}
\Delta^{(t+1)}_{\Phi_m\leftarrow\theta_k}&=\textrm{const}+\textstyle\sum_{i\neq m}\Delta^{(t+1)}_{\Phi_i\rightarrow\theta_k}+\log p(\theta_k)\,,
\end{align}
\end{subequations}
where $p(\lambda_l),p(\theta_k)$ are the pre-specified priors of the parameters. The approximated implementations of \emph{sum-product} message passing in terms of (\ref{eq:sp_fv_signal}) and (\ref{eq:sp_fv_parameter}) are detailed in Appendix \ref{app:sp_message_passing}. Let $\Gamma(x_n), \Gamma(\lambda_l), \Gamma(\theta_k)$ denote the factor nodes in the neighborhood of the variable nodes $x_n, \lambda_l, \theta_k$ respectively, we have the following posterior marginals:
\begin{subequations}
\label{eq:sm_post_dist}
\begin{align}
\label{eq:pm_x}
\begin{split}
p(x_n|\vy)&\propto\exp\Delta^{(t+1)}_{x_n}\\
&=\exp\left(\Delta^{(t+1)}_{\Omega_n\rightarrow x_n}+\textstyle\sum_{\Phi_m\in\Gamma(x_n)}\Delta^{(t+1)}_{\Phi_m\rightarrow x_n}\right)
\end{split}\\
\label{eq:pm_lambda}
\begin{split}
p(\lambda_l|\vy)&\propto\exp\Delta^{(t+1)}_{\lambda_l}\\
&=\exp\left(\log p(\lambda_l)+\textstyle\sum_{\Omega_n\in\Gamma(\lambda_l)}\Delta^{(t+1)}_{\Omega_n\rightarrow\lambda_l}\right)
\end{split}\\
\label{eq:pm_theta}
\begin{split}
p(\theta_k|\vy)&\propto\exp\Delta^{(t+1)}_{\theta_k}\\
&=\exp\left(\log p(\theta_k)+\textstyle\sum_{\Phi_m\in\Gamma(\theta_k)}\Delta^{(t+1)}_{\Phi_m\rightarrow\theta_k}\right)\,.
\end{split}
\end{align}
\end{subequations}
Using $p(x_n|\vy)$, the MMSE estimate of $\vx$ can then be computed: 
\begin{align}
\hat{x}_n=\mathbb{E}\left[x_n|\vy\right] = \int_{x_n}x_np(x_n|\vy)\,.
\end{align}

\subsection{Max-sum Message Passing}
\label{sec:max_sum}
\emph{Max-sum} message passing is used to compute the ``joint'' MAP estimates of the random variables in the graph: 
\begin{align}
(\hat{\vx},\hat{\boldsymbol\lambda},\hat{\boldsymbol\theta})=\arg\max_{\vx,\boldsymbol\lambda,\boldsymbol\theta}\,p(\vx,\boldsymbol\lambda,\boldsymbol\theta|\vy)\,.
\end{align}
For the \emph{max-sum} message passing, the message updates from the variable nodes to the factor nodes are the same as the aforementioned \emph{sum-product} message updates, i.e. (\ref{eq:ms_vf_x_phi}, \ref{eq:ms_vf_x_omega}, \ref{eq:ms_vf_lambda_omega}, \ref{eq:ms_vf_theta_phi}). We only need to change the message updates from the factor nodes to the variable nodes by replacing $\int$ with $\max$. Specifically, we have the following message updates between the variable node $x_n$ and the factor nodes $\Phi_m,\Omega_n$ in the $(t+1)$-th iteration: 
\begin{subequations}
\label{eq:ms_fv_signal}
\begin{align}
\label{eq:ms_fv_phi_x}
\begin{split}
\Delta^{(t+1)}_{\Phi_m\rightarrow x_n}&=\textrm{const}+\max_{\vx\backslash x_n,\boldsymbol\theta}\left[\vphantom{\textstyle\sum_{j\neq n}\Delta^{(t)}_{\Phi_m\leftarrow x_j}}\log\Phi_m\left(y_m, \vx, \boldsymbol\theta\right)\right.\\
&\quad\quad\left.+\textstyle\sum_{j\neq n}\Delta^{(t)}_{\Phi_m\leftarrow x_j}+\textstyle\sum_v\Delta^{(t)}_{\Phi_m\leftarrow \theta_v}\right]
\end{split}\\
\label{eq:ms_vf_x_phi}
\Delta^{(t+1)}_{\Phi_m\leftarrow x_n}&=\textrm{const}+\Delta^{(t+1)}_{\Omega_n\rightarrow x_n}+\textstyle\sum_{i\neq m}\Delta^{(t+1)}_{\Phi_i\rightarrow x_n}\\
\label{eq:ms_fv_omega_x}
\Delta^{(t+1)}_{\Omega_n\rightarrow x_n}&=\textrm{const}+\max_{\boldsymbol\lambda}\,\left[\log\Omega_n(x_n,\boldsymbol\lambda)+\textstyle\sum_u\Delta^{(t)}_{\Omega_n\leftarrow\lambda_u}\right]\\
\label{eq:ms_vf_x_omega}
\Delta^{(t+1)}_{\Omega_n\leftarrow x_n}&=\textrm{const}+\textstyle\sum_i\Delta^{(t+1)}_{\Phi_i\rightarrow x_n}\,.
\end{align}
\end{subequations}
The message updates involving the variable nodes $\lambda_l,\theta_k$ are then:
\begin{subequations}
\label{eq:ms_fv_parameter}
\begin{align}
\label{eq:ms_fv_omega_lambda}
\begin{split}
\Delta^{(t+1)}_{\Omega_n\rightarrow \lambda_l}&=\textrm{const}+\max_{x_n,\boldsymbol\lambda\backslash\lambda_l}\left[\vphantom{\textstyle\sum_{j\neq n}\Delta^{(t)}_{\Phi_m\leftarrow x_j}}\log\Omega_n(x_n,\boldsymbol\lambda)\right.\\
&\quad\quad\left.+\Delta^{(t+1)}_{\Omega_n\leftarrow x_n}+\textstyle\sum_{u\neq l}\Delta^{(t)}_{\Omega_n\leftarrow\lambda_u}\right]
\end{split}\\
\label{eq:ms_vf_lambda_omega}
\Delta^{(t+1)}_{\Omega_n\leftarrow\lambda_l}&=\textrm{const}+\textstyle\sum_{j\neq n}\Delta^{(t+1)}_{\Omega_j\rightarrow\lambda_l}+\log p(\lambda_l)\\
\label{eq:ms_fv_phi_theta}
\begin{split}
\Delta^{(t+1)}_{\Phi_m\rightarrow\theta_k}&=\textrm{const}+\max_{\boldsymbol\theta\backslash\theta_k,\vx}\left[\vphantom{\textstyle\sum_{j\neq n}\Delta^{(t)}_{\Phi_m\leftarrow x_j}}\log\Phi_m\left(y_m, \vx, \boldsymbol\theta\right)\right.\\
&\quad\quad\left.+\textstyle\sum_j\Delta^{(t)}_{\Phi_m\leftarrow x_j}+\textstyle\sum_{v\neq k}\Delta^{(t)}_{\Phi_m\leftarrow\theta_v}\right]
\end{split}\\
\label{eq:ms_vf_theta_phi}
\Delta^{(t+1)}_{\Phi_m\leftarrow\theta_k}&=\textrm{const}+\textstyle\sum_{i\neq m}\Delta^{(t+1)}_{\Phi_i\rightarrow\theta_k}+\log p(\theta_k)\,.
\end{align}
\end{subequations}
The approximated implementations of the \emph{max-sum} message passing are detailed in Appendix \ref{app:ms_message_passing}. Similarly, we have the following posterior distributions that are different from those in (\ref{eq:sm_post_dist}):
\begin{subequations}
\begin{align}
\begin{split}
&p(x_n,\hat{\vx}^{(t+1)}\backslash\hat{x}_n^{(t+1)}, \hat{\boldsymbol\lambda}^{(t+1)}, \hat{\boldsymbol\theta}^{(t+1)}|\vy)\propto \exp\Delta^{(t+1)}_{x_n}\\
&\quad\quad\quad\quad=\exp\left(\Delta^{(t+1)}_{\Omega_n\rightarrow x_n}+\textstyle\sum_{\Phi_m\in\Gamma(x_n)}\Delta^{(t+1)}_{\Phi_m\rightarrow x_n}\right)
\end{split}\\
\label{eq:pm_lambda_ms}
\begin{split}
&p(\hat{\vx}^{(t+1)}, \lambda_l,\hat{\boldsymbol\lambda}^{(t+1)}\backslash\hat{\lambda}_l^{(t+1)}|\vy)\propto \exp\Delta^{(t+1)}_{\lambda_l}\\
&\quad\quad\quad\quad=\exp\left(\log p(\lambda_l)+\textstyle\sum_{\Omega_n\in\Gamma(\lambda_l)}\Delta^{(t+1)}_{\Omega_n\rightarrow\lambda_l}\right)
\end{split}\\
\label{eq:pm_theta_ms}
\begin{split}
&p(\hat{\vx}^{(t+1)},\theta_k,\hat{\boldsymbol\theta}^{(t+1)}\backslash\hat{\theta}_k^{(t+1)}|\vy)\propto \exp\Delta^{(t+1)}_{\theta_k}\\
&\quad\quad\quad\quad=\exp\left(\log p(\theta_k)+\textstyle\sum_{\Phi_m\in\Gamma(\theta_k)}\Delta^{(t+1)}_{\Phi_m\rightarrow\theta_k}\right)\,,
\end{split}
\end{align}
\end{subequations}
where $\hat{\vx},\hat{\boldsymbol\lambda},\hat{\boldsymbol\theta}$ are the maximizing values computed from (\ref{eq:ms_fv_phi_x},\ref{eq:ms_fv_omega_x},\ref{eq:ms_fv_omega_lambda},\ref{eq:ms_fv_phi_theta}) accordingly. The ``joint'' MAP estimates of the signal $\vx$ and the parameters $\boldsymbol\lambda, \boldsymbol\theta$ are then:
\begin{subequations}
\begin{align}
\hat{x}_n&=\arg\max_{x_n}\,p(x_n,\hat{\vx}^{(t+1)}\backslash\hat{x}_n^{(t+1)}, \hat{\boldsymbol\lambda}^{(t+1)}, \hat{\boldsymbol\theta}^{(t+1)}|\vy)\\
\hat{\lambda}_l&=\arg\max_{\lambda_l}\,p(\hat{\vx}^{(t+1)}, \lambda_l,\hat{\boldsymbol\lambda}^{(t+1)}\backslash\hat{\lambda}_l^{(t+1)}|\vy)\\
\hat{\theta}_k&=\arg\max_{\theta_k}\,p(\hat{\vx}^{(t+1)},\theta_k,\hat{\boldsymbol\theta}^{(t+1)}\backslash\hat{\theta}_k^{(t+1)}|\vy)\,.
\end{align}
\end{subequations}

\subsection{Parameter Estimation}
\label{sec:pe}
The priors $p(\lambda_l), p(\theta_k)$ on the parameters are usually chosen to be some simple distributions. If we do not have any knowledge on how $\boldsymbol\lambda, \boldsymbol\theta$ are distributed, we can fairly assume a uniform prior and treat $p(\lambda_l), p(\theta_k)$ as constants. Since $\lambda_l,\theta_k$ are treated as random variables in the PE-GAMP framework, they will be jointly estimated along with the signal $\vx$ in the message-updating process.

\subsubsection{Sum-product Message Passing}
Take $\lambda_l$ for example, in the PE-GAMP, we propose to approximate the underlying distribution $p^{(t+1)}_{\Omega_n\leftarrow \lambda_l}(\lambda_l|\vy)\propto \exp(\Delta_{\Omega_n\leftarrow \lambda_l}^{(t+1)})$ using Dirac delta function:
\begin{align}
p^{(t+1)}_{\Omega_n\leftarrow\lambda_l}(\lambda_l|\vy)\approx\delta\left(\lambda_l-\hat{\lambda}_{\Omega_n\leftarrow\lambda_l}^{(t+1)}\right)\,,
\end{align}
where $\delta(\cdot)$ is the Dirac delta function, $\hat{\lambda}_{\Omega_n\leftarrow\lambda_l}^{(t+1)}$ can be computed using either the MAP or MMSE estimation:
\begin{subequations}
\label{eq:sp_pe_map_mmse}
\begin{align}
&\textrm{MAP estimation of $\lambda_l$:}\,\,\hat{\lambda}_{\Omega_n\leftarrow\lambda_l}^{(t+1)}\coloneqq \arg\max_{\lambda_l}\Delta_{\Omega_n\leftarrow\lambda_l}^{(t+1)}\\
&\textrm{MMSE estimation of $\lambda_l$:}\,\,\hat{\lambda}_{\Omega_n\leftarrow\lambda_l}^{(t+1)}\coloneqq\mathbb{E}[\lambda_l|\Delta_{\Omega_n\leftarrow\lambda_l}^{(t+1)}]\,,
\end{align}
\end{subequations}
where $\mathbb{E}[\lambda_l|\Delta_{\Omega_n\leftarrow\lambda_l}^{(t+1)}]$ is the mean of the distribution $\frac{1}{\mathcal{C}}\exp(\Delta_{\Omega_n\leftarrow\lambda_l}^{(t+1)})$, $\mathcal{C}$ is a normalizing constant.

The formulations for the rest parameters can be derived similarly. The reason behind the choice of Dirac delta approximation of $p_{\Omega_n\leftarrow\lambda_l}^{(t+1)}(\lambda_l|\vy)$ is its simplicity, it amounts to the scalar MAP or MMSE estimation of $\lambda_l$ from the posterior distribution $p^{(t+1)}_{\Omega_n\leftarrow \lambda_l}(\lambda_l|\vy)$. Other approximations often make it quite difficult to compute the message $\Delta_{\Omega_n\rightarrow\lambda_l}^{(t+1)}$ in (\ref{eq:sp_fv_omega_lambda}) due to the lack of closed-form solutions.

The updated messages from the factor nodes to the variable nodes are then:
\begin{subequations}
\label{eq:sp_simplified_messages}
\begin{align}
\begin{split}
\Delta^{(t+1)}_{\Phi_m\rightarrow x_n}&=\textrm{const}+\log\int_{\vx\backslash x_n}\left[\Phi_m\left(y_m, \vx, \hat{\boldsymbol\theta}_{\Phi_m}^{(t)}\right)\right.\\
&\quad\quad\left.\times\exp\left(\textstyle\sum_{j\neq n}\Delta^{(t)}_{\Phi_m\leftarrow x_j}\right)\right]
\end{split}\\
\Delta^{(t+1)}_{\Omega_n\rightarrow x_n}&=\textrm{const}+\log\Omega_n(x_n,\hat{\boldsymbol\lambda}_{\Omega_n}^{(t)})\\
\label{eq:delta_omega_lambda_simplified}
\begin{split}
\Delta^{(t+1)}_{\Omega_n\rightarrow \lambda_l}&=\textrm{const}+\log\int_{x_n}\left[\Omega_n\left(x_n,\lambda_l,\hat{\boldsymbol\lambda}_{\Omega_n}^{(t)}\backslash\hat{\lambda}_{\Omega_n\leftarrow\lambda_l}^{(t)}\right)\right.\\
&\quad\quad\left.\times\exp\left(\Delta^{(t+1)}_{\Omega_n\leftarrow x_n}\right)\right]
\end{split}\\
\label{eq:delta_phi_theta_simplified}
\begin{split}
\Delta^{(t+1)}_{\Phi_m\rightarrow\theta_k}&=\textrm{const}+\log\int_{\vx}\left[\Phi_m\left(y_m, \vx, \theta_k, \hat{\boldsymbol\theta}_{\Phi_m}^{(t)}\backslash\hat{\theta}_{\Phi_m\leftarrow\theta_k}^{(t)}\right)\right.\\
&\quad\quad\left.\times\exp\left(\textstyle\sum_j\Delta^{(t)}_{\Phi_m\leftarrow x_j}\right)\right]\,,
\end{split}
\end{align}
\end{subequations}
where $\hat{\boldsymbol\lambda}_{\Omega_n}^{(t)}$, $\hat{\boldsymbol\theta}_{\Phi_m}^{(t)}$ are scalar estimates from the previous $t$-th iteration at nodes $\Omega_n$ and $\Phi_m$ respectively. 
\begin{subequations}
\label{eq:pe_gamp_parameters_estimate}
\begin{align}
&\hat{\boldsymbol\lambda}_{\Omega_n}^{(t)}=\left\{\left.\hat{\lambda}_{\Omega_n\leftarrow\lambda_u}^{(t)}\right|u=1,\cdots,L\right\}\\
&\hat{\boldsymbol\theta}_{\Phi_m}^{(t)}=\left\{\left.\hat{\theta}_{\Phi_m\leftarrow\theta_v}^{(t)}\right|v=1,\cdots,K\right\}
\end{align}
\end{subequations}

\subsubsection{Max-sum Message Passing}
Take $\lambda_l$ for example, a straightforward way to solve the problems in (\ref{eq:ms_fv_omega_x}, \ref{eq:ms_fv_omega_lambda}) is to iteratively maximize each varaible in $\{x_n,\boldsymbol\lambda\backslash\lambda_l\}$ while keeping the rest fixed until convergence. However, it is inefficient and quite unnecessary. In practice one iteration would suffice. Hence we propose to use the following solutions as the approximate maximizing parameters:
\begin{align}
\label{eq:ms_pe_map}
\begin{split}
\hat{\lambda}_{\Omega_n\leftarrow\lambda_l}^{(t+1)}=\arg\max_{\lambda_l}\,&\log\Omega_n\left(\hat{x}_n^{(t)},\lambda_l,\hat{\boldsymbol\lambda}_{\Omega_n}^{(t)}\backslash\hat{\lambda}_{\Omega_n\leftarrow\lambda_l}^{(t)}\right)\\
&\quad + \Delta_{\Omega_n\leftarrow\lambda_l}^{(t)}.
\end{split}
\end{align}
The updated messages from the factor nodes to the variable nodes can be obtained by replacing ``$\int$'' in (\ref{eq:sp_simplified_messages}) with ``$\max$'' like before.

\subsubsection{The PE-GAMP Algorithm}
For the rest of the paper, parameter estimation operations like those in (\ref{eq:sp_pe_map_mmse}, \ref{eq:ms_pe_map}) will be abbreviated by the two functions $f_{\Omega_n\leftarrow\lambda_l}(\cdot), f_{\Phi_m\leftarrow\theta_k}(\cdot)$. 
\begin{align}
\hat{\lambda}_{\Omega_n\leftarrow\lambda_l}^{(t+1)}=f_{\Omega_n\leftarrow\lambda_l}(\cdot)\quad\textrm{and}\quad \hat{\theta}_{\Phi_m\leftarrow\theta_k}^{(t+1)}=f_{\Phi_m\leftarrow\theta_k}(\cdot)\,.
\end{align}
They are different from the input and output channels estimation functions $g_\textrm{in}(\cdot), g_\textrm{out}(\cdot)$ defined in \cite{GAMP11}. 

The proposed GAMP algorithm with built-in parameter estimation (PE-GAMP) can be summarized in Algorithm \ref{alg:pe_gamp}, where $\vq_\Phi\in\mathbb{R}^M, \vr_\Omega\in\mathbb{R}^N$ can be viewed as some new random variables created inside the original GAMP framework \cite{GAMP11}, and $\boldsymbol\tau_\Phi^q\in\mathbb{R}^M, \boldsymbol\tau_\Phi^s\in\mathbb{R}^M, \boldsymbol\tau_\Omega^r\in\mathbb{R}^N$ are their corresponding variances.  As is done in \cite{GAMP11}, further simplification will be made by replacing the variance vectors  with scalars when performing asymptotic analysis of Algorithm \ref{alg:pe_gamp}:
\begin{align}
\boldsymbol\tau_\Phi^q, \boldsymbol\tau_\Omega^r \xRightarrow[]{\textrm{Replace}} \tau_\Phi^q, \tau_\Omega^r\,.
\end{align}
\begin{algorithm}
\caption{GAMP with Built-in Parameter Estimation (PE-GAMP)}\label{alg:pe_gamp}
\begin{algorithmic}[1]
\Require The matrix $\vA\in\mathbb{R}^{M\times N}$; the observation $\vy\in\mathbb{R}^M$; the input and output channels estimation functions $g_\textrm{in}(\cdot), g_\textrm{out}(\cdot)$; the parameter estimation functions $f_{\Omega_n}(\cdot), f_{\Phi_m}(\cdot)$.
\State Set $\vs^{(-1)}=0$ and initialize $\hat{\vx}^{(0)}, \boldsymbol\tau_{\Omega}^x(0), \hat{\boldsymbol\lambda}_{\Omega_n}^{(0)}, \hat{\boldsymbol\theta}_{\Phi_m}^{(0)}$.
\For{$t=\{0,1,\cdots\}$}
	\State Output channel \emph{linear} update: For each $m=1,\cdots,M$
	\begin{subequations}
	\begin{align}
	&\tau_{\Phi_m}^q(t)=\textstyle\sum_n|A_{mn}|^2\tau_{\Omega_n}^x(t)\\
	&q_{\Phi_m}^{(t)}=\textstyle\sum_nA_{mn}\hat{x}_n^{(t)}-\tau_{\Phi_m}^q(t)s_m^{(t-1)}\\
	&\hat{z}_m^{(t)}=\textstyle\sum_nA_{mn}\hat{x}_n^{(t)}\,.
	\end{align}
	\end{subequations}
	\State Output channel \emph{nonlinear} update: For each $m=1,\cdots,M$
	\begin{subequations}
	\begin{align}
	&s_{\Phi_m}^{(t)}=g_\textrm{out}\left(t, q_{\Phi_m}^{(t)}, \tau_{\Phi_m}^q(t), y_m, \hat{\boldsymbol\theta}_{\Phi_m}^{(t)}\right)\\
	&\tau_{\Phi_m}^s(t)=-\frac{\partial}{\partial q}g_\textrm{out}\left(t, q_{\Phi_m}^{(t)}, \tau_{\Phi_m}^q(t), y_m, \hat{\boldsymbol\theta}_{\Phi_m}^{(t)}\right)\,.
	\end{align}
	\end{subequations}
	\State Input channel \emph{linear} update: For each $n=1,\cdots,N$
	\begin{subequations}
	\begin{align}
	&\tau_{\Omega_n}^r(t)=\left[\textstyle\sum_m|A_{mn}|^2\tau_{\Phi_m}^s(t)\right]^{-1}\\
	&r_{\Omega_n}^{(t)}=x_n^{(t)}+\tau_{\Omega_n}^r(t)\textstyle\sum_mA_{mn}s_m^{(t)}\,.
	\end{align}
	\end{subequations}
	\State Input \emph{nonlinear} update: For each $n=1,\cdots,N$
	\begin{subequations}
	\begin{align}
	\label{eq:pegamp_output_x}
	&\hat{x}_n^{(t+1)} = g_\textrm{in}\left(t, r_{\Omega_n}^{(t)}, \tau_{\Omega_n}^r(t), \hat{\boldsymbol\lambda}_{\Omega_n}^{(t)}\right)\\
	&\tau_{\Omega_n}^x(t+1)=\tau_{\Omega_n}^r(t)\frac{\partial}{\partial r} g_\textrm{in}\left(t, r_{\Omega_n}^{(t)}, \tau_{\Omega_n}^r(t), \hat{\boldsymbol\lambda}_{\Omega_n}^{(t)}\right)\,.
	\end{align}
	\end{subequations}
	\State \emph{Sum-product} message passing parameters update: For each $k=1,\cdots,K$ and $l=1,\cdots,L$.
	\begin{subequations}
	\label{eq:sp_parameters_update}	
	\begin{align}
	&\hat{\lambda}_{\Omega_n\leftarrow\lambda_l}^{(t+1)}=f_{\Omega_n\leftarrow\lambda_l}\left(t, \vr_{\Omega}^{(t)}, \boldsymbol\tau_{\Omega}^r(t), \lambda_l, \hat{\boldsymbol\lambda}_{\Omega_n}^{(t)}\backslash\hat{\lambda}_{\Omega_n\leftarrow\lambda_l}^{(t)}\right)\\
	&\hat{\theta}_{\Phi_m\leftarrow\theta_k}^{(t+1)}=f_{\Phi_m\leftarrow\theta_k}\left(t, \vq_{\Phi}^{(t)}, \vy, \boldsymbol\tau_{\Phi}^q(t), \theta_k, \hat{\boldsymbol\theta}_{\Phi_m}^{(t)}\backslash\hat{\theta}_{\Phi_m\leftarrow\theta_k}^{(t)}\right)\,.
	\end{align}
	\end{subequations}
	\State \emph{Max-sum} message passing parameters update: For each $k=1,\cdots,K$ and $l=1,\cdots,L$.
	\begin{subequations}
	\label{eq:sp_parameters_update}	
	\begin{align}
	&\hat{\lambda}_{\Omega_n\leftarrow\lambda_l}^{(t+1)}=f_{\Omega_n\leftarrow\lambda_l}\left(t, \hat{x}_n^{(t)}, \vr_{\Omega}^{(t)}, \boldsymbol\tau_{\Omega}^r(t), \lambda_l, \hat{\boldsymbol\lambda}_{\Omega_n}^{(t)}\backslash\hat{\lambda}_{\Omega_n\leftarrow\lambda_l}^{(t)}\right)\\
	&\hat{\theta}_{\Phi_m\leftarrow\theta_k}^{(t+1)}=f_{\Phi_m\leftarrow\theta_k}\left(t, \hat{\vz}^{(t)}, \vq_{\Phi}^{(t)}, \vy, \boldsymbol\tau_{\Phi}^q(t), \theta_k, \hat{\boldsymbol\theta}_{\Phi_m}^{(t)}\backslash\hat{\theta}_{\Phi_m\leftarrow\theta_k}^{(t)}\right)\,.
	\end{align}
	\end{subequations}
	\If {$\hat{\vx}^{(t+1)}$ reaches convergence}
		\State $\hat{\vx}=\hat{\vx}^{(t+1)}$;
		\State break;
	\EndIf
\EndFor
\State\Return Output $\hat{\vx}$;
\end{algorithmic}
\end{algorithm}

For the \emph{sum-product} message passing, PE-GAMP naturally produces MMSE estimation of $\vx$ in (\ref{eq:pegamp_output_x}). After the convergence is reached, we can also compute the MAP estimation of $\vx$ using $p(x_n|\vy)$: $\hat{x}_n = \arg\max_{x_n}p(x_n|\vy)$. For the \emph{max-sum} message passing, PE-GAMP naturally produces the ``joint'' MAP estimation of $\vx$ in (\ref{eq:pegamp_output_x}). However, there isn't any meaningful MMSE estimation of $\vx$ in this case.

\section{State Evolution Analysis of PE-GAMP}
\subsection{Review of the GAMP State Evolution Analysis}
We first introduce the definitions as well as assumptions used in the state evolution (SE) analysis \cite{GAMP11} that studies the empirical convergence behavior of the variables in the large system limit. It is a minor modification of the work from \cite{AMP11}.
\smallskip
\begin{definition}
A function $g(\cdot):\mathbb{R}^r\rightarrow\mathbb{R}^s$ is \emph{pseudo-Lipschitz} of order $k>1$, if there exists an $L>0$ such that $\forall \vx,\vy\in\mathbb{R}^r$,
\begin{align}
\|g(\vx)-g(\vy)\|\leq L(1+\|\vx\|^{k-1}+\|\vy\|^{k-1})\|\vx-\vy\|\,.
\end{align}
\end{definition}
\smallskip
\begin{definition}
Suppose $\{\vv^{[N]}\in\mathbb{R}^{sl_N}, N=1,2,\cdots\}$ is a sequence of vectors, and each $\vv^{[N]}$ contains $l_N$ blocks of vector components $\{\vv_i^{[N]}\in\mathbb{R}^s, i=1,\cdots,l_N\}$. The components of $\vv^{[N]}$ \emph{empirically converges with bounded moments of order k} to a random vector $\vv\in\mathbb{R}^s$ as $N\rightarrow\infty$ if: For all pesudo-Lipschitz continuous functions $g(\cdot)$ of order $k$,
\begin{align}
\lim_{N\rightarrow\infty}\frac{1}{l_N}\sum_{i=1}^{l_N}g\left(\vv_i^{[N]}\right)=\mathbb{E}\left[g(\vv)\right]<\infty\,.
\end{align}
When the nature of convergence is clear, it can be simply written as follows:
\begin{align}
\lim_{N\rightarrow\infty}\vv_i^{[N]}\stackrel{\normalfont\textrm{PL$(k)$}}{=}\vv\,.
\end{align}
\end{definition}
Based on the above pseudo-Lipschitz continuity and empirical convergence definitions, GAMP also makes the following assumptions about the estimation of $\vx\in\mathbb{R}^N$ \cite{AMP11,GAMP11}.
\smallskip
\begin{assumption}
\label{aspt_gamp}
The GAMP solves a series of estimation problems indexed by the input signal dimension $N$:
\begin{enumerate}[label={\alph*)}, nolistsep]
\item The output dimension $M$ is deterministic and scales linearly with the input dimension $N$: $\lim_{N\rightarrow\infty}\frac{N}{M}=\beta$ for some $\beta>0$.
\item The matrix $\vA\in\mathbb{R}^{M\times N}$ has i.i.d Gaussian entries $A_{ij}\sim \mathcal{N}(0,\frac{1}{M})$.
\item The components of initial condition $\hat{\vx}^{(0)},\tau_\Omega^x(0)$ and the input signal $\vx$ empirically converge with bounded moments of order $2k-2$ as follows:
\begin{subequations}
\begin{align}
&\lim_{N\rightarrow\infty}(\hat{x}_n^{(0)}, x_n)\stackrel{\normalfont\textrm{PL$(2k-2)$}}{=}(\hat{\mathcal{X}}^{(0)},\mathcal{X})\\
&\lim_{N\rightarrow\infty}\tau_{\Omega_n}^x(0)=\overline{\tau}_\Omega^x(0)\,.
\end{align}
\end{subequations}
\item The output vector $\vy\in\mathbb{R}^M$ depends on the transform output $\vz=\vA\vx\in\mathbb{R}^M$ and the noise vector $\vw\in\mathbb{R}^M$ through some function $g(\cdot)$. For $\forall m=1,\cdots,M$,
\begin{align}
y_m=g(z_m, w_m)\,.
\end{align}
$w_m$ empirically converges with bounded moments of order $2k-2$ to some random variable $\mathcal{W}\in\mathbb{R}$ with distribution $p(w)$. The conditional distribution of $\mathcal{Y}$ given $\mathcal{Z}$ is given by $p(y|z)$.
\item The channel estimation functions $g_\textrm{in}(\cdot)$, $g_\textrm{out}(\cdot)$ and their partial derivatives with respect to $r, q, z$ exist almost everywhere and are pseudo-Lipschitz continuous of order $k$. 
\end{enumerate}
\end{assumption}
The SE equations of the GAMP describe the limiting behavior of the following scalar random variables and scalar variances as $N\rightarrow\infty$:
\begin{subequations}
\label{eq:rv_gamp}
\begin{align}
\label{eq:rv_gamp_a}
&\boldsymbol\psi_\textrm{in}\coloneqq\{(x_n, \hat{x}_n^{(t+1)}, r_{\Omega_n}^{(t)}),\, n=1,\cdots,N\}\\
\label{eq:rv_gamp_b}
&\boldsymbol\psi_\textrm{out}\coloneqq\{(z_m, \hat{z}_m^{(t)}, y_m, q_{\Phi_m}^{(t)}),\, m=1,\cdots,M\}\\
\label{eq:rv_gamp_c}
&\boldsymbol\psi_\tau\coloneqq(\tau_\Phi^q, \tau_\Omega^r)\,.
\end{align}
\end{subequations}
\cite{GAMP11} showed that (\ref{eq:rv_gamp_a}-\ref{eq:rv_gamp_b}) empirically converge with bounded moments of order $k$ to the following random vectors:
\begin{subequations}
\begin{align}
&\lim_{N\rightarrow\infty}\boldsymbol\psi_\textrm{in}\stackrel{\normalfont\textrm{PL$(k)$}}{=}\overline{\boldsymbol\psi}_\textrm{in}\coloneqq(\mathcal{X}, \hat{\mathcal{X}}^{(t+1)}, \mathcal{R}_{\Omega}^{(t)})\\
&\lim_{N\rightarrow\infty}\boldsymbol\psi_\textrm{out}\stackrel{\normalfont\textrm{PL$(k)$}}{=}\overline{\boldsymbol\psi}_\textrm{out}\coloneqq(\mathcal{Z}, \hat{\mathcal{Z}}^{(t)}, \mathcal{Y}, \mathcal{Q}_{\Phi}^{(t)})\,,
\end{align}
\end{subequations}
where $\mathcal{R}_{\Omega}^{(t)}, \mathcal{Z}, \mathcal{Q}_{\Phi}^{(t)}$ are as follows for some computed $\alpha^r\in\mathbb{R}$, $\xi^r\in\mathbb{R}$, $\vK^q\in\mathbb{R}^{2\times2}$:
\begin{subequations}
\begin{align}
&\mathcal{R}_{\Omega}^{(t)}=\alpha^r\mathcal{X}+\mathcal{V}, \quad\quad \mathcal{V}\sim\mathcal{N}(0,\xi^r)\\
&(\mathcal{Z}, \mathcal{Q}_{\Phi}^{(t)})\sim\mathcal{N}(0,\vK^q)\,.
\end{align}
\end{subequations}
Additionally, for $\boldsymbol\psi_\tau$, the following convergence holds:
\begin{align}
&\lim_{N\rightarrow\infty}\boldsymbol\psi_\tau=\overline{\boldsymbol\psi}_\tau\coloneqq(\overline{\tau}_\Phi^q, \overline{\tau}_\Omega^r)\,.
\end{align}

\begin{algorithm}
\caption{PE-GAMP State Evolution}\label{alg:pe_gamp_se}
\begin{algorithmic}[1]
\Require The matrix $\vA\in\mathbb{R}^{M\times N}$; the observation $\vy\in\mathbb{R}^M$; the input and output channels estimation functions $g_\textrm{in}(\cdot), g_\textrm{out}(\cdot)$; the parameter estimation functions $f_{\theta_k}(\cdot), f_{\lambda_l}(\cdot)$.
\State Initialize $\overline{\tau}_\Omega^x(0), \overline{\boldsymbol\theta}^{(0)}, \overline{\boldsymbol\lambda}^{(0)}$ and set
\begin{align}
\vK^x(0)=\textrm{cov}\left(\mathcal{X}, \hat{\mathcal{X}}^{(0)}\right)\,.
\end{align}
\For{$t=\{0,1,\cdots\}$}
	\State Output channel update:
	\begin{subequations}
	\begin{align}
	&\overline{\tau}_\Phi^q(t)=\beta\overline{\tau}_\Omega^x(t),\quad\quad \vK^q(t)=\beta\vK^x(t)\\
	&\overline{\tau}_\Omega^r(t)=-\mathbb{E}^{-1}\left[\frac{\partial}{\partial q}g_\textrm{out}\left(t, \mathcal{Q}^{(t)}_\Phi, \overline{\tau}_\Phi^q(t), \mathcal{Y}, \overline{\boldsymbol\theta}_{\Phi_m}^{(t)}\right)\right]\\
	&\xi^r(t)=\left(\overline{\tau}_\Omega^r(t)\right)^2\mathbb{E}\left[g_\textrm{out}\left(t, \mathcal{Q}^{(t)}_\Phi, \overline{\tau}_\Phi^q(t), \mathcal{Y}, \overline{\boldsymbol\theta}_{\Phi_m}^{(t)}\right)\right]\\
	&\alpha^r(t) = \overline{\tau}_\Omega^r(t)\mathbb{E}\left[\frac{\partial}{\partial z}g_\textrm{out}\left(t, \mathcal{Q}^{(t)}_\Phi,  \overline{\tau}_\Phi^q(t), g(\mathcal{Z},\mathcal{W}), \overline{\boldsymbol\theta}_{\Phi_m}^{(t)}\right)\right]
	\end{align}
	\end{subequations}
	where the expectations are over the random variables $\mathcal{Z},\mathcal{Q}_\Phi^{(t)}, \mathcal{W}, \mathcal{Y}$.
	\State Input channel update:
	\begin{subequations}
	\begin{align}
	&\hat{\mathcal{X}}^{(t+1)} = g_\textrm{in}\left(t, \mathcal{R}^{(t)}_\Omega, \overline{\tau}_\Omega^r(t), \overline{\boldsymbol\lambda}_{\Omega_n}^{(t)}\right)\\
	&\overline{\tau}_\Omega^x(t+1)=\overline{\tau}_\Omega^r(t)\mathbb{E}\left[\frac{\partial}{\partial r}g_\textrm{in}\left(t, \mathcal{R}^{(t)}_\Omega, \overline{\tau}_\Omega^r(t), \overline{\boldsymbol\lambda}_{\Omega_n}^{(t)}\right)\right]\\
	&\vK^x(t+1)=\textrm{cov}\left(\mathcal{X},\hat{\mathcal{X}}^{(t+1)}\right)\,,
	\end{align}
	\end{subequations}
	where the expectation is over the random variables $\mathcal{X}, \mathcal{R}^{(t)}_\Omega$.
	\State \emph{Sum-product} message passing parameters update: For each $k=1,\cdots,K$ and $l=1,\cdots,L$
	\begin{subequations}
	\label{eq:sp_pe_state_evolution}
	\begin{align}
	&\overline{\lambda}_{\Omega_n\leftarrow\lambda_l}^{(t+1)}=f_{\Omega_n\leftarrow\lambda_l}\left(t,\mathcal{R}^{(t)}_\Omega, \overline{\tau}_\Omega^r(t), \lambda_l, \overline{\boldsymbol\lambda}_{\Omega_n}^{(t)}\backslash\overline{\lambda}_{\Omega_n\leftarrow\lambda_l}^{(t)}\right)\\
	&\overline{\theta}_{\Phi_m\leftarrow\theta_k}^{(t+1)}=f_{\Phi_m\leftarrow\theta_k}\left(t,\mathcal{Q}^{(t)}_\Phi,\mathcal{Y}, \overline{\tau}_\Phi^q(t), \theta_k, \overline{\boldsymbol\theta}_{\Phi_m}^{(t)}\backslash\overline{\theta}_{\Phi_m\leftarrow\theta_k}^{(t)}\right)\,.
	\end{align}
	\end{subequations}
	\State \emph{Max-sum} message passing parameters update: For each $k=1,\cdots,K$ and $l=1,\cdots,L$
	\begin{subequations}
	\label{eq:ms_pe_state_evolution}
	\begin{align}
	&\overline{\lambda}_{\Omega_n\leftarrow\lambda_l}^{(t+1)}=f_{\Omega_n\leftarrow\lambda_l}\left(t,\hat{\mathcal{X}}^{(t)},\mathcal{R}^{(t)}_\Omega, \overline{\tau}_\Omega^r(t), \lambda_l, \overline{\boldsymbol\lambda}_{\Omega_n}^{(t)}\backslash\overline{\lambda}_{\Omega_n\leftarrow\lambda_l}^{(t)}\right)\\
	&\overline{\theta}_{\Phi_m\leftarrow\theta_k}^{(t+1)}=f_{\Phi_m\leftarrow\theta_k}\left(t,\hat{\mathcal{Z}}^{(t)},\mathcal{Q}^{(t)}_\Phi,\mathcal{Y}, \overline{\tau}_\Phi^q(t), \theta_k, \overline{\boldsymbol\theta}_{\Phi_m}^{(t)}\backslash\overline{\theta}_{\Phi_m\leftarrow\theta_k}^{(t)}\right)
	\end{align}
	\end{subequations}
	\If {$\hat{\mathcal{X}}^{(t+1)}$ reaches convergence}
		\State $\hat{\mathcal{X}}=\hat{\mathcal{X}}^{(t+1)}$;
		\State break;
	\EndIf
\EndFor
\State\Return Output $\hat{\mathcal{X}}$;
\end{algorithmic}
\end{algorithm}

\subsection{PE-GAMP State Evolution Analysis}
The SE equations of the proposed PE-GAMP are given in Algorithm \ref{alg:pe_gamp_se}. In addition to (\ref{eq:rv_gamp_a}-\ref{eq:rv_gamp_c}), the state evolution (SE) analysis of PE-GAMP will study the limiting behavior of $\hat{\boldsymbol\lambda}_{\Omega_n}^{(t+1)},\hat{\boldsymbol\theta}_{\Phi_m}^{(t+1)}$ for each $n=1,\cdots,N$ and $m=1,\cdots,M$.

Eventually we would like to show that they empirically converge to the following random vectors for fixed $t$ as $N\rightarrow\infty$:
\begin{subequations}
\begin{align}
&\overline{\boldsymbol\lambda}_{\Omega_n}^{(t+1)}=\left\{\left.\overline{\lambda}_{\Omega_n\leftarrow\lambda_l}^{(t+1)}\right|l=1,\cdots,L\right\}\\
&\overline{\boldsymbol\theta}_{\Phi_m}^{(t+1)}=\left\{\left.\overline{\theta}_{\Phi_m\leftarrow\theta_k}^{(t+1)}\right|k=1,\cdots,K\right\}
\end{align}
\end{subequations}
To simplify notations, we assume the following for the \emph{sum-product} message passing:
\begin{subequations}
\begin{align}
&h_{\Omega_n\leftarrow\lambda_l}^{\Omega_j}(\cdot)=\Delta^{(t)}_{\Omega_j\rightarrow\lambda_l}+\frac{1}{N-1}\log p(\lambda_l)\\ 
&h_{\Phi_m\leftarrow\theta_k}^{\Phi_i}(\cdot)=\Delta^{(t)}_{\Phi_i\rightarrow\theta_k}+\frac{1}{M-1}\log p(\theta_k)
\end{align}
\end{subequations}
For \emph{max-sum} message passing, we assume:
\begin{subequations}
\begin{align}
\begin{split}
&h_{\Omega_n\leftarrow\lambda_l}^{\Omega_j}(\cdot)=\Delta^{(t)}_{\Omega_j\rightarrow\lambda_l}\\
&+\frac{1}{N-1}\left(\log p(\lambda_l)+\log\Omega_n\left(\hat{x}_n^{(t)},\lambda_l,\hat{\boldsymbol\lambda}_{\Omega_n}^{(t)}\backslash\hat{\lambda}_{\Omega_n\leftarrow\lambda_l}^{(t)}\right)\right)
\end{split}\\
\begin{split}
&h_{\Phi_m\leftarrow\theta_k}^{\Phi_i}(\cdot)=\Delta^{(t)}_{\Phi_i\rightarrow\theta_k}\\
&+\frac{1}{M-1}\left(\log p(\theta_k)+\log\Phi_m\left(y_m, \hat{\vx}^{(t)},\theta_k,\hat{\boldsymbol\theta}_{\Phi_m}^{(t)}\backslash\hat{\theta}_{\Phi_m\leftarrow\theta_k}^{(t)}\right)\right)
\end{split}
\end{align}
\end{subequations}

Since the parameter estimation of the \emph{max-sum} message passing and the MAP parameter estimation of the \emph{sum-product} message passing basically have the same form given in (\ref{eq:map_lambda_theta_h}), their state evolution analysis can be derived similarly. For the sake of conciseness, we will only give the empirical convergence proofs for the MAP and MMSE parameter estimations of the \emph{sum-product} message passing.

\subsubsection{MAP Parameter Estimation State Evolution}
We can also write the estimation functions as follows:
\begin{subequations}
\label{eq:map_lambda_theta_h}
\begin{align}
\label{eq:map_lambda_h}
\hat{\lambda}_{\Omega_n\leftarrow\lambda_l}^{(t+1)}&=\arg\max_{\lambda_l}\frac{1}{N-1}\textstyle\sum_{j\neq n}h_{\Omega_n\leftarrow\lambda_l}^{\Omega_j}(\cdot)\\
\label{eq:map_theta_h}
\hat{\theta}_{\Phi_m\leftarrow\theta_k}^{(t+1)}&=\arg\max_{\theta_k}\frac{1}{M-1}\textstyle\sum_{i\neq m}h_{\Phi_m\leftarrow\theta_k}^{\Phi_i}(\cdot)
\end{align}
\end{subequations}
In the large system limit $N\rightarrow\infty$, the state evolution equations (\ref{eq:sp_pe_state_evolution}) of the parameters update step in \emph{sum-product} message passing can then be written as:
\begin{subequations}
\begin{align}
\begin{split}
&\overline{\lambda}_{\Omega_n\leftarrow\lambda_l}^{(t+1)}=f_{\Omega_n\leftarrow\lambda_l}(\cdot)\\
&=\arg\max_{\lambda_l}\mathbb{E}\left[h_{\Omega_n\leftarrow\lambda_l}^{\Omega_j}\left(t,\mathcal{R}^{(t)}_\Omega, \overline{\tau}_\Omega^r(t), \lambda_l, \overline{\boldsymbol\lambda}_{\Omega_n}^{(t)}\backslash\overline{\lambda}_{\Omega_n\leftarrow\lambda_l}^{(t)}\right)\right]
\end{split}\\
\begin{split}
&\overline{\theta}_{\Phi_m\leftarrow\theta_k}^{(t+1)}=f_{\Phi_m\leftarrow\theta_k}(\cdot)\\
&=\arg\max_{\theta_k}\mathbb{E}\left[h_{\Phi_m\leftarrow\theta_k}^{\Phi_i}\left(t,\mathcal{Q}^{(t)}_\Phi,\mathcal{Y}, \overline{\tau}_\Phi^q(t), \theta_k, \overline{\boldsymbol\theta}_{\Phi_m}^{(t)}\backslash\overline{\theta}_{\Phi_m\leftarrow\theta_k}^{(t)}\right)\right]\,,
\end{split}
\end{align}
\end{subequations}
where the expectations are over the random variables $\mathcal{R}^{(t)}_\Omega$ and $\left\{\mathcal{Q}^{(t)}_\Phi,\mathcal{Y}\right\}$ respectively. 

Our proof of the convergence of the scalars in (\ref{eq:pe_gamp_parameters_estimate},\ref{eq:rv_gamp}) will make use of the Theorem \ref{thm:adaptive_gamp} from \cite{AMP_CPE14} in Appendix \ref{app:se_adaptive_gamp}. First, we give the following adapted assumptions for the MAP parameter estimation.
\smallskip
\begin{assumption}
\label{aspt_pe_gamp_map}
The priors on the parameters: $\{p(\boldsymbol\lambda),\boldsymbol\lambda\in\mathcal{U}_\lambda\}$, $\{p(\boldsymbol\theta),\boldsymbol\theta\in\mathcal{U}_\theta\}$ and the parameter estimation functions should satisfy:
\begin{enumerate}[label={\alph*)}, nolistsep]
\item The priors $p(\boldsymbol\lambda)<\infty\,,p(\boldsymbol\theta)<\infty$ are bounded, and the sets $\mathcal{U}_\lambda, \mathcal{U}_\theta$ are compact.
\item For the \emph{sum-product} message passing, the following estimations are well-defined, unique.
\begin{subequations}
\begin{align}
\label{eq:map_lambda_detail}
\begin{split}
&\lambda^*_{\Omega_n\leftarrow\lambda_l}=\\
&\arg\max_{\lambda_l\in\mathcal{U}_\lambda}\mathbb{E}\left[h_{\Omega_n\leftarrow\lambda_l}^{\Omega_j}(t, \mathcal{R}_\Omega^{(t)}, \tau_\Omega^r(t),\lambda_l, \hat{\boldsymbol\lambda}_{\Omega_n}^{(t)}\backslash\hat{\lambda}_{\Omega_n\leftarrow\lambda_l}^{(t)})\right]
\end{split}\\
\begin{split}
&\theta^*_{\Phi_m\leftarrow\theta_k}=\\
&\arg\max_{\theta_k\in\mathcal{U}_\theta}\mathbb{E}\left[h_{\Phi_m\leftarrow\theta_k}^{\Phi_i}(t, \mathcal{Q}_\Phi^{(t)},\mathcal{Y},\tau_\Phi^q(t),\theta_k,\hat{\boldsymbol\theta}_{\Phi_m}^{(t)}\backslash\hat{\theta}_{\Phi_m\leftarrow\theta_k}^{(t)})\right]\,,
\end{split}
\end{align}
\end{subequations}
where the expectations are with respect to $\mathcal{R}_\Omega^{(t)}$ and $\left\{\mathcal{Q}_\Phi^{(t)},\mathcal{Y}\right\}$.

\item $h_{\Omega_n\leftarrow\lambda_l}^{\Omega_j}(\cdot)$ is pseudo-Lipschitz continuous of order $2$ in $r_{\Omega_n}$, it is also continuous in $\lambda_l$ uniformly over $r_{\Omega_n}$ in the following sense: For every $\epsilon>0, \widetilde{\tau}_\Omega^r, \widetilde{\boldsymbol\lambda}\in\mathcal{U}_\lambda$, there exists an open neighborhood $\rho(\widetilde{\tau}_\Omega^r, \widetilde{\boldsymbol\lambda})$ of $(\widetilde{\tau}_\Omega^r, \widetilde{\boldsymbol\lambda}\in\mathcal{U}_\lambda)$, such that $\forall (\tau_\Omega^r, \boldsymbol\lambda)\in\rho(\widetilde{\tau}_\Omega^r, \widetilde{\boldsymbol\lambda})$ and all $r$,
\begin{align}
\left|h_{\Omega_n\leftarrow\lambda_l}^{\Omega_j}(t, r_{\Omega_n}, \tau_\Omega^r, \boldsymbol\lambda)-h_{\Omega_n\leftarrow\lambda_l}^{\Omega_j}(t, r_{\Omega_n}, \widetilde{\tau}_\Omega^r,\widetilde{\boldsymbol\lambda})\right|<\epsilon\,.
\end{align}
\item $h_{\Phi_m\leftarrow\theta_k}^{\Phi_i}(\cdot)$ is pseudo-Lipschitz continuous of order $2$ in $(q_{\Phi_m},y_m)$, it is also continuous in $\theta_k$ uniformly over $q_{\Phi_m}$ and $y_m$.
\end{enumerate}
\end{assumption}
\bigskip
\subsubsection{MMSE Parameter Estimation State Evolution}
For the MMSE parameter estimation, the estimation functions can be written as follows:
\begin{subequations}
\begin{align}
&\hat{\lambda}_{\Omega_n\leftarrow\lambda_l}^{(t+1)}=\int_{\lambda_l}\lambda_l\frac{\exp(\frac{1}{N-1}\textstyle\sum_{j\neq n}h_{\Omega_n\leftarrow\lambda_l}^{\Omega_j}(\cdot))}{\int_{\lambda_l}\exp(\frac{1}{N-1}\textstyle\sum_{j\neq n}h_{\Omega_n\leftarrow\lambda_l}^{\Omega_j}(\cdot))}\\
&\hat{\theta}_{\Phi_m\leftarrow\theta_k}^{(t+1)}=\int_{\theta_k}\theta_k\frac{\exp(\frac{1}{M-1}\textstyle\sum_{i\neq m}h_{\Phi_m\leftarrow\theta_k}^{\Phi_i}(\cdot))}{\int_{\theta_k}\exp(\frac{1}{M-1}\textstyle\sum_{i\neq m}h_{\Phi_m\leftarrow\theta_k}^{\Phi_i}(\cdot))}\,.
\end{align}
\end{subequations}
The state evolution equations (\ref{eq:sp_pe_state_evolution}) of the parameters update step in Algorithm \ref{alg:pe_gamp_se} can then be written as:
\begin{subequations}
\label{eq:mmse_paramters_state_evolution}
\begin{align}
&\overline{\lambda}_{\Omega_n\leftarrow\lambda_l}^{(t+1)}=f_{\Omega_n\leftarrow\lambda_l}(\cdot)=\int_{\lambda_l}\lambda_l\frac{\exp(\mathbb{E}\left[h_{\Omega_n\leftarrow\lambda_l}^{\Omega_j}(\cdot)\right])}{\int_{\lambda_l}\exp(\mathbb{E}\left[h_{\Omega_n\leftarrow\lambda_l}^{\Omega_j}(\cdot)\right])}\\
&\overline{\theta}_{\Phi_m\leftarrow\theta_k}^{(t+1)}=f_{\Phi_m\leftarrow\theta_k}(\cdot)=\int_{\theta_k}\theta_k\frac{\exp(\mathbb{E}\left[h_{\Phi_m\leftarrow\theta_k}^{\Phi_i}(\cdot)\right])}{\int_{\theta_k}\exp(\mathbb{E}\left[h_{\Phi_m\leftarrow\theta_k}^{\Phi_i}(\cdot)\right])}\,,
\end{align}
\end{subequations}
where the expectations are over the random variables $\mathcal{R}^{(t)}_\Omega$ and $\left\{\mathcal{Q}^{(t)}_\Phi,\mathcal{Y}\right\}$. To prove the convergence, we assume the following adapted assumptions for MMSE parameter estimation.
\smallskip
\begin{assumption}
\label{aspt_pe_gamp_mmse}
The priors on the parameters: $\{p(\boldsymbol\lambda),\boldsymbol\lambda\in\mathcal{U}_\lambda\}$, $\{p(\boldsymbol\theta),\boldsymbol\theta\in\mathcal{U}_\theta\}$ and the parameter estimation functions should satisfy:
\begin{enumerate}[label={\alph*)}, nolistsep]
\item Assumption \ref{aspt_pe_gamp_map}(a).
\item Assumption \ref{aspt_pe_gamp_map}(c).
\item Assumption \ref{aspt_pe_gamp_map}(d).
\end{enumerate}
\end{assumption}
\bigskip
\subsubsection{Empirical Convergence Analysis}
We next give the following Lemma \ref{lemma:map_pl_continuity} about the estimation functions $f_{\Omega_n\leftarrow\lambda_l}(\cdot), f_{\Phi_m\leftarrow\theta_k}(\cdot)$ for the proposed PE-GAMP:
\smallskip
\begin{lemma}
\label{lemma:map_pl_continuity}
Under Assumption \ref{aspt_pe_gamp_map} for MAP parameter estimation and Assumption \ref{aspt_pe_gamp_mmse} for MMSE parameter estimation, the estimation functions $f_{\Omega_n\leftarrow\lambda_l}\left(t, \vr_{\Omega}^{(t)}, \tau_{\Omega}^r(t), \lambda_l, \hat{\boldsymbol\lambda}_{\Omega_n}^{(t)}\backslash\hat{\lambda}_{\Omega_n\leftarrow\lambda_l}^{(t)}\right)$ can be considered as a function of $\vr_{\Omega}^{(t)}$ that satisfies the weak pseudo-Lipschitz continuity property: If the sequence of vector $\vr_{\Omega}^{(t)}$ indexed by $N$ empirically converges with bounded moments of order $k=2$ and the sequence of scalers $\tau_{\Omega}^r(t), \left.\hat{\boldsymbol\lambda}_{\Omega_n}^{(t)}\backslash\hat{\lambda}_{\Omega_n\leftarrow\lambda_l}^{(t)}\right.$ also converge as follows:
\begin{subequations}
\begin{align}
\label{eq:rv_convergence}
&\lim_{N\rightarrow\infty}\vr_{\Omega}^{(t)}\stackrel{\normalfont\textrm{PL$(k)$}}{=}\mathcal{R}_\Omega^{(t)}\\
&\lim_{N\rightarrow\infty}\tau_{\Omega}^r(t)=\overline{\tau}_\Omega^r(t)\\
&\lim_{N\rightarrow\infty}\left.\hat{\boldsymbol\lambda}_{\Omega_n}^{(t)}\backslash\hat{\lambda}_{\Omega_n\leftarrow\lambda_l}^{(t)}\right.=\left.\overline{\boldsymbol\lambda}_{\Omega_n}^{(t)}\backslash\overline{\lambda}_{\Omega_n\leftarrow\lambda_l}^{(t)}\right.\,.
\end{align}
\end{subequations}
Then,
\begin{align}
\label{eq:eca_conclusion}
\begin{split}
&\lim_{N\rightarrow\infty}f_{\Omega_n\leftarrow\lambda_l}\left(t, \vr_{\Omega}^{(t)}, \tau_{\Omega}^r(t), \lambda_l, \hat{\boldsymbol\lambda}_{\Omega_n}^{(t)}\backslash\hat{\lambda}_{\Omega_n\leftarrow\lambda_l}^{(t)}\right)\\
&\quad=f_{\Omega_n\leftarrow\lambda_l}\left(t,\mathcal{R}^{(t)}_\Omega, \overline{\tau}_\Omega^r(t), \lambda_l, \overline{\boldsymbol\lambda}_{\Omega_n}^{(t)}\backslash\overline{\lambda}_{\Omega_n\leftarrow\lambda_l}^{(t)}\right)\,.
\end{split}
\end{align} 
Similarly, $f_{\Phi_m\leftarrow\theta_k}\left(t, \vq_{\Phi}^{(t)}, \vy, \tau_{\Phi}^q(t), \theta_k, \hat{\boldsymbol\theta}_{\Phi_m}^{(t)}\backslash\hat{\theta}_{\Phi_m\leftarrow\theta_k}^{(t)}\right)$ also satisfies the weak pseudo-Lipschitz continuity property.
\end{lemma}
\begin{proof}
Please refer to Appendix \ref{app:proof_lemma}.
\end{proof}
Additionally, we make the following assumptions about the proposed PE-GAMP algorithm.
\smallskip
\begin{assumption}
\label{aspt_pe_gamp}
The PE-GAMP solves a series of estimation problems indexed by the input signal dimension $N$:
\begin{enumerate}[label={\alph*)}, nolistsep]
\item Assumptions \ref{aspt_gamp}(a) to \ref{aspt_gamp}(d) with $k=2$.
\item The scalar estimation function $g_\textrm{in}(t, r_{\Omega_n}, \tau_\Omega^r, \boldsymbol\lambda)$ and its derivative $g_\textrm{in}^\prime(t, r_{\Omega_n}, \tau_\Omega^r, \boldsymbol\lambda)$ with respect to $r_{\Omega_n}$ are continuous in $\boldsymbol\lambda$ uniformly over $r_{\Omega_n}$: For every $\epsilon>0,t,\widetilde{\tau}_\Omega^r, \widetilde{\boldsymbol\lambda}\in\mathcal{U}_\lambda$, there exists an open neighborhood $\rho(\widetilde{\tau}_\Omega^r, \widetilde{\boldsymbol\lambda})$ of $(\widetilde{\tau}_\Omega^r, \widetilde{\boldsymbol\lambda}\in\mathcal{U}_\lambda)$ such that $\forall (\tau_\Omega^r, \boldsymbol\lambda)\in\rho(\widetilde{\tau}_\Omega^r, \widetilde{\boldsymbol\lambda})$ and $r$,
\begin{subequations}
\begin{align}
&|g_\textrm{in}(t, r_{\Omega_n}, \tau_\Omega^r, \boldsymbol\lambda)-g_\textrm{in}(t, r_{\Omega_n}, \widetilde{\tau}_\Omega^r, \widetilde{\boldsymbol\lambda})|<\epsilon\\
&|g_\textrm{in}^\prime(t, r_{\Omega_n}, \tau_\Omega^r, \boldsymbol\lambda)-g_\textrm{in}^\prime(t, r_{\Omega_n}, \widetilde{\tau}_\Omega^r, \widetilde{\boldsymbol\lambda})|<\epsilon\,.
\end{align}
\end{subequations}
In addition, $g_\textrm{in}(\cdot), g_\textrm{in}^\prime(\cdot)$ is pseudo-Lipschitz continuous in $r_{\Omega_n}$ with a Lipschitz constant that can be selected continuously in $\tau_\Omega^r$ and $\boldsymbol\lambda$. $g_\textrm{out}(t,q_{\Phi_m},\tau_\Phi^q,y_m, \boldsymbol\theta), g_\textrm{out}^\prime(t,q_{\Phi_m},\tau_\Phi^q,y_m, \boldsymbol\theta)$ also satisfy analogous continuity assumptions with respect to $q,y,\tau_\Phi^q,\boldsymbol\theta$.
\item For each $m=1,\cdots,M$ and $n=1,\cdots,N$, the components of the initial condition $\hat{\boldsymbol\lambda}_{\Omega_n}^{(0)}, \hat{\boldsymbol\theta}_{\Phi_m}^{(0)}$ converge as follows:
\begin{align}
\label{eq:pe_gamp_parameter_init}
\lim_{N\rightarrow\infty}(\hat{\boldsymbol\lambda}_{\Omega_n}^{(0)}, \hat{\boldsymbol\theta}_{\Phi_m}^{(0)})=(\overline{\boldsymbol\lambda}_{\Omega_n}^{0}, \overline{\boldsymbol\theta}_{\Phi_m}^{(0)})\,.
\end{align}
\end{enumerate}
\end{assumption}
Specifically, Assumptions \ref{aspt_pe_gamp}(a) and \ref{aspt_pe_gamp}(b) are the same as those in \cite{AMP_CPE14}; Assumptions \ref{aspt_pe_gamp}(c) is made for the proposed PE-GAMP. We then have the following Corollary \ref{crl:pe_gamp_convergence} using Theorem \ref{thm:adaptive_gamp}:
\smallskip
\begin{corollary}
\label{crl:pe_gamp_convergence}
Consider the proposed PE-GAMP with scalar variances under the Assumptions $\lceil$\ref{aspt_pe_gamp_map},\ref{aspt_pe_gamp}$\rfloor$ for MAP parameter estimation and Assumptions $\lceil$\ref{aspt_pe_gamp_mmse},\ref{aspt_pe_gamp}$\rfloor$ for MMSE parameter estimation. Then for any fixed iteration number $t$: the scalar components of (\ref{eq:pe_gamp_parameters_estimate},\ref{eq:rv_gamp}) empirically converge with bounded moments of order $k=2$ as follows:
\begin{subequations}
\begin{align}
&\lim_{N\rightarrow\infty}\boldsymbol\psi_\textrm{in}\stackrel{\normalfont\textrm{PL$(k)$}}{=}\overline{\boldsymbol\psi}_\textrm{in},\quad\lim_{N\rightarrow\infty}\boldsymbol\psi_\textrm{out}\stackrel{\normalfont\textrm{PL$(k)$}}{=}\overline{\boldsymbol\psi}_\textrm{out}\\
&\lim_{N\rightarrow\infty}\boldsymbol\psi_\tau=\overline{\boldsymbol\psi}_\tau\\
&\lim_{N\rightarrow\infty}\hat{\boldsymbol\theta}_{\Phi_m}^{(t+1)}=\overline{\boldsymbol\theta}_{\Phi_m}^{(t+1)},\quad\lim_{N\rightarrow\infty}\hat{\boldsymbol\lambda}_{\Omega_n}^{(t+1)}=\overline{\boldsymbol\lambda}_{\Omega_n}^{(t+1)}\,.
\end{align}
\end{subequations}
\end{corollary}
\begin{proof}
Please refer to Appendix \ref{app:proof_corollary}.
\end{proof}

\section{Numerical Results}
\label{sec:numerical_results}
Depending on the various sparse signal recovery tasks, we can assume the sparse signal $\vx$ and the noise $\vw$ are generated from the following input and output channels:
\begin{itemize}
\item {\bfseries Bernoulli-Gaussian mixture (BGm) Input Channel:}
The sparse signal $\vx\in\mathbb{R}^N$ can be modeled as a mixture of Bernoulli and Gaussian mixture distributions:
\begin{align}
\label{eq:bgm_distribution}
\begin{split}
p(x_j|\boldsymbol\lambda)&= (1-\lambda_1)\delta(x_j)\\
&\quad\quad+\lambda_1\sum_{c=1}^C\lambda_{c+1}\cdot\mathcal{N}(x_j;\lambda_{c+2},\lambda_{c+3})\,,
\end{split}
\end{align}
where $x_j\in\mathbb{R}$; $\delta(\cdot)$ is Dirac delta function; $\lambda_1\in[0,1]$ is the sparsity rate; for the $c$-th Gaussian mixture, $\lambda_{c+1}\in[0,1]$ is the mixture weight, $\lambda_{c+2}\in\mathbb{R}$ is the nonzero coefficient mean and $\lambda_{c+3}\in(0,\infty)$ is the nonzero coefficient variance; all the mixture weights should sum to $1$: $\sum_{c=1}^C\lambda_{3c-1}=1$.
\item {\bfseries Bernoulli-Exponential mixture (BEm) Input Channel:}
Nonnegative sparse signal $\vx\in\mathbb{R}^N$ can be modeled as a mixture of Bernoulli and Exponential mixture distributions:
\begin{align}
\label{eq:bem_distribution}
\begin{split}
p(x_j|\boldsymbol\lambda) &= (1-\lambda_1)\delta(x_j)\\
&\quad+\lambda_1\sum_{c=1}^C\lambda_{c+1}\cdot\lambda_{c+2}\exp\left(-\lambda_{c+2}x_j\right)\,,
\end{split}
\end{align}
where $x_j\in[0, \infty)$; $\lambda_1\in[0,1]$ is the sparsity rate; for the $c$-th Exponential mixture, $\lambda_{c+1}\in[0,1]$ is the mixture weight and $\lambda_{c+2}\in(0,\infty)$; all the mixture weights should sum to $1$: $\sum_{c=1}^C\lambda_{2c}=1$.
\item {\bfseries Laplace Input Channel:}
The sparse signal $\vx\in\mathbb{R}^N$ follows the following Laplace distribution:
\begin{align}
\label{eq:laplace_distribution}
p(x_j|\boldsymbol\lambda) = \frac{\lambda_1}{2}\exp\left(-\lambda_1|x_j|\right)\,,
\end{align}
where $x_j\in\mathbb{R}$; $\lambda_1\in(0,\infty)$.
\item {\bfseries Additive White Gaussian Noise (AWGN) Output Channel:}
The noise $\vw\in\mathbb{R}^M$ is assumed to be white Gaussian noise:
\begin{align}
\label{eq:adgn_distribution}
p(w_i|\boldsymbol\theta) = \mathcal{N}(w_i;0,\theta_1)\,,
\end{align}
where $w_i\in\mathbb{R}$ is the noise; $\theta_1\in(0,\infty)$ is its variance.
\end{itemize}
Using the above channels we can create three sparse signal recovery models: 1) BGm + AWGN; 2) BEm + AWGM; 3) Laplace + AWGN. 

\subsection{MAP Parameter Estimation}
As is shown in Appendix \ref{app:map_pe}, for the models with BGm and BEm input channels, \emph{max-sum} message passing cannot be used to perform the inference task on the sparse signal since the maximizing $\vx$ in (\ref{eq:ms_fv_phi_x},\ref{eq:ms_fv_omega_lambda},\ref{eq:ms_fv_phi_theta}) would be all zeros. (\ref{eq:pm_x}) from \emph{sum-product} message passing cannot produce any useful MAP estimation of $\vx$ for the same reason. In this case, we can only use \emph{sum-product} message passing to perform MMSE estimation of $\vx$. 

For the model with Laplace input channel, although \emph{max-sum} message passing can be used to obtain the MAP estimation of $\vx$, it cannot be used to compute the MAP estimation of $\lambda_1$, since the $\hat{\lambda}_1$ that maximizes (\ref{eq:ms_pe_map}) is always $\infty$ and the maximizing $\hat{\theta}_1$ is always $0$. On the other hand, \emph{sum-product} message passing can be used to compute the MMSE estimation and MAP estimation of $x_n$ based on $p(x_n|\vy)$, however they doesn't have the best recovery performance. Here we propose to employ \emph{sum-product} message passing to compute the ``marginal'' MAP estimates $\{\hat{\lambda}_1,\hat{\theta}_1\}$ using the marginal posterior distributions $p(\lambda_1|\vy), p(\theta_1|\vy)$, as opposed to the MAP estimates in (\ref{eq:map_pe_gamp}). $\{\hat{\lambda}_1,\hat{\theta}_1\}$ can then be used as the inputs to \emph{max-sum} message passing to obtain the MAP estimate of $\vx$. This essentially is the Lasso mentioned at the beginning of this paper, except now that we have provided a way to automatically estimate the parameters.

In this case, the two recovery models mentioned earlier both rely on \emph{sum-product} message passing to perform parameter estimation. For the \emph{sum-product} message passing, ``MMSE parameter estimation'' is often quite difficult to compute, in this paper we will focus on using the ``MAP parameter estimation'' approach to estimate the parameters. Since we don't have any knowledge about the priors of $\boldsymbol\lambda,\boldsymbol\theta$, we will fairly choose the \emph{uniform} prior for each parameter. 

The proposed PE-GAMP computes MAP estimations of the parameters in the \emph{sum-product} message passing as follows:
\begin{subequations}
\label{eq:map_pe_gamp}
\begin{align}
\begin{split}
\hat{\lambda}_{\Omega_n\leftarrow\lambda_l}^{(t+1)}&=\arg\max_{\lambda_l}h_{\Omega_n\leftarrow\lambda_l}^{(t+1)}(\cdot)\\
&=\arg\max_{\lambda_l}\sum_{j\neq n}\Delta^{(t+1)}_{\Omega_j\rightarrow\lambda_l}+\log p(\lambda_l)
\end{split}\\
\begin{split}
\hat{\theta}_{\Phi_m\leftarrow\theta_k}^{(t+1)}&=\arg\max_{\theta_k}h_{\Phi_m\leftarrow\theta_k}^{(t+1)}(\cdot)\\
&=\arg\max_{\theta_k}\sum_{i\neq k}\Delta^{(t+1)}_{\Phi_i\rightarrow\theta_k} + \log p(\theta_k)\,.
\end{split}
\end{align}
\end{subequations}

Specifically, we use the line search method given in the following Algorithm \ref{alg:line_search} to find $\hat{\lambda}_{\Omega_n\leftarrow\lambda_l}^{(t+1)}$. 
\begin{algorithm}
\caption{Line Search Method}\label{alg:line_search}
\begin{algorithmic}[1]
\Require $\hat{\lambda}_{\Omega_n\leftarrow\lambda_l}^{(t)}$, $\frac{\partial h_{\Omega_n\leftarrow\lambda_l}^{(t+1)}}{\partial \lambda_l}$, $0<\zeta<1$, $\eta_+>0$, $\eta_-<0$
\State Set $\hat{\lambda}_{\Omega_n\leftarrow\lambda_l}^{(t+1)}(0)=\hat{\lambda}_{\Omega_n\leftarrow\lambda_l}^{(t)}$.
\For {$i=1,2,\cdots$}
\If {$\frac{\partial h_{\Omega_n\leftarrow\lambda_l}^{(t+1)}}{\partial \lambda_l}\left|_{\lambda_l=\hat{\lambda}_{\Omega_n\leftarrow\lambda_l}^{(t+1)}(i-1)}\right.>0$}
\begin{align} 
\hat{\lambda}_{\Omega_n\leftarrow\lambda_l}^{(t+1)}=\hat{\lambda}_{\Omega_n\leftarrow\lambda_l}^{(t+1)}(i-1)+\eta_+\,.
\end{align}
\While {$h_{\Omega_n\leftarrow\lambda_l}^{(t+1)}\left|_{\hat{\lambda}_{\Omega_n\leftarrow\lambda_l}^{(t+1)}}\right.<h_{\Omega_n\leftarrow\lambda_l}^{(t+1)}\left|_{\hat{\lambda}_{\Omega_n\leftarrow\lambda_l}^{(t+1)}(i-1)}\right.$}
\begin{subequations}
\begin{align}
\eta_+&=\eta_+\cdot\zeta\\
\hat{\lambda}_{\Omega_n\leftarrow\lambda_l}^{(t+1)}&=\hat{\lambda}_{\Omega_n\leftarrow\lambda_l}^{(t+1)}(i-1)+\eta_+\,.
\end{align}
\end{subequations}
\EndWhile
\ElsIf {$\frac{\partial h_{\Omega_n\leftarrow\lambda_l}^{(t+1)}}{\partial \lambda_l}\left|_{\lambda_l=\hat{\lambda}_{\Omega_n\leftarrow\lambda_l}^{(t+1)}(i-1)}\right.<0$}
\begin{align}
\hat{\lambda}_{\Omega_n\leftarrow\lambda_l}^{(t+1)}=\hat{\lambda}_{\Omega_n\leftarrow\lambda_l}^{(t+1)}(i-1)+\eta_-
\end{align}
\While {$h_{\Omega_n\leftarrow\lambda_l}^{(t+1)}\left|_{\hat{\lambda}_{\Omega_n\leftarrow\lambda_l}^{(t+1)}}\right.<h_{\Omega_n\leftarrow\lambda_l}^{(t+1)}\left|_{\hat{\lambda}_{\Omega_n\leftarrow\lambda_l}^{(t+1)}(i-1)}\right.$}
\begin{subequations}
\begin{align}
\eta_-&=\eta_-*\zeta\\
\hat{\lambda}_{\Omega_n\leftarrow\lambda_l}^{(t+1)}&=\hat{\lambda}_{\Omega_n\leftarrow\lambda_l}^{(t+1)}(i-1)+\eta_-\,.
\end{align}
\end{subequations}
\EndWhile
\Else
	\State break;
\EndIf
\State Set $\hat{\lambda}_{\Omega_n\leftarrow\lambda_l}^{(t+1)}(i)=\hat{\lambda}_{\Omega_n\leftarrow\lambda_l}^{(t+1)}$
\If {$\hat{\lambda}_{\Omega_n\leftarrow\lambda_l}^{(t+1)}(i)$ reaches convergence}
\State break;
\EndIf
\EndFor
\State\Return Output $\hat{\lambda}_{\Omega_n\leftarrow\lambda_l}^{(t+1)}$;
\end{algorithmic}
\end{algorithm}	

The maximizing $\hat{\theta}_{\Phi_m\leftarrow\theta_k}^{(t+1)}$ can be found similarly. The line search method requires computing the derivatives of $h_{\Omega_n\leftarrow\lambda_l}^{(t+1)}(\cdot), h_{\Phi_m\leftarrow\theta_k}^{(t+1)}(\cdot)$ with respect to the parameters, which are given in Appendix \ref{app:map_pe} for different channels. 

\subsection{Comparison with EM Parameter Estimation}
Here we discuss the differences between the proposed PE-GAMP with MAP parameter estimation and the EM-GAMP with EM parameter estimation \cite{EMBGAMP11,EMGMAMP13}. 

First of all, the EM parameter estimation is essentially maximum likelihood estimation. EM \cite{EM77} tries to find the parameters $\boldsymbol\lambda,\boldsymbol\theta$ that maximize the likelihood $p(\vy|\boldsymbol\lambda), p(\vy|\boldsymbol\theta)$. While the proposed PE-GAMP with MAP parameter estimation tries to maximize the following posterior distributions at nodes $\Omega_n, \Phi_m$ using Bayes' rule:
\begin{subequations}
\begin{align}
&p_{\Omega_n}(\boldsymbol\lambda|\vy)\propto p_{\Omega_n}(\vy|\boldsymbol\lambda)p(\boldsymbol\lambda)\\
&p_{\Phi_m}(\boldsymbol\theta|\vy)\propto p_{\Phi_m}(\vy|\boldsymbol\theta)p(\boldsymbol\theta)\,.
\end{align}
\end{subequations}
Compared to EM estimation, the MAP estimation is able to draw information from the priors $p(\boldsymbol\lambda), p(\boldsymbol\theta))$ to guide the estimation process.

Secondly, the two methods also differ in the way they compute the maximizing parameters. For the sake of simplification and a fair comparison, we will assume the priors of the parameters $p(\boldsymbol\lambda),p(\boldsymbol\theta)$ to be uniform distributions. Specifically, EM treats $\vx,\vw$ as hidden variables and maximizes $\mathbb{E}[\log p(\vx,\vw;\boldsymbol\lambda,\boldsymbol\theta)|\vy,\hat{\boldsymbol\lambda^{(t)}},\hat{\boldsymbol\theta^{(t)}}]$ iteratively until convergence. Take the parameter $\lambda_l$ for example, in the $(t+1)$-th iteration the following expression will be maximized under the GAMP framework \cite{EMGMAMP13}:
\begin{align}
\label{eq:em_maximization}
\begin{split}
&\max_{\lambda_l}\sum_j\int p\left(x_j|\vy,\hat{\lambda}_l^{(t)}\right)\log p(x_j|\lambda_l)\,dx_j\\
&\propto\sum_j\int p\left(x_j|\hat{\lambda}_l^{(t)}\right) \mathcal{N}\left(x_j; r_{\Omega_j}^{(t)}, \tau_{\Omega_j}^r(t)\right)\log p(x_j|\lambda_l)\,dx_j\,,
\end{split}
\end{align}
where $\hat{\lambda}_l^{(t)}$ is the estimated parameter in the previous $t$-th iteration. \cite{EMGMAMP13} gives the closed-form expression for Bernoulli-Gaussian mixture distributions. However, (\ref{eq:em_maximization}) is quite difficult to evaluate for more complicated distributions, which greatly limits its applicabilities. The proposed PE-GAMP with MAP parameter estimation has a much simpler expression though:
\begin{align}
\max_{\lambda_l}\sum_{j\neq n}\log\int p(x_j|\lambda_l)\mathcal{N}\left(x_j; r_{\Omega_j}^{(t)}, \tau_{\Omega_j}^r(t)\right)\,dx_j\,,
\end{align}
This enables us to consider more complex distributions with the proposed PE-GAMP. For instance, in this paper we have included the formulations to estimate the parameters for sparse signals with Laplace prior and Bernoulli-Exponential mixture prior in Appendix \ref{app:map_pe}.

\begin{figure*}[tbp]
\centering
\subfigure[]{
\label{fig:ptc_noiseless_bg}
\includegraphics[height=2.8in]{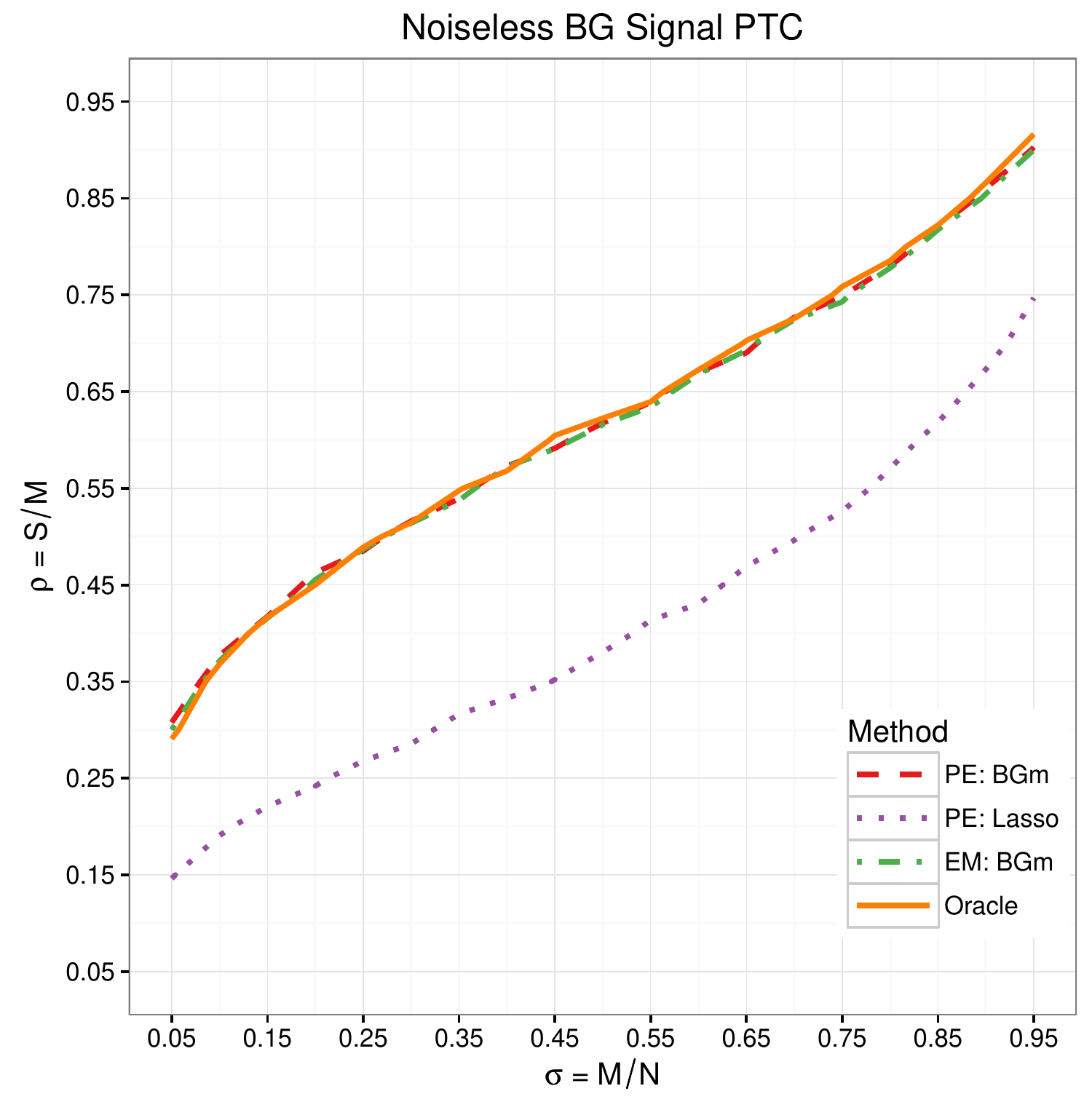}}
\subfigure[]{
\label{fig:ptc_noiseless_nonneg}
\includegraphics[height=2.8in]{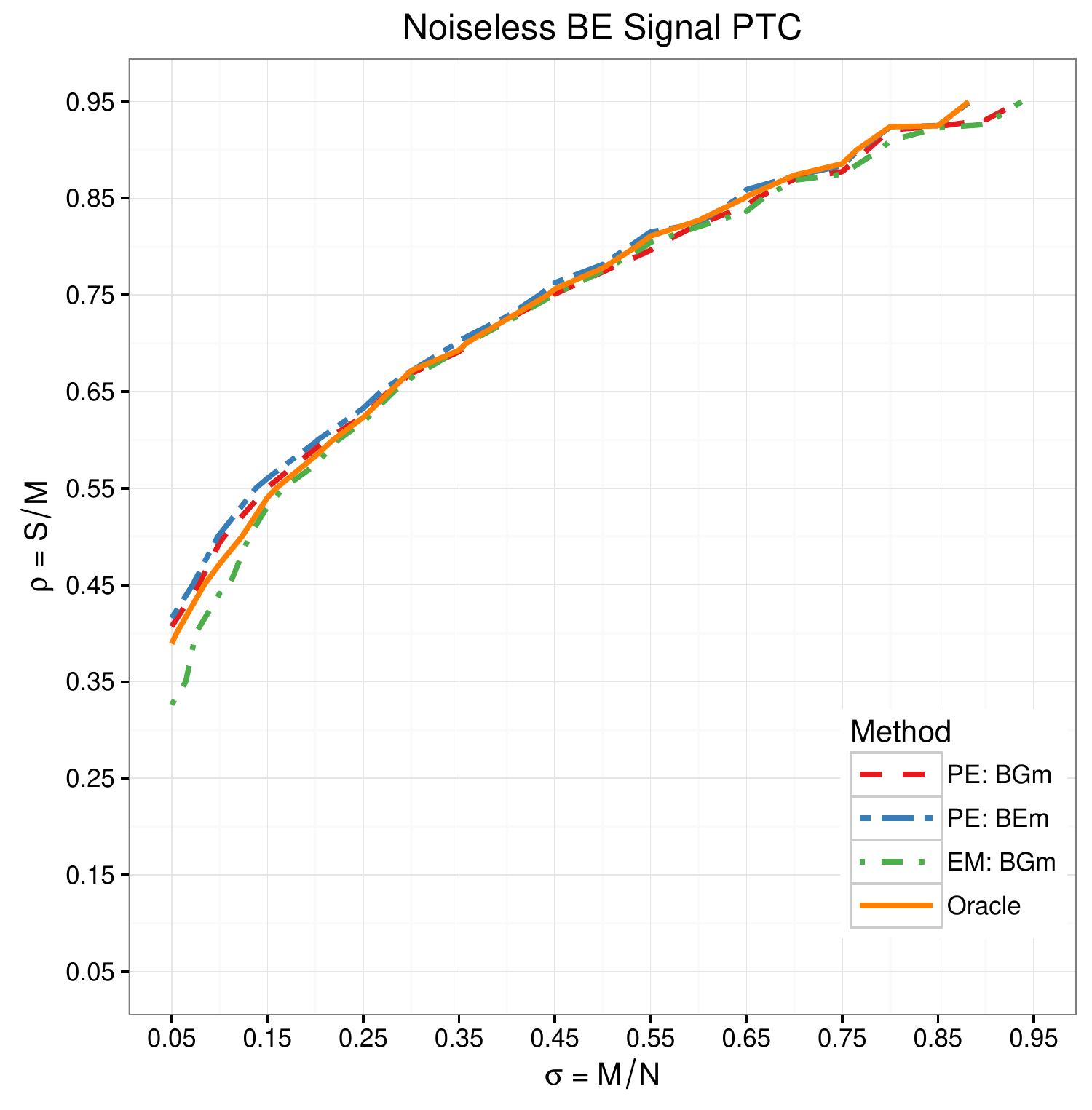}}
\caption{\small The phase transition curves (PTC) of different GAMP methods in the noiseless case. (a) Bernoulli-Gaussian (BG) sparse signal; (b) Bernoulli-Exponential (BE) sparse signal.} 
\label{fig:ptc_noiseless}
\end{figure*}

\begin{figure*}[tbp]
\centering
\subfigure[]{
\label{fig:snr_noisy_bg}
\includegraphics[height=2.8in]{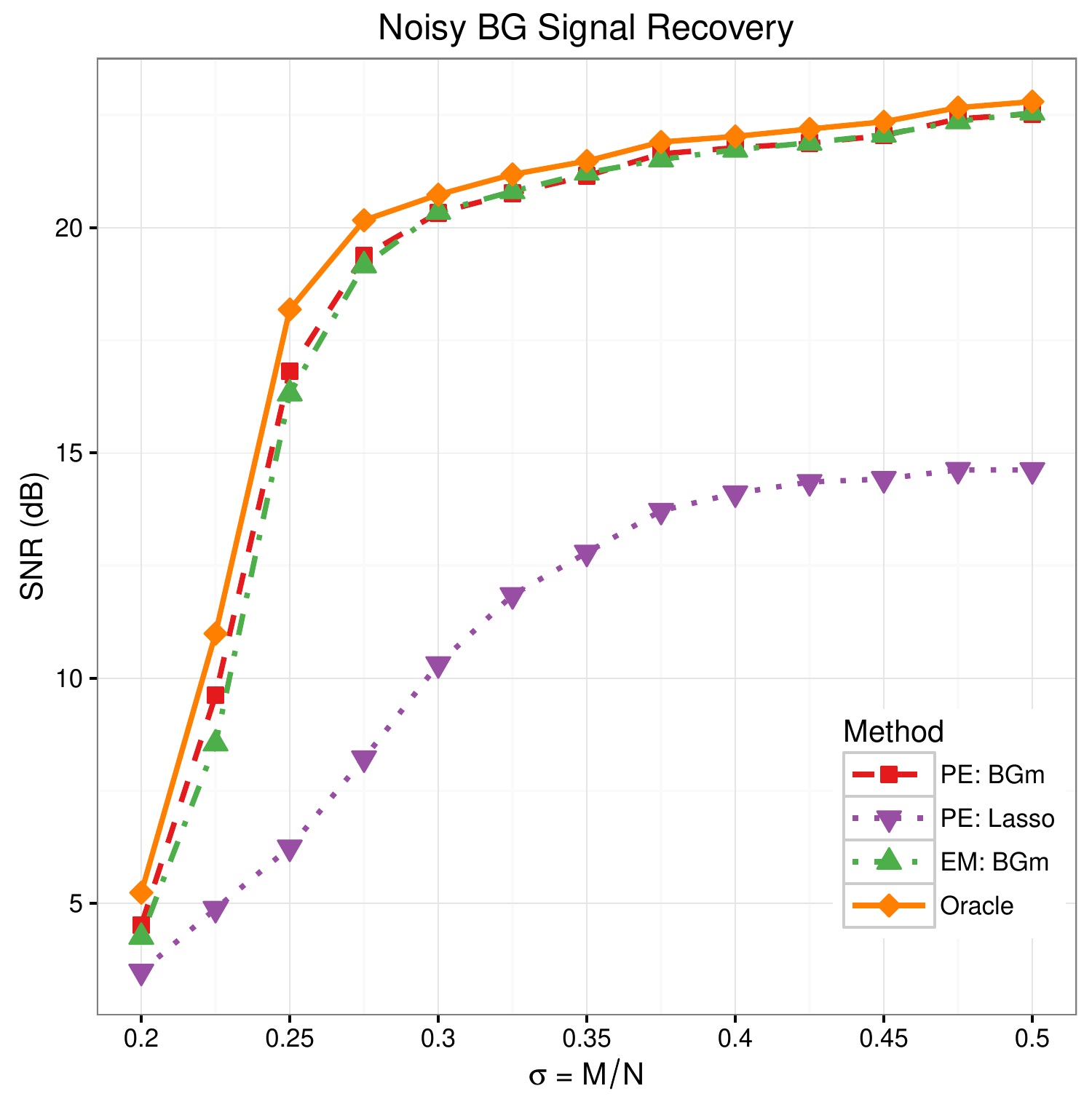}}
\subfigure[]{
\label{fig:snr_noisy_nonneg}
\includegraphics[height=2.8in]{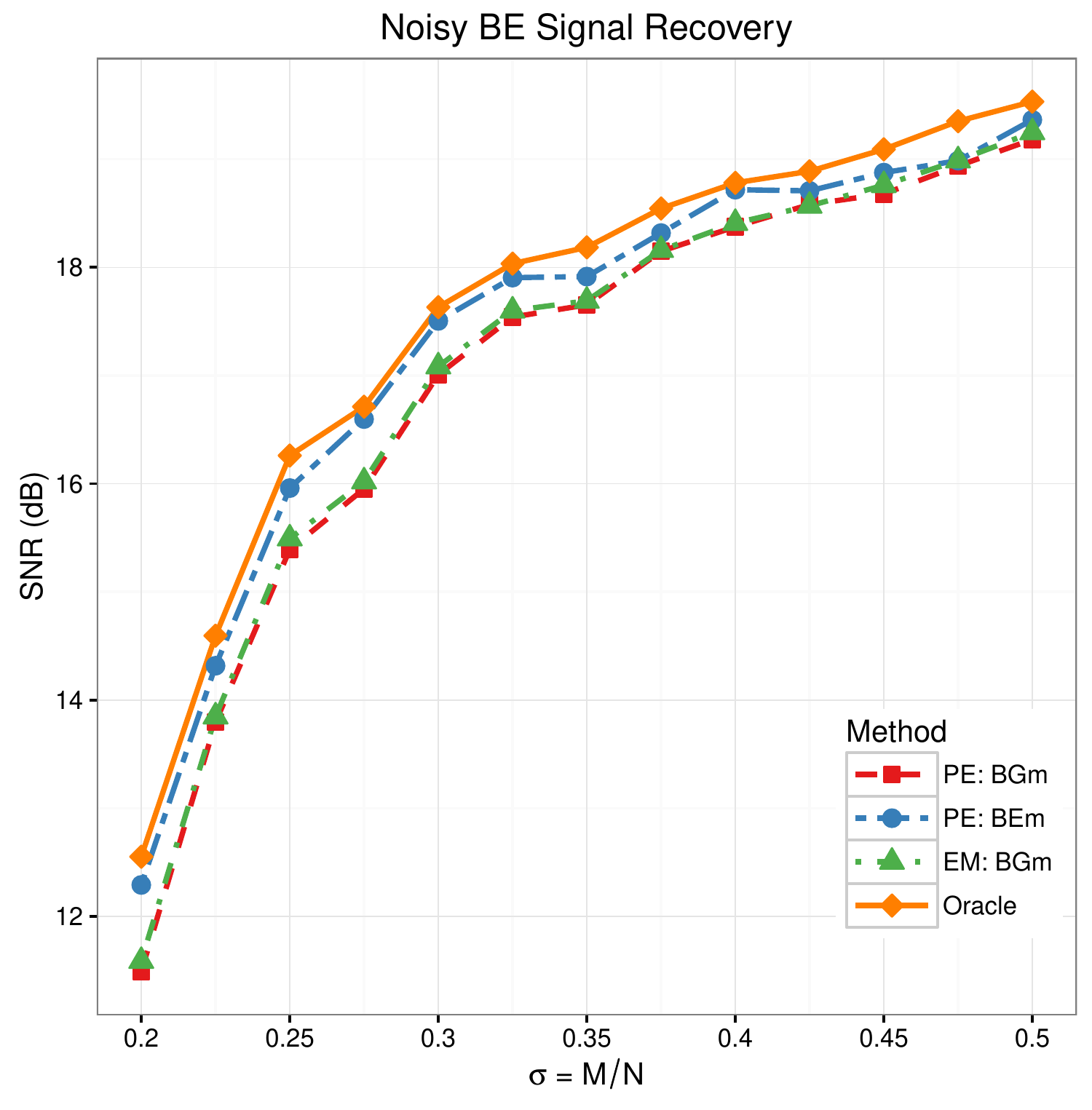}}
\caption{\small The signal-to-noise-ratio (SNR) of the recovered sparse signals using different GAMP methods in the noisy case. (a) Bernoulli-Gaussian (BG) sparse signal; (b) Bernoulli-Exponential (BE) sparse signal.} 
\label{fig:snr_noisy}
\end{figure*}

\subsection{Noiseless Sparse Signal Recovery}
We first perform noiseless sparse signal recovery experiments and compare the empirical phase transition curves (PTC) of PE-GAMP and EM-BGm-GAMP \cite{EMGMAMP13}. Besides, oracle experiments where the ``true'' parameters are known are also performed. Specifically, we fix $N=1000$ and vary the over-sampling ratio $\sigma=\frac{M}{N}\in[0.05,0.1,0.15,\cdots,0.95]$ and the under-sampling ratio $\rho=\frac{S}{M}\in[0.05,0.1,0.15,\cdots,0.95]$, where $S$ is the sparsity of the signal, i.e. the number of nonzero coefficients. For each combination of $\sigma$ and $\rho$, we randomly generate $100$ pairs of $\{\vx,\vA\}$: $A$ is a $M\times N$ random Gaussian matrix with normalized and centralized rows; the \emph{nonzero} entries of the sparse signal $\vx\in\mathbb{R}^N$ are i.i.d. generated according to the following two different distributions:
\begin{enumerate}
\item Gaussian distribution $x\sim\mathcal{N}(0,1)$.
\item Exponential distribution $x\sim \exp(-x)$, $x\geq0$.
\end{enumerate}
In other words, the sparse signals $\vx$ follow Bernoulli-Gaussian (BG) and Bernoulli-Exponential (BE) distributions respectively. Given the measurement vector $\vy=\vA\vx$ and the sensing matrix $\vA$, we try to recover the signal $\vx$. If $\epsilon=\|\vx-\hat{\vx}\|_2/\|\vx\|_2<10^{-3}$, the recovery is considered to be a success. Based on the $100$ trials, we compute the success recovery rate for each combination of $\sigma$ and $\rho$ and plot the PTCs in Fig. \ref{fig:ptc_noiseless}. 

The PTC is the contour that corresponds to the 0.5 success rate in the domain $(\sigma,\rho)\in(0,1)^2$, it divides the domain into a ``success'' phase (lower right) and a ``failure'' phase (upper left). For the BG sparse signals (Fig. \ref{fig:ptc_noiseless_bg}), the PE-BGm-GAMP and EM-BGm-GAMP perform equally well and match the performance of the oracle-GAMP. The BGm prior they assumed about the sparse signal is a perfect match, which is much better than Laplace prior assumed by PE-Lasso-GAMP. 

For the BE sparse signals (Fig. \ref{fig:ptc_noiseless_nonneg}), the BEm prior assumed by PE-BEm-GAMP is the perfect match. However, we can see that the PTC of PE-BGm-GAMP is only slightly worse, the BGm prior is still a strong contestant in this case. Although both PE-BGm-GAMP and EM-BGm-GAMP assume the BGm prior, PE-BGm-GAMP is more robust and performs better than EM-BGm-GAMP when the sampling rate is low. PE-BEm-GAMP is the only one that matches the performance of the oracle-GAMP.

\begin{figure*}[tbp]
\centering
\subfigure[]{
\label{fig:psnr_barbara}
\includegraphics[height=2.8in]{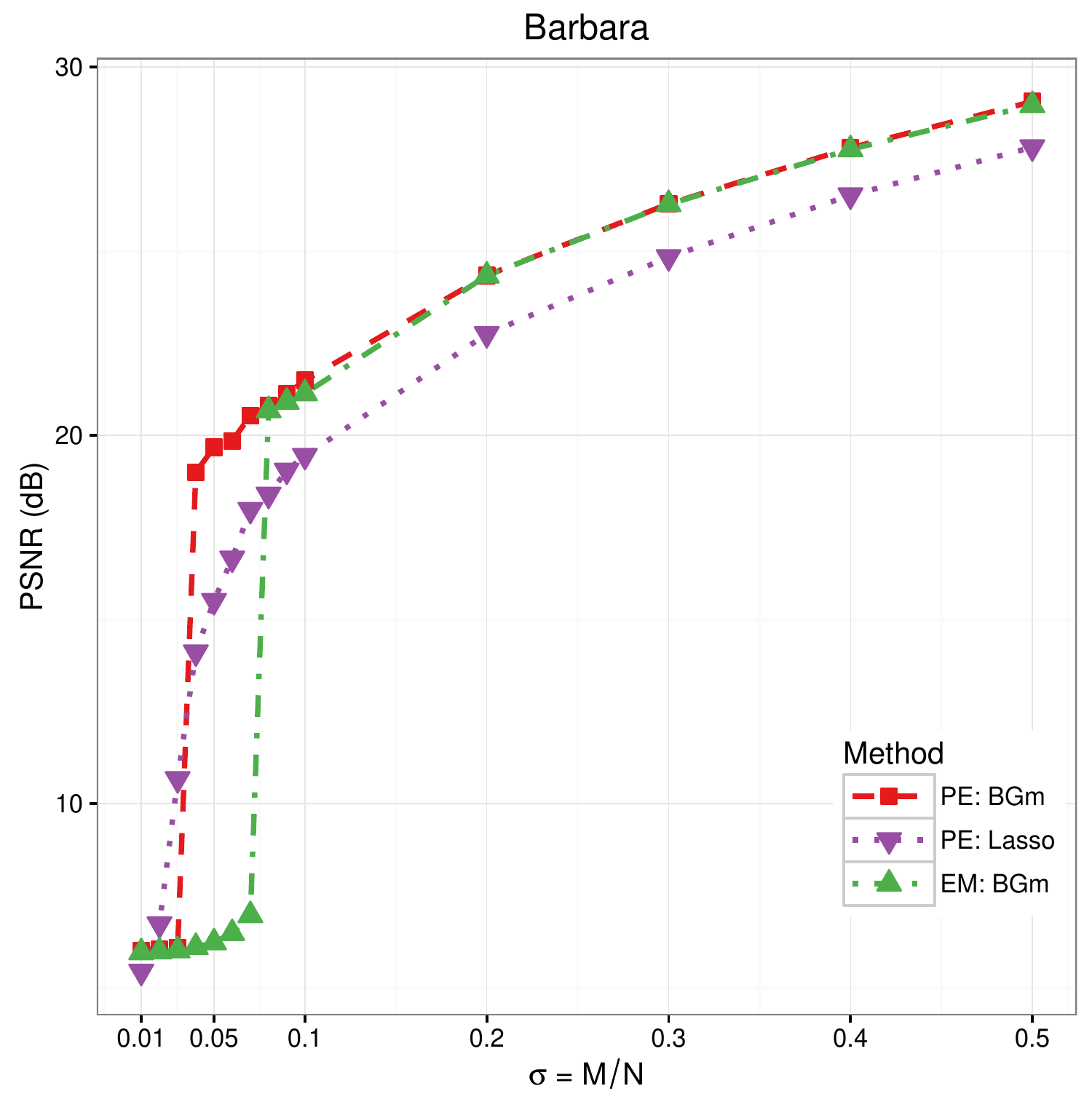}}
\subfigure[]{
\label{fig:psnr_boat}
\includegraphics[height=2.8in]{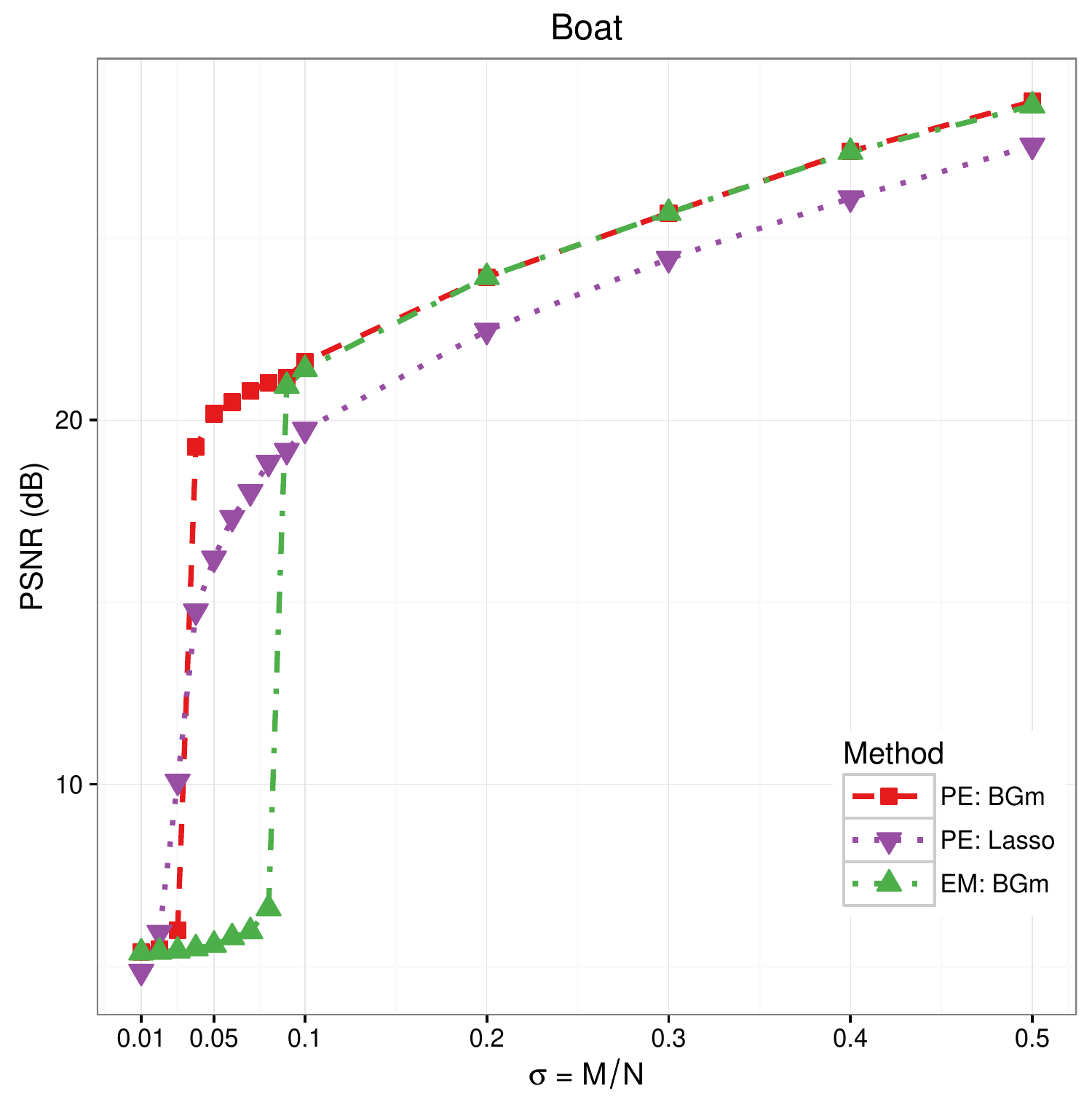}}
\subfigure[]{
\label{fig:psnr_lena}
\includegraphics[height=2.8in]{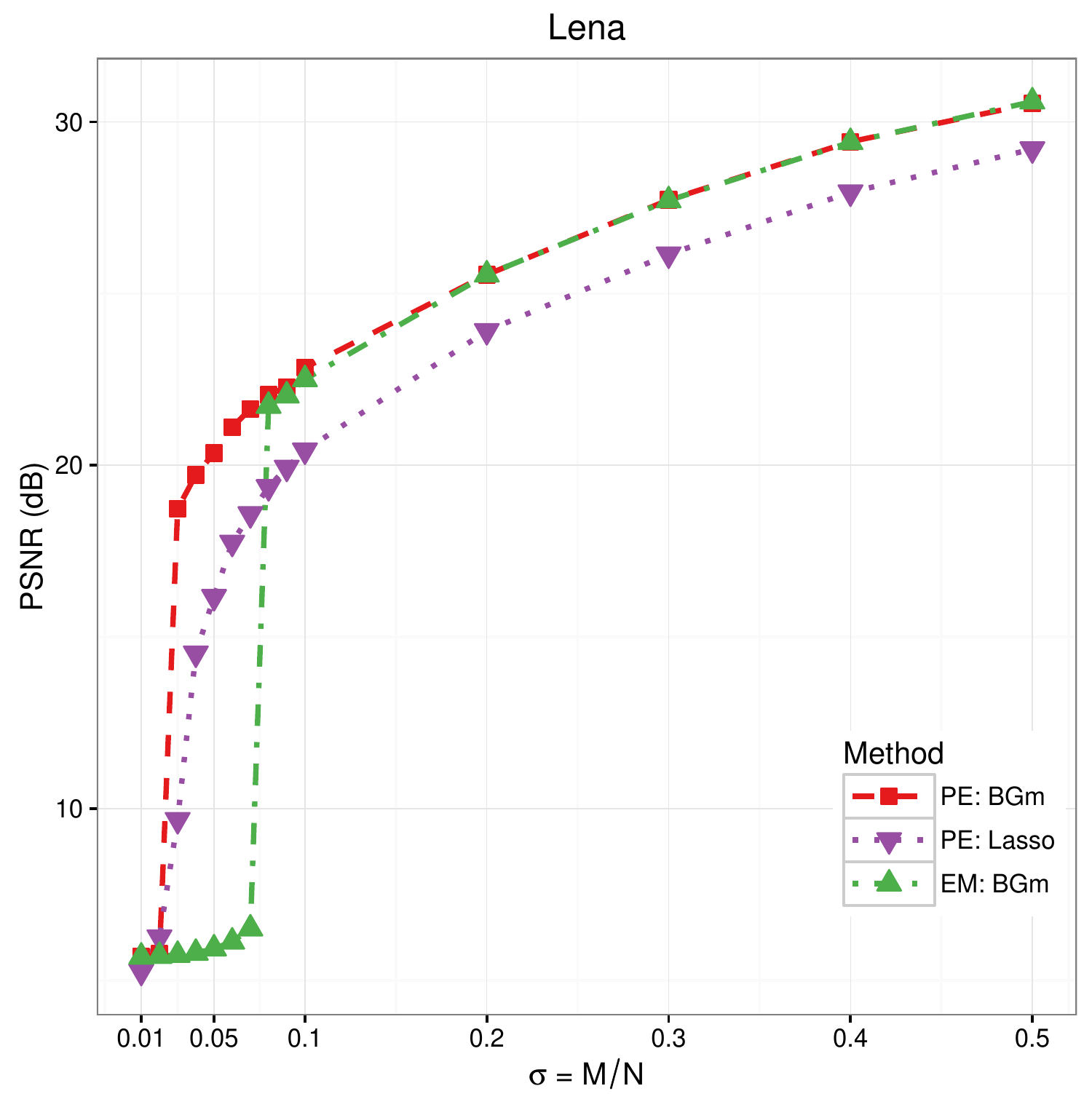}}
\subfigure[]{
\label{fig:psnr_peppers}
\includegraphics[height=2.8in]{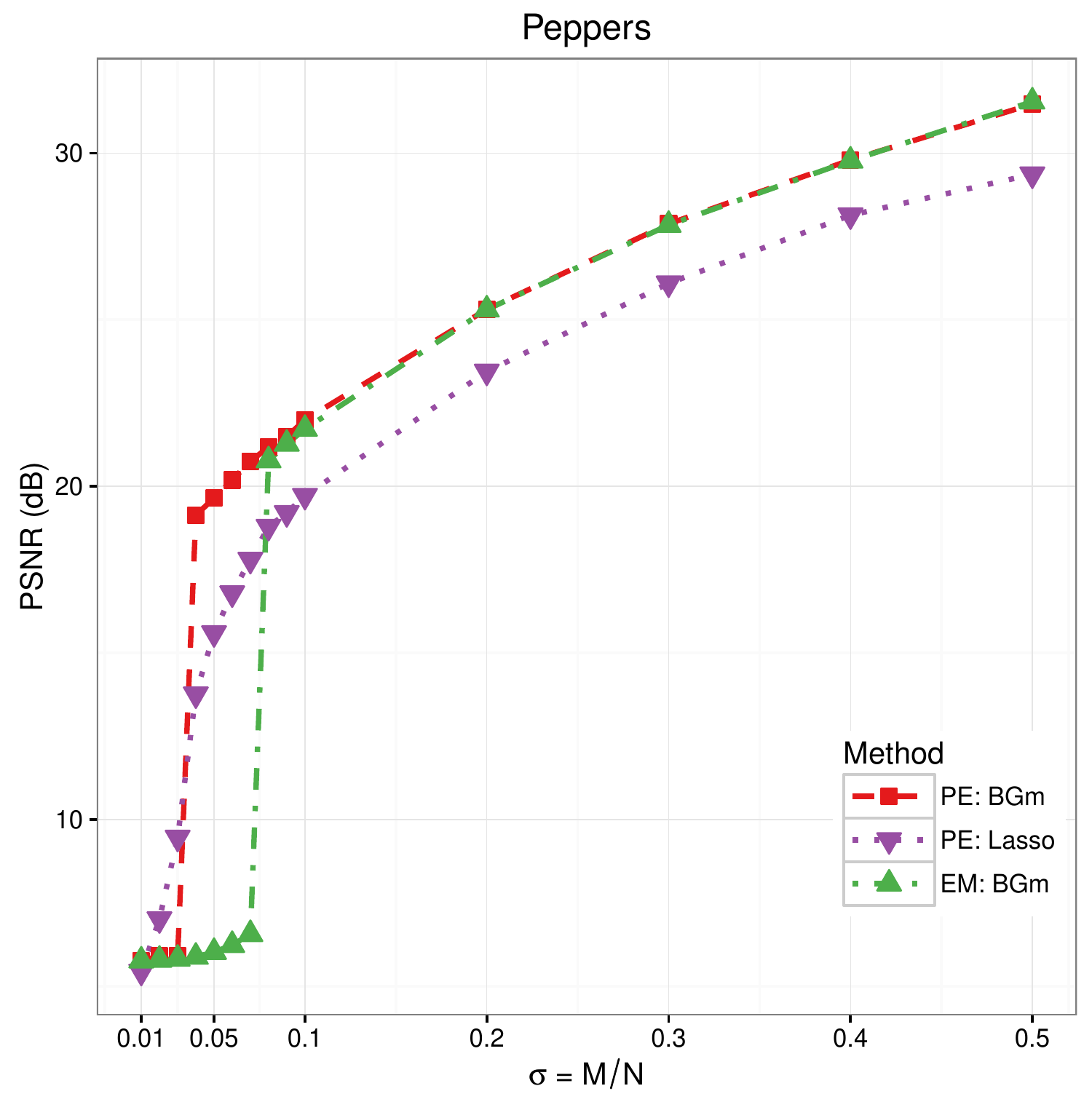}}
\caption{\small The peak-signal-to-noise-ratio (PSNR) of the recovered images from ``noisy'' measurements using different GAMP methods. (a) Barbara; (b) Boat; (c) Lena; (d) Peppers.} 
\label{fig:psnr_images}
\end{figure*}

\subsection{Noisy Sparse Signal Recovery}
We next try to recover the sparse signal $\vx$ from a noisy measurement vector $\vy$. Specifically, we fix $S=100,N=1000$ and increase the number of measurement $M$. $\vy\in\mathbb{R}^M$ is generated as follows:
\begin{align}
\vy=\vA\vx+\nu\vw\,,
\end{align}
where $\nu>0$ controls the amount of noise added to $\vy$, the entries of $\vw$ are i.i.d Gaussian $\mathcal{N}(0,1)$. We choose $\nu=0.05$ for the BG sparse signals and $\nu=0.1$ for the BE sparse signals. This creats a measurement $\vy$ with signal to noies ratio (SNR) around $20$ dB. We randomly generate $100$ triples of $\{\vx,\vA,\vw\}$. The average SNRs of the recovered signals $\hat{\vx}$ are shown in Fig. \ref{fig:snr_noisy}. 

In the noisy case, the oracle-GAMP performs the best as expected since the ``true'' parameters are used to recover the sparse signal, and the GAMP methods using estimated parameters are not bad either. For the BG sparse signals (Fig. \ref{fig:snr_noisy_bg}), we can see that PE-BGm-GAMP performs better than EM-BGm-GAMP when the sampling ratio is small. Since BGm is a better match than the Laplace prior, both PE-BGm-GAMP and EM-BGm-GAMP perform much better than PE-Lasso-GAMP. For the BE sparse signals (Fig. \ref{fig:snr_noisy_nonneg}), the BEm prior is a better match than the BGm prior. PE-BEm-GAMP is able to perform better than PE-BGm-GAMP and EM-BGm-GAMP，especially when the sampling ratio is small. Additionally, the solutions produced by PE-BEm-GAMP is guaranteed to be non-negative, while those by PE-BGm-GAMP and EM-BGm-GAMP generally contains negative coefficients. For applications that requires non-negative sparse solutions, such as hyperspectral unmixing \cite{SHyperUM11}, non-negative sparse coding for image classification \cite{ScSPM09}, etc, PE-BEm-GAMP offers a convenient way to solve the parameter estimation problem.

\subsection{Real Image Recovery}
Real images are considered to be approximately sparse under some proper basis, such as the DCT basis, wavelet basis, etc. Here we compare the recovery performances of PE-BGm-GAMP, PE-Lasso-GAMP, and EM-BGm-GAMP based on varying noisy measurements of the $4$ real images in Fig. \ref{fig:real_images}: Barbara, Boat, Lena, Peppers. We use the Daubechies $6$ (db6) wavelet \cite{DBWav92} as the sparsifying basis and i.i.d. random Gaussian matrix $\vA$ as the measurement matrix. The noise are generated using i.i.d. Gaussian distribution $\mathcal{N}(0,1)$, and the SNR of the measurement vector $\vy$ is around $30$ dB. The peak-signal-to-noise-ratio (PSNR) of the recovered images are shown in Fig. \ref{fig:psnr_images}. We can see that both PE-BGm-GAMP and EM-BGm-GAMP perform better than PE-Lasso-GAMP when the sampling ratio $\sigma>0.1$. When $\sigma$ is small, PE-BGm-GAMP and PE-Lasso-GAMP are more robust and generally perform better than EM-BGm-GAMP.
\begin{figure}[h]
\captionsetup[subfigure]{labelformat=empty}
\centering
\subfigure{
\includegraphics[width=0.75in]{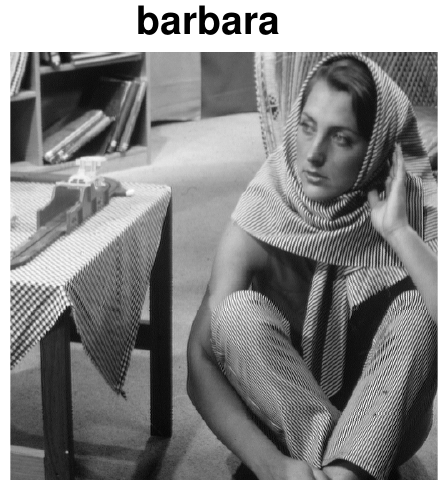}}
\subfigure{
\includegraphics[width=0.75in]{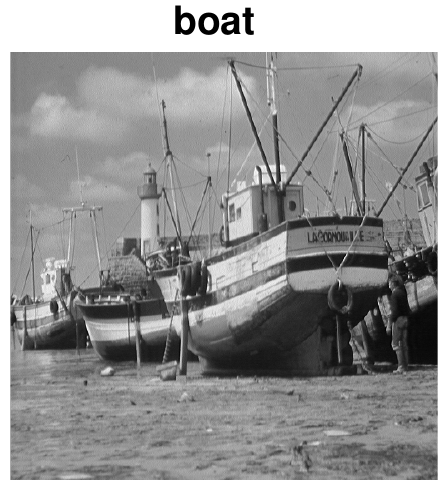}}
\subfigure{
\includegraphics[width=0.75in]{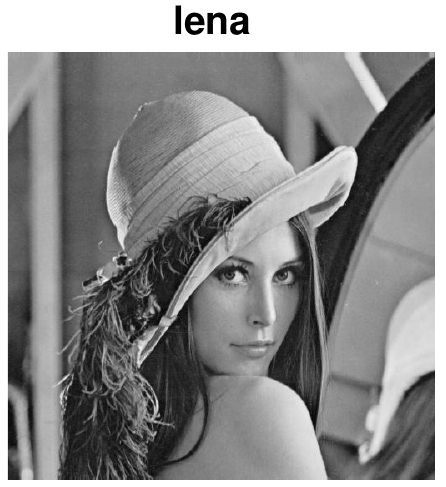}}
\subfigure{
\includegraphics[width=0.75in]{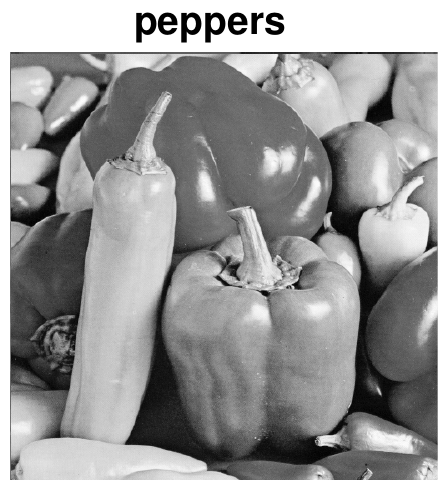}}
\caption{\small The real images used in the recovery experiments.} 
\label{fig:real_images}
\end{figure}

\subsection{Non-negative Sparse Coding for Image Classification}
The image classification task typically involves two steps: 1) extracting features, and 2) training a classifier based on such features. In the first step, low-level descriptors, such as SIFT \cite{SIFT99}, HOG \cite{HOG05}, etc, are extracted from local image patches, and then encoded to produce the high-level representations of the images, usually a vector $\vv\in\mathbb{R}^D$. Here we use the popular Bag-of-Words (BoW) model \cite{BOF03,BOF04} to encode the low level SIFT descriptors $\vy\in\mathbb{R}^M$. To do this, we first need to assign each $\vy$ to one or several ``visual words'' in some pre-trained dictionary/codebook $\vA$. In \cite{ScSPM09}, it is shown that this process can be formulated as a sparse coding problem:
\begin{align}
\begin{split}
\min_{\vx}&\quad\|\vy-\vA\vx\|_2^2\\
\textrm{subject to:}&\quad\vx\geq0,\quad\vx\textrm{ is sparse}.
\end{split}
\end{align}
where $\vx$ is the sparse code of $\vy$ in the dictionary $\vA$. In \cite{ScSPM09}, the sparsity constrain on $\vx$ is enforced with the $l_1$ norm regularization, i.e. Lasso. Both PE-BGm-GAMP and EM-BGm-GAMP can produce negative sparse codes, and are not suited for the task. Here we can use the proposed PE-BEm-GAMP to solve the above non-negative sparse coding problem.

Specifically, we perform image classification on the popular Caltech-101 dataset \cite{Caltech101}, which contains 9144 images belonging to $102$ classes (101 object classes and a background class). Following the suggestions of the original dataset \cite{Caltech101}, we randomly select 30 samples per class for training and \emph{up to} 50 samples per class for testing. This process is randomly repeated 10 times and the average classification accuracy is computed as the final result. 
\begin{figure}[h]
\centering
\includegraphics[width=2.5in]{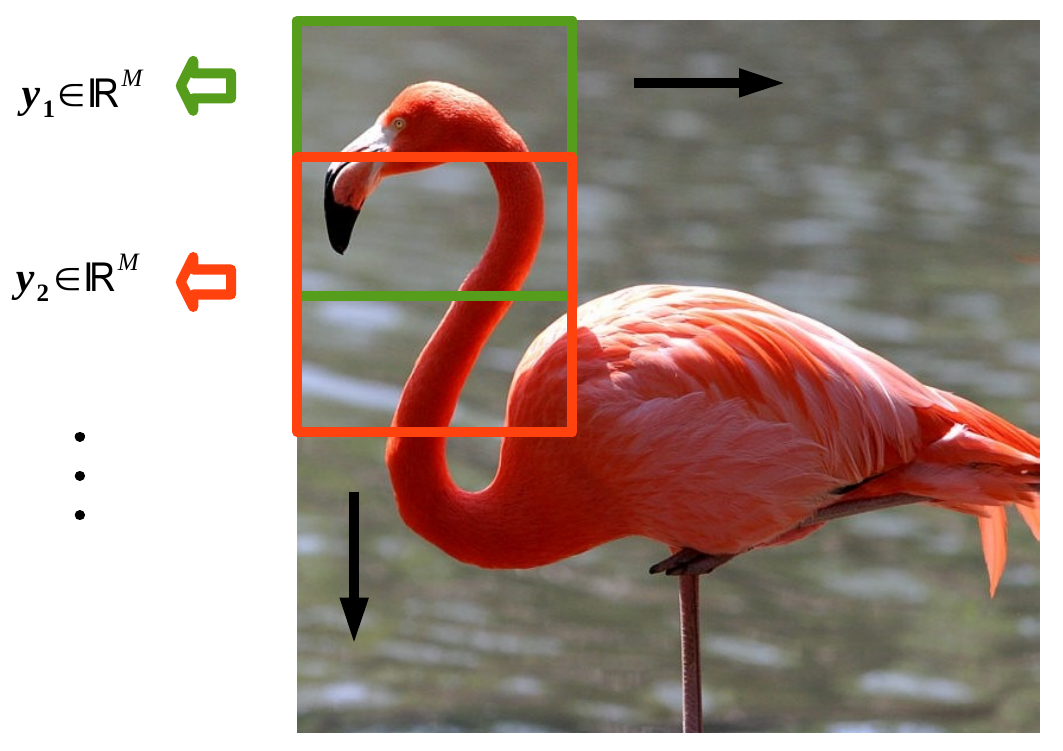}
\caption{\small Low-level SIFT features are densely sampled from local image patches.} 
\label{fig:dsift}
\end{figure}

Each image is converted to grayscale and resized to be no larger than $300\times300$ pixels while preserving the aspect ratio. The normalized local SIFT descriptors $\vy\in\mathcal{R}^{128}\geq0$ are extracted from $16\times16$ image patches densely sampled on the grid with a step size of $8$ pixels \cite{vlfeat}, as is shown in Fig. \ref{fig:dsift}. We use k-means \cite{Kmeans02} to train a $128\times1024$ normalized dictionary $\vA$. After the non-negative sparse coding, each local image patch is converted to a sparse vector $\vx\in\mathbb{R}^{1024}\geq0$. For each image, those sparse vectors are then \emph{max-pooled} using a $3$-level spatial pyramid matching \cite{SPM06} to produce a vector $\vv\in\mathbb{R}^{21504}$. As is usually done, linear support vector machine (SVM) \cite{SVM95,CC01a} is used as the classifier and the parameters of SVM are chosen using cross-validation. The average classification accuracy across all classes is $60.22\pm0.94\%$. 

\section{Conclusion and Future Work}
We proposed an approximate message passing algorithm with built-in parameter estimation (PE-GAMP) to recover under-sampled sparse signals. In the PE-GAMP framework, the parameters are treated as random variables with pre-specified priors, their posterior distributions can then be directly approximated by loopy belief propagation. This allows us to perform MAP and MMSE estimation of the parameters and update them during the message passing to recover sparse signals. Following the same assumptions made by the original GAMP \cite{GAMP11, AMP_CPE14}, state evolution analysis of the proposed PE-GAMP shows that it converges empirically. 

Compared with previous EM based parameter estimation methods, PE-GAMP could draw information from the prior distributions of the parameters. As is evident from both simulated and real experiments, PE-GAMP is also much simpler, more robust and perform better in the low sampling ratio settings. With its simpler formulation, PE-GAMP enjoys wider applicabilities and enables us to consider more complex signal distributions.

Here we mainly focused on MAP parameter estimation of the parameters. In the future, we would like to explore possible MMSE parameter estimation methods. From the non-negative sparse coding experiments, we observed that the proposed PE-GAMP was still able to achieve convergence even though the entries of the measurement matrix, i.e. dictionary, were not i.i.d Gaussian $\mathcal{N}(0,\frac{1}{M})$. Given this interesting observation, we would also like to investigate state evolution analysis for more generalized measurement matrices in our future work. 

\begin{appendices}
\counterwithin{assumption}{section}
\counterwithin{theorem}{section}
\renewcommand\thetable{\thesection\arabic{table}}
\renewcommand\thefigure{\thesection\arabic{figure}}

\section{PE-GAMP: Sum-product Message Passing}
\label{app:sp_message_passing}
Approximate message passing uses quadratic/Gaussian approximations of the messages from the variable nodes to the factor nodes to perform loopy belief propagation. To maintain consistency with \cite{GAMP11}, we use the same notations for the quadratic approximations of messages involving $\vx$. Specifically, $\Delta_{x_n}^{(t)}, \Delta_{\Phi_m\leftarrow x_n}^{(t)}$ in the $t$-th iteration can be used to construct the following distributions about $x_n$:
\begin{subequations}
\begin{align}
p(x_n|\vy)&\propto \exp(\Delta_{x_n}^{(t)})\\
p^{(t)}_{\Phi_m\leftarrow x_n}(x_n|\vy)&\propto\exp(\Delta_{\Phi_m\leftarrow x_n}^{(t)})\,.
\end{align}
\end{subequations}
We then have the following expectations and variances definitions:
\begin{subequations}
\label{eq:x_tau_x_def_app}
\begin{align}
\hat{x}_n^{(t)}&\coloneqq \mathbb{E}[x_n|\Delta_{x_n}^{(t)}]\\
\tau_{\Omega_n}^x(t)&\coloneqq\textrm{var}[x_n|\Delta_{x_n}^{(t)}]\\
\label{eq:x_phi_m_left_x_n}
\hat{x}_{\Phi_m\leftarrow x_n}^{(t)}&\coloneqq \mathbb{E}[x_n|\Delta_{\Phi_m\leftarrow x_n}^{(t)}]\\
\tau_{\Phi_m\leftarrow x_n}^x(t)&\coloneqq\textrm{var}[x_n|\Delta_{\Phi_m\leftarrow x_n}^{(t)}]\,.
\end{align}
\end{subequations}
If the entries $a_{mn}$ of the sensing matrix $A$ is small, $\tau_{\Omega_n}^x(t)\approx\tau_{\Phi_m\leftarrow x_n}^x(t)$. The message $\Delta_{\Phi_m\leftarrow x_n}^{(t)}$ in the $t$-th iteration will be approximated quadratically:
\begin{align}
\begin{split}
\Delta_{\Phi_m\leftarrow x_n}^{(t)} &\approx\textrm{const} -\frac{1}{2\tau_{\Phi_m\leftarrow x_n}^x(t)}\left(x_n-\hat{x}_{\Phi_m\leftarrow x_n}^{(t)}\right)^2\\
&\approx\textrm{const}-\frac{1}{2\tau_{\Omega_n}^x(t)}\left(x_n-\hat{x}_{\Phi_m\leftarrow x_n}^{(t)}\right)^2\,,
\end{split}
\end{align}
which makes the approximation of $p^{(t)}_{\Phi_m\leftarrow x_n}(x_n|\vy)$ a Gaussian distribution. Similarly we have the following approximations for $\vx$ involving the node $\Omega_n$:
\begin{subequations}
\label{eq:approx_lambda_theta_x_and_variances}
\begin{align}
\label{eq:x_omega_n_x_n_app}
\hat{x}_{\Omega_n\leftarrow x_n}^{(t)}&\coloneqq\mathbb{E}[x_n|\Delta_{\Omega_n\leftarrow x_n}^{(t)}]\\
\label{eq:tau_omega_n_x_n_app}
\tau_{\Omega_n\leftarrow x_n}^x(t)&\coloneqq\textrm{var}[x_n|\Delta_{\Omega_n\leftarrow x_n}^{(t)}]\\
\label{eq:delta_omega_n_x_n_app}
\Delta_{\Omega_n\leftarrow x_n}^{(t)}&\approx\textrm{const} -\frac{1}{2\tau_{\Omega_n\leftarrow x_n}^x(t)}\left(x_n-\hat{x}_{\Omega_n\leftarrow x_n}^{(t)}\right)^2\,.
\end{align}
\end{subequations}
In the proposed PE-GAMP, we use Dirac delta approximation of the messages involving the parameters $\boldsymbol\lambda,\boldsymbol\theta$. Specifically, the parameters are estimated using MAP or MMSE estimations:
\begin{enumerate}
\item MAP estimation:
\begin{subequations}
\label{eq:map_app}
\begin{align}
\label{eq:map_lambda_omega_n_lambda_l_app}
\hat{\lambda}_{\Omega_n\leftarrow\lambda_l}^{(t)}&\coloneqq\arg\max_{\lambda_l}\Delta_{\Omega_n\leftarrow\lambda_l}^{(t)}\\
\label{eq:map_theta_phi_m_theta_k_app}
\hat{\theta}_{\Phi_m\leftarrow\theta_k}^{(t)}&\coloneqq\arg\max_{\theta_k}\Delta_{\Phi_m\leftarrow\theta_k}^{(t)}
\end{align}
\end{subequations}
\item MMSE estimation:
\begin{subequations}
\label{eq:mmse_app}
\begin{align}
\label{eq:mmse_lambda_omega_n_lambda_l_app}
\hat{\lambda}_{\Omega_n\leftarrow\lambda_l}^{(t)}&\coloneqq\mathbb{E}[\lambda_l|\Delta_{\Omega_n\leftarrow\lambda_l}^{(t)}]\\
\label{eq:mmse_theta_phi_m_theta_k_app}
\hat{\theta}_{\Phi_m\leftarrow\theta_k}^{(t)}&\coloneqq\mathbb{E}[\theta_k|\Delta_{\Phi_m\leftarrow\theta_k}^{(t)}]
\end{align}
\end{subequations}
\end{enumerate}
The corresponding messages involving the parameters $\boldsymbol\lambda,\boldsymbol\theta$ in the $(t)$-th iteration can then be approximated as follows:
\begin{subequations}
\label{eq:approx_messages_variable_factor}
\begin{align}
\exp\left(\Delta_{\Omega_n\leftarrow\lambda_l}^{(t)}\right)&\approx\delta\left(\lambda_l-\hat{\lambda}_{\Omega_n\leftarrow\lambda_l}^{(t)}\right)\\
\exp\left(\Delta_{\Phi_m\leftarrow\theta_k}^{(t)}\right)&\approx\delta\left(\theta_k-\hat{\theta}_{\Phi_m\leftarrow\theta_k}^{(t)}\right)
\end{align}
\end{subequations}
Using approximated messages from the variable node to factor node in (\ref{eq:approx_messages_variable_factor}), $\Delta_{\Phi_m\rightarrow x_n}^{(t+1)}$ can then be computed:
\begin{align}
\label{eq:sp_fv_phi_x_app1}
\begin{split}
&\Delta_{\Phi_m\rightarrow x_n}^{(t+1)}=\textrm{const}\\
&\quad+\log\int_{\vx\backslash x_n}\Phi_m\left(y_m,\vx,\hat{\boldsymbol\theta}_{\Phi_m}^{(t)}\right)\cdot\exp\left(\textstyle\sum_{j\neq n}\Delta_{\Phi_m\leftarrow x_j}^{(t)}\right).
\end{split}
\end{align}
Direct integration with respect to $\vx\backslash x_n$ in (\ref{eq:sp_fv_phi_x_app1}) is quite difficult. If we go back to the original belief propagation, we can see that the message $\Delta_{\Phi_m\rightarrow x_n}^{(t+1)}$ essentially performs the following computation:
\begin{align}
\begin{split}
\log p(y_m,x_n)&=\log\int_{\vx\backslash x_n,\boldsymbol\theta}p(y_m,\vx,\theta)\\
&=\log\int_{\vx\backslash x_n, \boldsymbol\theta}p(y_m|\vx, \boldsymbol\theta)p(\vx)p(\boldsymbol\theta)\,.
\end{split}
\end{align}
Let $z_m^\prime=z_m-a_{mn}x_n=\sum_{j\neq n}a_{mj}x_j$, $\log p(y_m,x_n)$ can also be written as:
\begin{align}
\label{eq:p_y_xn_app}
\begin{split}
\log p(y_m, x_n)&=\log\int_{z_m^\prime,\boldsymbol\theta}p(y_m, x_n, z_m^\prime,\boldsymbol\theta )\\
&=\log\int_{z_m^\prime, \boldsymbol\theta}p(y_m|x_n,z_m^\prime,\theta)p(z_m^\prime)p(\boldsymbol\theta)\,.
\end{split}
\end{align}
Translating (\ref{eq:p_y_xn_app}) back to the message gives us:
\begin{align}
\label{eq:p_y_xn_zm_prime}
\begin{split}
\Delta_{\Phi_m\rightarrow x_n}^{(t+1)}&=\textrm{const}+\log\int_{z_m^\prime}\left[\vphantom{-\frac{1}{2\left(\tau^q_{\Phi_m}(t)-a_{mn}^2\tau_{\Omega_n}^x(t)\right)}}\Phi\left(y_m, x_n, z_m^\prime,\hat{\boldsymbol\theta}_{\Phi_m}^{(t)}\right)\right.\\
&\quad\times\exp\left(-\frac{1}{2\left(\tau^q_{\Phi_m}(t)-a_{mn}^2\tau_{\Omega_n}^x(t)\right)}\right.\\
&\quad\times\left.\left.\left(z_m^\prime-\left(q_{\Phi_m}^{(t)}-a_{mn}\hat{x}_{\Phi_m\leftarrow x_n}^{(t)}\right)\right)^2\vphantom{-\frac{1}{2\left(\tau^q_{\Phi_m}(t)-a_{mn}^2\tau_{\Omega_n}^x(t)\right)}}\right)\right]\,,
\end{split}
\end{align}
where $\tau_{\Phi_m}^q(t), q_{\Phi_m}^{(t)}$ are as follows:
\begin{subequations}
\begin{align}
\tau_{\Phi_m}^q(t)&=\sum_ja_{mj}^2\tau_{\Omega_j}^x(t)\\
q^{(t)}_{\Phi_m} &= \sum_ja_{mj}\hat{x}^{(t)}_{\Phi_m\leftarrow x_j}\,.
\end{align}
\end{subequations}
If $a_{mn}$ is small, $a_{mn}^2\tau_{\Omega_n}^x(t)$ can be neglected. Since the integration of $z_m^\prime$ is from $-\infty$ to $\infty$, we replace $z_m^\prime$ with $z_m = z_m^\prime+a_{mn}x_n$. (\ref{eq:p_y_xn_zm_prime}) then becomes:
\begin{align}
\label{eq:p_y_z_app}
\begin{split}
&\Delta_{\Phi_m\rightarrow x_n}^{(t+1)} = \textrm{const}+\log\int_{z_m}\left[\Phi\left(y_m,z_m,\hat{\boldsymbol\theta}_{\Phi_m}^{(t)}\right)\vphantom{\frac{-1}{2\tau_{\Phi_m}^q(t)}}\times\right.\\
&\left.\exp\left(\frac{-1}{2\tau_{\Phi_m}^q(t)}\left(z_m-\left(q_{\Phi_m}^{(t)}+a_{mn}\left(x_n-\hat{x}_{\Phi_m\leftarrow x_n}^{(t)}\right)\right)\right)^2\right)\right]
\end{split}
\end{align}

\subsection{GAMP Update}
For completeness, we include the GAMP update from \cite{GAMP11} to compute $\Delta_{\Phi_m\rightarrow x_n}^{(t+1)}, \Delta_{\Phi_m\leftarrow x_n}^{(t+1)}$. The following function $H\left(q, \tau^q, y, \boldsymbol\theta\right)$ is defined:
\begin{align}
H\left(q, \tau^q, y, \boldsymbol\theta\right) = \log\int_{z}\Phi(y,z,\boldsymbol\theta)\cdot\exp\left(-\frac{1}{2\tau^q}\left(z-q\right)^2\right)\,.
\end{align}
$\Delta_{\Phi_m\rightarrow x_n}^{(t+1)}$ in (\ref{eq:p_y_z_app}) can then be written as:
\begin{align}
\begin{split}
&\Delta_{\Phi_m\rightarrow x_n}^{(t+1)}=\textrm{const}\\
&+H\left(q_{\Phi_m}^{(t)}+a_{mn}\left(x_n-\hat{x}_{\Phi_m\leftarrow x_n}^{(t)}\right), \tau_{\Phi_m}^q(t), y_m, \hat{\boldsymbol\theta}_{\Phi_m}^{(t)}\right)\,.
\end{split}
\end{align}
Next, we try to approximate the message $\Delta_{\Phi_m\rightarrow x_n}^{(t+1)}$ up to second order Taylor series at $q_{\Phi_m}^{(t)}$. We define the following:
\begin{align}
g_\textrm{out}(q,\tau^q,y,\boldsymbol\theta)\coloneqq\frac{\partial}{\partial q}H\left(q, \tau^q, y,\boldsymbol\theta\right)\,.
\end{align}
Let $s_{\Phi_m}^{(t)}, \tau_{\Phi_m}^s(t)$ be the first and second order of $H(\cdot)$ at $q_{\Phi_m}^{(t)}$:
\begin{align}
s_{\Phi_m}^{(t)}&=g_\textrm{out}\left(t,q_{\Phi_m}^{(t)},\tau_{\Phi_m}^q(t),y_m,\hat{\boldsymbol\theta}_{\Phi_m}^{(t)}\right)\\
\tau_{\Phi_m}^s(t)&=-\frac{\partial}{\partial q}g_\textrm{out}\left(t,q_{\Phi_m}^{(t)},\tau_{\Phi_m}^q(t),y_m,\hat{\boldsymbol\theta}_{\Phi_m}^{(t)}\right)\,.
\end{align}
$\Delta_{\Phi_m\rightarrow x_n}^{(t+1)}$ can then be approximated by:
\begin{align}
\label{eq:delta_phi_m_x_n_app}
\begin{split}
\Delta_{\Phi_m\rightarrow x_n}^{(t+1)} &\approx\textrm{const}+ s_{\Phi_m}^{(t)}a_{mn}\left(x_n-\hat{x}_{\Phi_m\leftarrow x_n}^{(t)}\right)\\
&\quad\quad-\frac{\tau_{\Phi_m}^s(t)}{2}a^2_{mn}\left(x_n-\hat{x}_{\Phi_m\leftarrow x_n}^{(t)}\right)^2\,.
\end{split}
\end{align}
$\Delta_{\Phi_m\leftarrow x_n}^{(t+1)}$ will then be computed as is done in \cite{GAMP11}:
\begin{subequations}
\begin{align}
\begin{split}
\Delta_{\Phi_m\leftarrow x_n}^{(t+1)}&\approx\textrm{const}+\Delta_{\Omega_n\rightarrow x_n}^{(t+1)}\\
&\quad\quad-\frac{1}{2\tau_{\Phi_m\leftarrow x_n}^r(t)}\left(r_{\Phi_m\leftarrow x_n}^{(t)}-x_n\right)^2
\end{split}\\
\label{eq:tau_phi_m_left_x_n}
\tau_{\Phi_m\leftarrow x_n}^r(t)&=\left(\sum_{i\neq m}a_{in}^2\tau_{\Phi_i}^s(t)\right)^{-1}\\
\label{eq:r_phi_m_left_x_n}
r_{\Phi_m\leftarrow x_n}^{(t)}&=\hat{x}^{(t)}_{\Phi_m\leftarrow x_n}+\tau_{\Phi_m\leftarrow x_n}^r(t)\sum_{i\neq m}s_{\Phi_i}^{(t)}a_{in}\,.
\end{align}
\end{subequations}
(\ref{eq:tau_phi_m_left_x_n}, \ref{eq:r_phi_m_left_x_n}) are approximated:
\begin{subequations}
\begin{align}
\tau_{\Phi_m\leftarrow x_n}^r(t)&\approx \tau_{\Omega_n}^r(t) = \left(\sum_ia_{in}^2\tau_{\Phi_i}^s(t)\right)^{-1}\\
\begin{split}
r_{\Phi_m\leftarrow x_n}^{(t)}&\approx \left(\hat{x}_{\Phi_m}^{(t)} + \tau_{\Omega_n}^r(t)\sum_is_{\Phi_i}^{(t)}a_{in}\right) -\tau_{\Omega_n}^r(t)a_{mn}s_{\Phi_m}^{(t)}\\
&= r_{\Omega_n}^{(t)}-\tau_{\Omega_n}^r(t)a_{mn}s_{\Phi_m}^{(t)}\,.
\end{split}
\end{align}
\end{subequations}
The following definition is also made in \cite{GAMP11}:
\begin{align}
\begin{split}
&g_\textrm{in}\left(r_{\Omega_n}^{(t)},\tau_{\Omega_n}^r(t),\hat{\boldsymbol\lambda}_{\Omega_n}^{(t)}\right)\coloneqq\\
&\quad\frac{\int_{x_n}x_n\exp\left(\Delta_{\Omega_n\rightarrow x_n}^{(t)}-\frac{1}{2\tau_{\Omega_n}^r(t)}\left(r_{\Omega_n}^{(t)}-x_n\right)^2\right)}{\int_{x_n}\exp\left(\Delta_{\Omega_n\rightarrow x_n}^{(t)}-\frac{1}{2\tau_{\Omega_n}^r(t)}\left(r_{\Omega_n}^{(t)}-x_n\right)^2\right)}\,.
\end{split}
\end{align}
$\hat{x}_{\Phi_m\leftarrow x_n}^{(t+1)}, \hat{x}_n^{(t+1)}, q_{\Phi_m}^{(t+1)}$ are then \cite{GAMP11}:
\begin{subequations}
\begin{align}
\hat{x}_{\Phi_m\leftarrow x_n}^{(t+1)}&= g_{\textrm{in}}\left(r_{\Phi_m\leftarrow x_n}^{(t)},\tau_{\Phi_m\leftarrow x_n}^r(t),\hat{\boldsymbol\lambda}_{\Omega_n}^{(t)}\right)\\
\hat{x}_n^{(t+1)}&= g_\textrm{in}\left(r_{\Omega_n}^{(t)},\tau_{\Omega_n}^r(t),\hat{\boldsymbol\lambda}_{\Omega_n}^{(t)}\right)\\
q_{\Phi_m}^{(t+1)}&\approx\sum_ja_{mj}\hat{x}_j^{(t+1)}-\tau_{\Phi_m}^q(t)s_{\Phi_m}^{(t)}\,.
\end{align}
\end{subequations}

\subsection{Parameter Update}
Similarly we can compute the rest messages from the factor nodes to variable nodes in the proposed PE-GAMP using Dirac delta approximation of the messages involving the parameters:
\begin{subequations}
\label{eq:factor_to_variable_app}
\begin{align}
\label{eq:factor_to_variable_x_app}
&\Delta_{\Omega_n\rightarrow x_n}^{(t+1)}=\textrm{const}+\log\Omega_n\left(x_n,\hat{\boldsymbol\lambda}_{\Omega_n}^{(t)}\right)\\
\begin{split}
\label{eq:factor_to_variable_lambda_app}
&\Delta_{\Omega_n\rightarrow\lambda_l}^{(t+1)}=\textrm{const}+\log\int_{x_n}\left[\vphantom{-\sum_{u\neq l}\frac{1}{2\tau_{\Omega_n\leftarrow\lambda_u}^\lambda(t)}\left(\lambda_u-\hat{\lambda}_{\Omega_n\leftarrow\lambda_u}^{(t)}\right)^2}\Omega_n\left(x_n,\lambda_l,\hat{\boldsymbol\lambda}_{\Omega_n}^{(t)}\backslash\hat{\lambda}_{\Omega_n\leftarrow\lambda_l}^{(t)}\right)\right.\\
&\quad\times\left.\exp\left(\vphantom{-\sum_{u\neq l}\frac{1}{2\tau_{\Omega_n\leftarrow\lambda_u}^\lambda(t)}\left(\lambda_u-\hat{\lambda}_{\Omega_n\leftarrow\lambda_u}^{(t)}\right)^2}-\frac{1}{2\tau^x_{\Omega_n\leftarrow x_n}(t+1)}\left(x_n-\hat{x}_{\Omega_n\leftarrow x_n}^{(t+1)}\right)^2\right)\right]
\end{split}\\
\begin{split}
\label{eq:factor_to_variable_theta_app}
&\Delta_{\Phi_m\rightarrow\theta_k}^{(t+1)}=\textrm{const}+\log\int_{z_m}\left[\vphantom{-\sum_{u\neq l}\frac{1}{2\tau_{\Omega_n\leftarrow\lambda_u}^\lambda(t)}\left(\lambda_u-\hat{\lambda}_{\Omega_n\leftarrow\lambda_u}^{(t)}\right)^2}\Phi\left(y_m,z_m,\theta_k,\hat{\boldsymbol\theta}_{\Phi_m}^{(t)}\backslash\hat{\theta}_{\Phi_m\leftarrow\theta_k}^{(t)}\right)\right.\\
&\quad\times\left.\exp\left(\vphantom{-\sum_{u\neq l}\frac{1}{2\tau_{\Omega_n\leftarrow\lambda_u}^\lambda(t)}\left(\lambda_u-\hat{\lambda}_{\Omega_n\leftarrow\lambda_u}^{(t)}\right)^2}-\frac{1}{2\tau_{\Phi_m}^q(t)}\left(z_m-q_{\Phi_m}^{(t)}\right)^2\right)\right]\,.
\end{split}
\end{align}
\end{subequations}
We next compute (\ref{eq:approx_lambda_theta_x_and_variances}) in the $(t+1)$-th iteration. Using (\ref{eq:delta_phi_m_x_n_app}) we can get $\Delta_{\Omega_n\leftarrow x_n}^{(t+1)}$ in (\ref{eq:delta_omega_n_x_n_app}) first:
\begin{align}
\Delta_{\Omega_n\leftarrow x_n}^{(t+1)}\approx -\frac{1}{2\tau_{\Omega_n}^r(t)}\left(x_n-r_{\Omega_n}^{(t)}\right)^2\,.
\end{align}
(\ref{eq:x_omega_n_x_n_app},\ref{eq:tau_omega_n_x_n_app}) in the $(t+1)$-th iteration are then:
\begin{align}
\hat{x}_{\Omega_n\leftarrow x_n}^{(t+1)} = r_{\Omega_n}^{(t)},\quad\quad\quad\tau_{\Omega_n\leftarrow x_n}^x(t+1)=\tau_{\Omega_n}^r(t)\,.
\end{align}
The parameteres $\hat{\boldsymbol\lambda}_{\Omega_n}^{(t+1)},\hat{\boldsymbol\theta}_{\Phi_m}^{(t+1)}$ in the $(t+1)$-th iteration can then be computed using (\ref{eq:map_app}) or (\ref{eq:mmse_app}).

\section{PE-GAMP: Max-sum Message Passing}
\label{app:ms_message_passing}
The approximated \emph{max-sum} message passing also uses quadratic approximation of the messages. It is in many ways similar to the \emph{sum-product} message passing presented previously in Appendix \ref{app:sp_message_passing}. A few differences in do exists though. Specifically, the definitions in (\ref{eq:x_tau_x_def_app}) are changed into:
\begin{subequations}
\begin{align}
\hat{x}_n^{(t)}&\coloneqq\arg\max_{x_n}\Delta_{x_n}^{(t)}\\
\tau_{\Omega_n}^x(t)&\coloneqq-\left(\frac{\partial^2 \Delta_{x_n}^{(t)}}{\partial x_n^2}\left|_{x_n=\hat{x}_n^{(t)}}\right.\right)^{-1}\\
\hat{x}_{\Phi_m\leftarrow x_n}^{(t)}&\coloneqq\arg\max_{x_n}\Delta_{\Phi_m\leftarrow x_n}^{(t)}\\
\tau_{\Phi_m\leftarrow x_n}^x(t)&\coloneqq-\left(\frac{\partial^2 \Delta_{\Phi_m\leftarrow x_n}^{(t)}}{\partial x_n^2}\left|_{x_n=\hat{x}_{\Phi_m\leftarrow x_n}^{(t)}}\right.\right)^{-1}
\end{align}
\end{subequations}

In the proposed PE-GAMP, the parameters are computed as follows:
\begin{subequations}
\label{eq:ms_pe_map_app}
\begin{align}
\begin{split}
\label{eq:ms_pe_map_app_lambda}
\hat{\lambda}_{\Omega_n\leftarrow\lambda_l}^{(t)}=\arg\max_{\lambda_l}\,&\log\Omega_n\left(\hat{x}_n^{(t-1)},\lambda_l,\hat{\boldsymbol\lambda}_{\Omega_n}^{(t-1)}\backslash\hat{\lambda}_{\Omega_n\leftarrow\lambda_l}^{(t-1)}\right)\\
&\quad + \Delta_{\Omega_n\leftarrow\lambda_l}^{(t-1)}
\end{split}\\
\begin{split}
\label{eq:ms_pe_map_app_theta}
\hat{\theta}_{\Phi_m\leftarrow\theta_k}^{(t)}=\arg\max_{\theta_k}\,&\log\Phi_m\left(y_m,\hat{\vx}^{(t-1)},\theta_k,\hat{\boldsymbol\theta}_{\Phi_m}^{(t-1)}\backslash\hat{\theta}_{\Phi_m\leftarrow\theta_k}^{(t-1)}\right)\\
&\quad + \Delta_{\Phi_m\leftarrow\theta_k}^{(t-1)}\,.
\end{split}
\end{align}
\end{subequations}

\subsection{GAMP Update}
The definitions of $H(q,\tau^q,y,\boldsymbol\theta)$ and the input function $g_\textrm{in}(\cdot)$ are also different from \emph{sum-product} message passing. \cite{GAMP11} has the following definitions:
\begin{align}
&H(q,\tau^q,y)=\max_z\left[\log\Phi(y,z,\boldsymbol\theta)-\frac{1}{2\tau^q}(z-q)^2\right]\\
\begin{split}
&g_\textrm{out}\left(r_{\Omega_n}^{(t)},\tau_{\Omega_n}^r(t),\hat{\boldsymbol\lambda}_{\Omega_n}^{(t)}\right)=\\
&\quad\quad\arg\max_{x_n}\left[\Delta_{\Omega_n\rightarrow x_n}^{(t)}-\frac{1}{2\tau_{\Omega_n}^r(t)}\left(r_{\Omega_n}^{(t)}-x_n\right)^2\right]\,.
\end{split}
\end{align}
The message $\Delta_{\Phi_m\rightarrow x_n}^{(t+1)}$ is also different from (\ref{eq:p_y_z_app}) in Appendix \ref{app:sp_message_passing}. In \cite{GAMP11}, it is given as follows:
\begin{align}
\begin{split}
&\Delta_{\Phi_m\rightarrow x_n}^{(t+1)}\approx\max_{z_m}\left[\vphantom{-\frac{1}{2\tau_{\Phi_m}^q(t)}}\log\Phi\left(z_m,y_m,\hat{\boldsymbol\lambda}_{\Omega_n}^{(t)}\right)\right.\\
&\quad\quad\left.-\frac{1}{2\tau_{\Phi_m}^q(t)}\left(z_m-\left(q_{\Phi_m}^{(t)}+a_{mn}\left(x_n-\hat{x}_{\Phi_m\leftarrow x_n}^{(t)}\right)\right)\right)^2\right]\,.
\end{split}
\end{align}

\subsection{Parameter Update}
The messages in (\ref{eq:factor_to_variable_app}) are also updated:
\begin{subequations}
\label{eq:factor_to_variable_ms_app}
\begin{align}
\label{eq:factor_to_variable_lambda_ms_app}
\begin{split}
\Delta_{\Omega_n\rightarrow x_n}^{(t+1)}&=\log\Omega_n(x_n,\hat{\boldsymbol\lambda}_{\Omega_n}^{(t)})
\end{split}\\
\label{eq:factor_to_variable_x_ms_app}
\begin{split}
\Delta_{\Omega_n\rightarrow\lambda_l}^{(t+1)}&=\max_{x_n}\left[\vphantom{-\sum_{u\neq l}\frac{1}{2\tau_{\Omega_n\leftarrow\lambda_u}^\lambda(t)}\left(\lambda_u-\hat{\lambda}_{\Omega_n\leftarrow\lambda_u}^{(t)}\right)^2}\log\Omega_n\left(x_n,\lambda_l,\hat{\boldsymbol\lambda}_{\Omega_n}^{(t)}\backslash\hat{\lambda}_{\Omega_n\leftarrow\lambda_l}^{(t)}\right)\right.\\
&\quad\quad\left.-\frac{1}{2\tau^r_{\Omega_n}(t)}\left(x_n-r_{\Omega_n}^{(t)}\right)^2\right]
\end{split}\\
\label{eq:factor_to_variable_theta_ms_app}
\begin{split}
\Delta_{\Phi_m\rightarrow\theta_k}^{(t+1)}&=\max_{z_m}\left[\vphantom{-\sum_{u\neq l}\frac{1}{2\tau_{\Omega_n\leftarrow\lambda_u}^\lambda(t)}\left(\lambda_u-\hat{\lambda}_{\Omega_n\leftarrow\lambda_u}^{(t)}\right)^2}\log\Phi\left(y_m,z_m,\theta_k\hat{\boldsymbol\theta}_{\Phi_m}^{(t)}\backslash\hat{\theta}_{\Phi_m\leftarrow\theta_k}^{(t)}\right)\right.\\
&\quad\quad\left.-\frac{1}{2\tau_{\Phi_m}^q(t)}\left(z_m-q_{\Phi_m}^{(t)}\right)^2\right]\,.
\end{split}
\end{align}
\end{subequations}
The parameters $\hat{\boldsymbol\lambda}_{\Omega_n}^{(t+1)},\hat{\boldsymbol\theta}_{\Phi_m}^{(t+1)}$ in the $(t+1)$-th iteration can then be computed using (\ref{eq:ms_pe_map_app}).

\section{State Evolution Analysis of Adaptive-GAMP} 
\label{app:se_adaptive_gamp}
In \cite{AMP_CPE14}, the following Assumption \ref{aspt:adaptive_gamp} are made about the estimation problem:
\smallskip
\begin{assumption}
\label{aspt:adaptive_gamp}
The adaptive-GAMP with parameter estimation solves a series of estimation problem indexed by the input signal dimension $N$:
\begin{enumerate}[label={\alph*)}, nolistsep]
\item Assumptions \ref{aspt_gamp}(a) to \ref{aspt_gamp}(d) with $k=2$.
\item Assumption \ref{aspt_pe_gamp}(b).
\item For every $t$, the estimation (adaptation) function $f_{\boldsymbol\lambda}(t, \vr_\Omega^{(t)}, \tau_\Omega^r(t))$ can be considered as a function of $\vr_\Omega^{(t)}$ that satisfies the weak pseudo-Lipschitz continuity property: If the sequence of vector $\vr_\Omega^{(t)}$  indexed by $N$ empirically converges with bounded moments of order $k=2$ and the sequence of scalars $\tau_\Omega^r(t)$ converge as follows:
\begin{align}
\lim_{N\rightarrow\infty}\vr_\Omega^{(t)}\stackrel{\normalfont\textrm{PL$(k)$}}{=}\mathcal{R}_\Omega^{(t)},\quad\lim_{M\rightarrow\infty}\tau_\Omega^r(t)=\overline{\tau}_\Omega^r(t)\,.
\end{align}
Then,
\begin{align}
\lim_{N\rightarrow\infty}f_{\boldsymbol\lambda}(t, \vr_\Omega^{(t)}, \tau_\Omega^r(t)) = f_{\boldsymbol\lambda}(t, \mathcal{R}_\Omega^{(t)}, \overline{\tau}_\Omega^r(t))\,.
\end{align}
Similarly $f_{\boldsymbol\theta}(t,\vq_\Phi^{(t)},\vy,\tau_\Phi^q(t))$ also satisfies the weak pseudo-Lipschitz continuity property.
\end{enumerate}
\end{assumption}
Theorem \ref{thm:adaptive_gamp} is then given to describe the limiting behavior of the scalar variables in the adaptive-GAMP algorithm \cite{AMP_CPE14}.
\smallskip
\begin{theorem}
\label{thm:adaptive_gamp}
Consider the adaptive-GAMP with scalar variances under the Assumption \ref{aspt:adaptive_gamp}. $\forall t$, the components of the following sets of scalars empirically converges with bounded moments of order $k=2$:
\begin{subequations}
\begin{align}
&\lim_{N\rightarrow\infty}\boldsymbol\psi_\textrm{in}\stackrel{\normalfont\textrm{PL$(k)$}}{=}\overline{\boldsymbol\psi}_\textrm{in},\quad\lim_{N\rightarrow\infty}\boldsymbol\psi_\textrm{out}\stackrel{\normalfont\textrm{PL$(k)$}}{=}\overline{\boldsymbol\psi}_\textrm{out}\\
&\lim_{N\rightarrow\infty}\boldsymbol\psi_\tau=\overline{\boldsymbol\psi}_\tau\\
&\lim_{N\rightarrow\infty}\hat{\boldsymbol\theta}^{(t+1)}=\overline{\boldsymbol\theta}^{(t+1)},\quad\lim_{N\rightarrow\infty}\hat{\boldsymbol\lambda}^{(t+1)}=\overline{\boldsymbol\lambda}^{(t+1)}\,.
\end{align}
\end{subequations}
\end{theorem}

\section{Proof of Lemma \ref{lemma:map_pl_continuity}}
\label{app:proof_lemma}
Here we give the proof for $f_{\Omega_n\leftarrow\lambda_l}(\cdot)$, the proof for $f_{\Phi_m\leftarrow\theta_k}(\cdot)$ can be derived similarly.
\medskip
\begin{enumerate}
\item {\bfseries MAP Parameter Estimation:}
The proof of the continuity of $f_{\Omega_n\leftarrow\lambda_l}(\cdot), f_{\Phi_m\leftarrow\theta_k}(\cdot)$ is adapted from the work in \cite{AMP_CPE14}. In the $t$-th iteration, the following estimation indexed by singal dimensionality $N$ can be computed:
\begin{align}
\hat{\lambda}_{\Omega_n\leftarrow\lambda_l}^{(t+1)}[N] = f_{\Omega_n\leftarrow\lambda_l}\left(t, \vr_{\Omega}^{(t)}, \tau_{\Omega}^r(t), \lambda_l, \hat{\boldsymbol\lambda}_{\Omega_n}^{(t)}\backslash\hat{\lambda}_{\Omega_n\leftarrow\lambda_l}^{(t)}\right)\,.
\end{align} 
We then have a sequence $\{\hat{\lambda}_{\Omega_n\leftarrow\lambda_l}^{(t+1)}[N]\}$ indexed by $N=2,3,\cdots$. Since $\hat{\lambda}_{\Omega_n\leftarrow\lambda_l}^{(t+1)}[N]\in\mathcal{U}_\lambda$ and $\mathcal{U}_\lambda$ is compact, it suffices to show that any sequence $\{\hat{\lambda}_{\Omega_n\leftarrow\lambda_l}^{(t+1)}[N]\}$ converges to the same limiting point $\lambda_l^*$ shown in (\ref{eq:map_lambda_detail}). According to (\ref{eq:map_lambda_h}), we have:
\begin{align}
\begin{split}
&\sum_{j\neq n}h_{\Omega_n\leftarrow\lambda_l}^{\Omega_j}(t, \vr_{\Omega_n}^{(t)}, \tau_\Omega^r(t),\hat{\lambda}_{\Omega_n\leftarrow\lambda_l}^{(t+1)}[N], \hat{\boldsymbol\lambda}_{\Omega_n}^{(t)}\backslash\hat{\lambda}_{\Omega_n\leftarrow\lambda_l}^{(t)})\\
&\quad\geq \sum_{j\neq n}h_{\Omega_n\leftarrow\lambda_l}^{\Omega_j}(t, \vr_{\Omega_n}^{(t)}, \tau_\Omega^r(t),\lambda_l^*, \hat{\boldsymbol\lambda}_{\Omega_n}^{(t)}\backslash\hat{\lambda}_{\Omega_n\leftarrow\lambda_l}^{(t)})\,.
\end{split}
\end{align}
Suppose that $\hat{\lambda}_{\Omega_n\leftarrow\lambda_l}^{(t+1)}[N]$ converges to some point $\hat{\lambda}_{\Omega_n\leftarrow\lambda_l}^{(t+1)}$: $\hat{\lambda}_{\Omega_n\leftarrow\lambda_l}^{(t+1)}[N]\rightarrow\hat{\lambda}_{\Omega_n\leftarrow\lambda_l}^{(t+1)}$ as $N\rightarrow\infty$. With (\ref{eq:rv_convergence}) and the continuity condition of the open neighborhood $\rho(\widetilde{\tau}_\Omega^r, \widetilde{\boldsymbol\lambda})$ in Assumption \ref{aspt_pe_gamp_map}(c), we have:
\begin{align}
\label{eq:compare_delta}
\begin{split}
&\sum_{j\neq n}h_{\Omega_n\leftarrow\lambda_l}^{\Omega_j}(t, \vr_{\Omega_n}^{(t)}, \overline{\tau}_\Omega^r(t),\hat{\lambda}_{\Omega_n\leftarrow\lambda_l}^{(t+1)}[N], \overline{\boldsymbol\lambda}_{\Omega_n}^{(t)}\backslash\overline{\lambda}_{\Omega_n\leftarrow\lambda_l}^{(t)})\\
&\quad\geq \sum_{j\neq n}h_{\Omega_n\leftarrow\lambda_l}^{\Omega_j}(t, \vr_{\Omega_n}^{(t)}, \overline{\tau}_\Omega^r(t),\lambda_l^*, \overline{\boldsymbol\lambda}_{\Omega_n}^{(t)}\backslash\overline{\lambda}_{\Omega_n\leftarrow\lambda_l}^{(t)})\,.
\end{split}
\end{align}
Since $h_{\Omega_n\leftarrow\lambda_l}^{\Omega_j}(\cdot)$ is pseudo-Lipschitz continuous in $r_{\Omega_n}^{(t)}$, the left-hand side of (\ref{eq:compare_delta}) can be rewritten as follows as $N\rightarrow\infty$:
\begin{align}
\begin{split}
&\frac{1}{N-1}\sum_{j\neq n}h_{\Omega_n\leftarrow\lambda_l}^{\Omega_j}(t, \vr_{\Omega_n}^{(t)}, \overline{\tau}_\Omega^r(t),\hat{\lambda}_{\Omega_n\leftarrow\lambda_l}^{(t+1)}[N], \overline{\boldsymbol\lambda}_{\Omega_n}^{(t)}\backslash\overline{\lambda}_{\Omega_n\leftarrow\lambda_l}^{(t)})\\
&=\mathbb{E}\left[\sum_{j\neq n}h_{\Omega_n\leftarrow\lambda_l}^{\Omega_j}(t,\mathcal{R}_\Omega^{(t)}, \overline{\tau}_\Omega^r(t),\hat{\lambda}_{\Omega_n\leftarrow\lambda_l}^{(t+1)}[N], \overline{\boldsymbol\lambda}_{\Omega_n}^{(t)}\backslash\overline{\lambda}_{\Omega_n\leftarrow\lambda_l}^{(t)})\right]\,.
\end{split}
\end{align}
The right-hand side of (\ref{eq:compare_delta}) can be rewritten similarly. (\ref{eq:compare_delta}) then becomes:
\begin{align}
\begin{split}
&\mathbb{E}\left[h_{\Omega_n\leftarrow\lambda_l}^{\Omega_j}(t,\mathcal{R}_\Omega^{(t)}, \overline{\tau}_\Omega^r(t),\hat{\lambda}_{\Omega_n\leftarrow\lambda_l}^{(t+1)}[N], \overline{\boldsymbol\lambda}_{\Omega_n}^{(t)}\backslash\overline{\lambda}_{\Omega_n\leftarrow\lambda_l}^{(t)})\right]\\
&\geq\mathbb{E}\left[h_{\Omega_n\leftarrow\lambda_l}^{\Omega_j}(t,\mathcal{R}_\Omega^{(t)}, \overline{\tau}_\Omega^r(t),\lambda_l^*, \overline{\boldsymbol\lambda}_{\Omega_n}^{(t)}\backslash\overline{\lambda}_{\Omega_n\leftarrow\lambda_l}^{(t)})\right]\,.
\end{split}
\end{align}
Assumption \ref{aspt_pe_gamp_map}(b) states that $\lambda_l^*$ is the unique maxima of the right-hand side, we then have:
\begin{align}
\lim_{N\rightarrow\infty}\hat{\lambda}_{\Omega_n\leftarrow\lambda_l}^{(t+1)}[N]=\lambda_l^*\,,
\end{align}
which proves (\ref{eq:eca_conclusion}). 
\medskip

\item {\bfseries MMSE Parameter Estimation:}
Using the compactness of the sets $\mathcal{U}_\lambda$ in Assumption \ref{aspt_pe_gamp_mmse}(a) and the continuity condition of the open neighborhood $\rho(\widetilde{\tau}_\Omega^r, \widetilde{\boldsymbol\lambda})$ in Assumption \ref{aspt_pe_gamp_mmse}(b), we have the following:
\begin{align}
\label{eq:mmse_pe_lemma_1_1}
\begin{split}
&\lim_{N\rightarrow\infty}\frac{1}{N-1}\sum_{j\neq n}h_{\Omega_n\leftarrow\lambda_l}^{\Omega_j}(t, \vr_{\Omega_n}^{(t)}, \tau_\Omega^r(t), \lambda_l, \hat{\boldsymbol\lambda}_{\Omega_n}^{(t)}\backslash\hat{\lambda}_{\Omega_n\leftarrow\lambda_l}^{(t)})\\
&=\frac{1}{N-1}\sum_{j\neq n}h_{\Omega_n\leftarrow\lambda_l}^{\Omega_j}(t, \vr_{\Omega_n}^{(t)}, \overline{\tau}_\Omega^r(t),\lambda_l, \overline{\boldsymbol\lambda}_{\Omega_n}^{(t)}\backslash\overline{\lambda}_{\Omega_n\leftarrow\lambda_l}^{(t)})\,.
\end{split}
\end{align}
Since $h_{\Omega_n\leftarrow\lambda_l}^{\Omega_j}$ is pseudo-Lipschitz continuous in $r_{\Omega_n}^{(t)}$, we also have:
\begin{align}
\label{eq:mmse_pe_lemma_1_2}
\begin{split}
&\lim_{N\rightarrow\infty}\frac{1}{N-1}\sum_{j\neq n}h_{\Omega_n\leftarrow\lambda_l}^{\Omega_j}(t, \vr_{\Omega_n}^{(t)}, \overline{\tau}_\Omega^r(t),\lambda_l, \overline{\boldsymbol\lambda}_{\Omega_n}^{(t)}\backslash\overline{\lambda}_{\Omega_n\leftarrow\lambda_l}^{(t)})\\
&=\mathbb{E}\left[h_{\Omega_n\leftarrow\lambda_l}^{\Omega_j}(t,\mathcal{R}_\Omega^{(t)}, \overline{\tau}_\Omega^r(t),\lambda_l, \overline{\boldsymbol\lambda}_{\Omega_n}^{(t)}\backslash\overline{\lambda}_{\Omega_n\leftarrow\lambda_l}^{(t)})\right]\,.
\end{split}
\end{align} 
Combining (\ref{eq:mmse_pe_lemma_1_1}) and (\ref{eq:mmse_pe_lemma_1_2}), we then have:
\begin{align}
\begin{split}
&\lim_{N\rightarrow\infty}\frac{1}{N-1}\sum_{j\neq n}h_{\Omega_n\leftarrow\lambda_l}^{\Omega_j}(t, \vr_{\Omega_n}^{(t)}, \tau_\Omega^r(t),\lambda_l, \hat{\boldsymbol\lambda}_{\Omega_n}^{(t)}\backslash\hat{\lambda}_{\Omega_n\leftarrow\lambda_l}^{(t)})\\
&=\mathbb{E}\left[h_{\Omega_n\leftarrow\lambda_l}^{\Omega_j}(t,\mathcal{R}_\Omega^{(t)}, \overline{\tau}_\Omega^r(t),\lambda_l, \overline{\boldsymbol\lambda}_{\Omega_n}^{(t)}\backslash\overline{\lambda}_{\Omega_n\leftarrow\lambda_l}^{(t)})\right]\,.
\end{split}
\end{align}
Using the continuity property of the exponential function $\exp(\cdot)$, as $N\rightarrow\infty$ we can get:
\begin{align}
\frac{\exp(\frac{1}{N-1}\sum_{j\neq n}h_{\Omega_n\leftarrow\lambda_l}^{\Omega_j}(\cdot))}{\int_{\lambda_l}\exp(\frac{1}{N-1}\sum_{j\neq n}h_{\Omega_n\leftarrow\lambda_l}^{\Omega_j}(\cdot))}=\frac{\exp(\mathbb{E}\left[h_{\Omega_n\leftarrow\lambda_l}^{\Omega_j}(\cdot)\right])}{\int_{\lambda_l}\exp(\mathbb{E}\left[h_{\Omega_n\leftarrow\lambda_l}^{\Omega_j}(\cdot)\right])}\,.
\end{align}
Since the set $\mathcal{U}_\lambda$ is compact and the mean of a probability distribution is unique, we have:
\begin{align}
\begin{split}
&\lim_{N\rightarrow\infty}\int_{\lambda_l}\lambda_l\frac{\exp(\frac{1}{N-1}\sum_{j\neq n}h_{\Omega_n\leftarrow\lambda_l}^{\Omega_j}(\cdot))}{\int_{\lambda_l}\exp(\frac{1}{N-1}\sum_{j\neq n}h_{\Omega_n\leftarrow\lambda_l}^{\Omega_j}(\cdot))}=\\
&\quad\int_{\lambda_l}\lambda_l\frac{\exp(\mathbb{E}\left[h_{\Omega_n\leftarrow\lambda_l}^{\Omega_j}(\cdot)\right])}{\int_{\lambda_l}\exp(\mathbb{E}\left[h_{\Omega_n\leftarrow\lambda_l}^{\Omega_j}(\cdot)\right])}\,,
\end{split}
\end{align}
which proves (\ref{eq:eca_conclusion}). 
\end{enumerate}

\section{Proof of Corollary \ref{crl:pe_gamp_convergence}}
\label{app:proof_corollary}
We only need to show the empirical convergences of $\hat{\boldsymbol\lambda}_{\Omega_n}^{(t+1)}, \hat{\boldsymbol\theta}_{\Phi_m}^{(t+1)}$. From Assumption \ref{aspt_pe_gamp_map} we can get Lemma \ref{lemma:map_pl_continuity}, which corresponds to Assumption \ref{aspt:adaptive_gamp}(c) in Appendix \ref{app:se_adaptive_gamp}. Using Theorem \ref{thm:adaptive_gamp}, we have:
\begin{subequations}
\begin{align}
\label{eq:some_convergence}
&\lim_{N\rightarrow\infty}r_{\Omega_n}^{(t)}\stackrel{\normalfont\textrm{PL$(k)$}}{=}\mathcal{R}_\Omega^{(t)}\\
&\lim_{N\rightarrow\infty}(q_{\Phi_m}^{(t)}, y_m)\stackrel{\normalfont\textrm{PL$(k)$}}{=}(\mathcal{Q}_\Phi^{(t)}, \mathcal{Y})\\
&\lim_{N\rightarrow\infty}\boldsymbol\psi_\tau=\overline{\boldsymbol\psi}_\tau\,.
\end{align}
\end{subequations} 
The empirical convergences of the parameters can be proved using induction. For $t=0$, the convergences of $\hat{\boldsymbol\lambda}_{\Omega_n}^{(0)}, \hat{\boldsymbol\theta}_{\Phi_m}^{(0)}$ hold according to (\ref{eq:pe_gamp_parameter_init}) in Assumption \ref{aspt_pe_gamp}(c). With (\ref{eq:pe_gamp_parameter_init}, \ref{eq:some_convergence}), we can use Lemma \ref{lemma:map_pl_continuity} to obtain:
\begin{align}
\lim_{N\rightarrow\infty}\hat{\boldsymbol\lambda}_{\Omega_n}^{(1)}=\overline{\boldsymbol\lambda}_{\Omega_n}^{(1)}, \quad\lim_{N\rightarrow\infty}\hat{\boldsymbol\theta}_{\Phi_m}^{(1)}=\overline{\boldsymbol\theta}_{\Phi_m}^{(1)}\,.
\end{align}
The convergences of the rest scalars can be obtained directly using Theorem \ref{thm:adaptive_gamp}. Hence the following holds for any $t$.
\begin{align}
\lim_{N\rightarrow\infty}\hat{\boldsymbol\lambda}_{\Omega_n}^{(t+1)}=\overline{\boldsymbol\lambda}_{\Omega_n}^{(t+1)}, \quad\lim_{N\rightarrow\infty}\hat{\boldsymbol\theta}_{\Phi_m}^{(t+1)}=\overline{\boldsymbol\theta}_{\Phi_m}^{(t+1)}\,.
\end{align}

\section{MAP Parameter Estimation}
\label{app:map_pe}
Here we choose four popular channels used in the sparse signal recovery models as examples to demonstrate how to perform MAP parameter estimation with the proposed PE-GAMP. The line search method is used to find the maximizing parameters, it requires computing the derivatives of the following $h_{\Omega_n\leftarrow\lambda_l}^{(t+1)}(\cdot), h_{\Phi_m\leftarrow\theta_k}^{(t+1)}(\cdot)$ in (\ref{eq:map_pe_gamp}) with respect to the parameters $\lambda_l, \theta_k$. 
\begin{subequations}
\begin{align}
&h_{\Omega_n\leftarrow\lambda_l}^{(t+1)}(\cdot)=\sum_{j\neq n}\Delta^{(t+1)}_{\Omega_j\rightarrow\lambda_l}+\log p(\lambda_l)\\
&h_{\Phi_m\leftarrow\theta_k}^{(t+1)}(\cdot)=\sum_{i\neq m}\Delta^{(t+1)}_{\Phi_i\rightarrow\theta_k} + \log p(\theta_k)
\end{align}
\end{subequations}
The derivatives of $\log p(\lambda_l), \log p(\theta_k))$ depends on the chosen priors and are easy to compute. Here we give the derivatives of $\Delta_{\Omega_j\rightarrow\lambda_l}^{(t+1)}, \Delta_{\Phi_i\rightarrow\theta_k}^{(t+1)}$ with respect to $\lambda_l, \theta_k$ in details. 

\subsection{Sum-product Message Passing}
We compute the derivatives for different channels respectively as follows:
\smallskip
\begin{enumerate}
\item {\bfseries Bernoulli-Gaussian mixture Input Channel:}
BGm distribution is given in (\ref{eq:bgm_distribution}). We then have:
\begin{align}
\label{eq:bg_omega_lambda_app}
\begin{split}
&\log\int_{-\infty}^{\infty}\Omega_j(x_j,\boldsymbol\lambda)\cdot\exp\left(\Delta_{\Omega_j\leftarrow x_j}^{(t+1)}\right)\,dx_j\\
&=\log\int_{x_j}p(x_j|\boldsymbol\lambda)\exp\left(-\frac{1}{2\tau_{\Omega_j}^r(t)}\left(x_j-r_{\Omega_j}^{(t)}\right)\right)\\
%&=\log\int_{x_j}\left[\vphantom{\exp\left(-\frac{1}{2\tau_{\Omega_j}^r(t)}\left(x_j-r_{\Omega_j}^{(t)}\right)\right)}\left(\left(1-\lambda_1\right)\delta(x_j)+\lambda_1\sum_{c=1}^C\mathcal{N}\left(x_j;\lambda_{c+1},\lambda_{c+2}\right)\right)\right.\\
%&\quad\quad\left.\times\exp\left(-\frac{1}{2\tau_{\Omega_j}^r(t)}\left(x_j-r_{\Omega_j}^{(t)}\right)\right)\right]\\
%&=\log\left[(1-\lambda_1)\exp\left(-\frac{\left(r_{\Omega_j}^{(t)}\right)^2}{2\tau_{\Omega_j}^r(t)}\right)+\lambda_1\cdot\sqrt{\frac{\tau_{\Omega_j}^r(t)}{\lambda_3+\tau_{\Omega_j}^r(t)}}\cdot\exp\left(-\frac{1}{2}\frac{1}{\left(\lambda_3+\tau_{\Omega_j}^r(t)\right)}\left(\lambda_2-r_{\Omega_j}^{(t)}\right)^2\right)\right]\\
&=\log\left[(1-\lambda_1)\cdot\kappa_1+\lambda_1\sum_{c=1}^C\lambda_{c+1}\cdot\kappa_2(\lambda_{c+2},\lambda_{c+3})\right]\\
&=\log\left[(1-\lambda_1)\cdot\kappa_1+\lambda_1\sum_{c=1}^C\lambda_{c+1}\cdot\kappa_2(c)\right]\,,
\end{split}
\end{align}
where $\kappa_1$ doesn't depend on $\boldsymbol\lambda$; for the $c$-th Gaussian mixture, $\kappa_2(c)=\kappa_2(\lambda_{c+2},\lambda_{c+3})$ depends on $\lambda_{c+2},\lambda_{c+3}$.
\begin{subequations}
\begin{align}
\label{eq:kappa_1_bgm}
&\kappa_1=\exp\left(-\frac{\left(r_{\Omega_j}^{(t)}\right)^2}{2\tau_{\Omega_j}^r(t)}\right)\\
\begin{split}
&\kappa_2(c)=\sqrt{\frac{\tau_{\Omega_j}^r(t)}{\lambda_{c+3}+\tau_{\Omega_j}^r(t)}}\exp\left(-\frac{1}{2}\frac{\left(\lambda_{c+2}-r_{\Omega_j}^{(t)}\right)^2}{\left(\lambda_{c+3}+\tau_{\Omega_j}^r(t)\right)}\right)\,.
\end{split}
\end{align}
\end{subequations}
(\ref{eq:bg_omega_lambda_app}) is essentially (\ref{eq:delta_omega_lambda_simplified}). Let $\kappa_3(\boldsymbol\lambda)$ be as follows:
\begin{align}
\kappa_3(\boldsymbol\lambda)=\frac{\lambda_1}{(1-\lambda_1)\cdot\kappa_1+\lambda_1\sum_{c=1}^C\lambda_{c+1}\cdot\kappa_2(c)}\,.
\end{align}
Let $\boldsymbol\lambda\backslash\lambda_l$ denote the parameter sequence generated by removing $\lambda_l$ from $\boldsymbol\lambda$. Taking derivatives of (\ref{eq:delta_omega_lambda_simplified}) with respect to $\lambda_1, \lambda_{c+2}, \lambda_{c+3}$, we have:
\begin{subequations}
\begin{align}
\frac{\partial \Delta_{\Omega_j\rightarrow\lambda_1}^{(t+1)}}{\partial \lambda_1}&=\frac{-\kappa_1+\sum_{c=1}^C\hat{\lambda}_{c+1}^{(t)}\cdot\kappa_2\left(\hat{\lambda}_{c+2}^{(t)}, \hat{\lambda}_{c+3}^{(t)}\right)}{(1-\lambda_1)\cdot\kappa_1+\lambda_1\sum_{c=1}^C\hat{\lambda}_{c+1}^{(t)}\cdot\kappa_2\left(\hat{\lambda}_{c+2}^{(t)}, \hat{\lambda}_{c+3}^{(t)}\right)}\\
\begin{split}
\frac{\partial \Delta_{\Omega_j\rightarrow\lambda_{c+2}}^{(t+1)}}{\partial \lambda_{c+2}}&=\kappa_3\left(\hat{\boldsymbol\lambda}^{(t)}\backslash\hat{\lambda}_{c+2}^{(t)},\lambda_{c+2}\right)\\
&\quad\quad\times\frac{-\kappa_2\left(\lambda_{c+2},\hat{\lambda}_{c+3}^{(t)}\right)\left(\lambda_{c+2}-r_{\Omega_j}^{(t)}\right)}{\hat{\lambda}_{c+2}^{(t)}+\tau_{\Omega_j}^r(t)}
\end{split}\\
\begin{split}
\frac{\partial \Delta_{\Omega_j\rightarrow\lambda_{c+3}}^{(t+1)}}{\partial \lambda_{c+3}}&=-\kappa_3\left(\hat{\boldsymbol\lambda}^{(t)}\backslash\hat{\lambda}_{c+3}^{(t)},\lambda_{c+3}\right)\frac{\kappa_2\left(\hat{\lambda}_{c+2}^{(t)},\lambda_{c+3}\right)}{2\left(\lambda_{c+3}+\tau_{\Omega_j}^r(t)\right)}\\
&\quad\quad\times\left(1-\frac{\left(\hat{\lambda}_{c+2}^{(t)}-r_{\Omega_j}^{(t)}\right)^2}{\lambda_{c+3}+\tau_{\Omega_j}^r(t)}\right)\,,
\end{split}
\end{align}
\end{subequations}
where $\hat{\boldsymbol\lambda}^{(t)}$ are the estimated parameters in the previous $t$-th iteration. 

The updates for the weights $\lambda_{c+1}$ are more complicated, they need to satisfy the nonnegative and sum-to-one constrains. Here we can rewrite the weight $\lambda_{c+1}$ as follows :
\begin{align}
\label{eq:mixture_weights}
\lambda_{c+1} = \frac{\exp(\omega_c)}{\sum_{k=1}^C\exp(\omega_k)}\,,
\end{align}
where $\omega_c\in\mathbb{R}$. We then can remove the constrains on $\lambda_{c+1}$ and maximize $\Delta_{\Omega_j\rightarrow\lambda_{c+1}}^{(t+1)}$ with respect to $\omega_c$ instead. The derivative is then:
\begin{align}
\begin{split}
&\frac{\partial \Delta_{\Omega_j\rightarrow\lambda_{c+1}}^{(t+1)}}{\partial \omega_c}=\frac{\partial \Delta_{\Omega_j\rightarrow\lambda_{c+1}}^{(t+1)}}{\partial \lambda_{c+1}}\cdot\frac{\partial \lambda_{c+1}}{\partial \omega_c}\\
&=\frac{\hat{\lambda}_1^{(t)}}{(1-\hat{\lambda}_1^{(t)})\cdot\kappa_1+\hat{\lambda}_1^{(t)}\sum_{k=1}^C\lambda_{k+1}\cdot\kappa_2\left(\hat{\lambda}_{k+2}^{(t)}, \hat{\lambda}_{k+3}^{(t)}\right)}\\
&\quad\quad\times\sum_{k=1}^C\kappa_2\left(\hat{\lambda}_{k+2}^{(t)}, \hat{\lambda}_{k+3}^{(t)}\right)\left(\lambda_{k+1}\cdot\boldsymbol1_{(k=c)}-\lambda_{k+1}\lambda_{c+1}\right)\,,
\end{split}
\end{align}
where $\boldsymbol1_{(k=c)}=1$ if $k=c$ and $\boldsymbol1_{(k=c)}=0$ if $k\neq c$.
\smallskip
\item {\bfseries Bernoulli-Exponential mixture Input Channel:}
BEm distribution is given in (\ref{eq:bem_distribution}), we then have:
\begin{align}
\begin{split}
&\log\int_{0}^{\infty}\Omega_j(x_j,\boldsymbol\lambda)\cdot\exp\left(\Delta_{\Omega_j\leftarrow x_j}^{(t+1)}\right)\,dx_j\\
&=\log\int_{x_j}p(x_j|\boldsymbol\lambda)\exp\left(-\frac{1}{2\tau_{\Omega_j}^r(t)}\left(x_j-r_{\Omega_j}^{(t)}\right)\right)\\
%&=\log\int_{x_j}\left[\left((1-\lambda_1)\delta(x_j)+\lambda_1\sum_{c=1}^C\lambda_{c+1}\cdot\lambda_{c+2}\exp\left(-\lambda_{c+2}x_j\right)\right)\right.\\
%&\quad\quad\left.\times\exp\left(-\frac{1}{2\tau_{\Omega_j}^r(t)}\left(x_j-r_{\Omega_j}^{(t)}\right)\right)\right]\\
&=\log\left[(1-\lambda_1)\cdot\kappa_1+\lambda_1\sum_{c=1}^C\lambda_{c+1}\cdot\kappa_2(\lambda_{c+2})\right]\\
&=\log\left[(1-\lambda_1)\cdot\kappa_1+\lambda_1\sum_{c=1}^C\lambda_{c+1}\cdot\kappa_2(c)\right]\,,
\end{split}
\end{align}
where $\kappa_1$ is the same as (\ref{eq:kappa_1_bgm}); for the $c$-th Exponential mixture, $\kappa_2(c)$ depends on $\lambda_{c+2}$.
\begin{align}
\kappa_2(c)=\lambda_{c+2}\cdot\sqrt{\frac{\pi\tau_{\Omega_j}^r(t)}{2}}\cdot\erfcx\left(-\frac{r_{\Omega_j}^{(t)}-\lambda_{c+2}\cdot\tau_{\Omega_j}^r(t)}{\sqrt{2\tau_{\Omega_j}^r(t)}}\right)\,,
\end{align}
where $\erfcx(\cdot)$ is the scaled complementary error function. Taking the derivative w.r.t. $\lambda_1, \lambda_{c+2}$, we have:
\begin{subequations}
\begin{align}
&\frac{\partial \Delta_{\Omega_j\rightarrow\lambda_1}^{(t+1)}}{\partial \lambda_1}=\frac{-\kappa_1+\sum_{c=1}^C\hat{\lambda}_{c+1}^{(t)}\cdot\kappa_2\left(\hat{\lambda}_{c+2}^{(t)}\right)}{(1-\lambda_1)\cdot\kappa_1+\lambda_1\sum_{c=1}^C\hat{\lambda}_{c+1}^{(t)}\cdot\kappa_2\left(\hat{\lambda}_{c+2}^{(t)}\right)}\\
\begin{split}
&\frac{\partial \Delta_{\Omega_j\rightarrow\lambda_{c+2}}^{(t+1)}}{\partial \lambda_{c+2}}=\kappa_3\left(\hat{\boldsymbol\lambda}^{(t)}\backslash\hat{\lambda}_{c+2}^{(t)},\lambda_{c+2}\right)\\
&\quad\quad\times\hat{\lambda}_{c+1}^{(t)}\left(\lambda_{c+2}\tau_{\Omega_j}^r(t)+\left(r_{\Omega_j}^{(t)}-\lambda_{c+2}\tau_{\Omega_j}^r(t)\right)\kappa_2(c)\right)\,.
\end{split}
\end{align}
\end{subequations}
We write the mixture weights $\lambda_{c+1}$ in the same form as (\ref{eq:mixture_weights}), and take the derivative w.r.t. $\omega_c$:
\begin{align}
\begin{split}
&\frac{\partial \Delta_{\Omega_j\rightarrow\lambda_{c+1}}^{(t+1)}}{\partial \omega_c}=\frac{\partial \Delta_{\Omega_j\rightarrow\lambda_{c+1}}^{(t+1)}}{\partial \lambda_{c+1}}\cdot\frac{\partial \lambda_{c+1}}{\partial \omega_c}\\
&=\frac{\hat{\lambda}_1^{(t)}}{(1-\hat{\lambda}_1^{(t)})\cdot\kappa_1+\hat{\lambda}_1^{(t)}\sum_{k=1}^C\lambda_{k+1}\cdot\kappa_2\left(\hat{\lambda}_{k+2}^{(t)}\right)}\\
&\quad\quad\times\sum_{k=1}^C\kappa_2\left(\hat{\lambda}_{k+2}^{(t)}\right)\left(\lambda_{k+1}\cdot\boldsymbol1_{(k=c)}-\lambda_{k+1}\lambda_{c+1}\right)\,.
\end{split}
\end{align}
\smallskip
\item {\bfseries Laplace Input Channel:}
Laplace distribution is given in (\ref{eq:laplace_distribution}). Similarly, we have:
\begin{align}
\begin{split}
&\log\int_{-\infty}^{\infty}\Omega_j(x_j,\boldsymbol\lambda)\cdot\exp\left(\Delta_{\Omega_j\leftarrow x_j}^{(t+1)}\right)\,dx_j\\
&=\log\int_{x_j}p(x_j|\vy)\cdot\exp\left(-\frac{\left(x_j-r_{\Omega_j}^{(t)}\right)^2}{2\tau_{\Omega_j}^r(t)}\right)\\
%&=\log\int_{x_j}\frac{\lambda_1}{2}\exp(-\lambda_1|x_j|)\cdot\exp\left(-\frac{\left(x_j-r_{\Omega_j}^{(t)}\right)^2}{2\tau_{\Omega_j}^r(t)}\right)\\
%&=\log\left[\frac{\lambda_1}{2}\exp\left(-\frac{\left(r_{\Omega_j}^{(t)}\right)^2}{2\tau_{\Omega_j}^r(t)}\right)\left(\int_0^\infty\exp\left(-\frac{1}{2\tau_{\Omega_j}^r(t)}x_j^2+\left(\frac{\left(r_{\Omega_j}^{(t)}\right)^2}{\tau_{\Omega_j}^r(t)}-\lambda_1\right)x_j\right)dx_j\right.\right.\\
%&\quad\quad\left.\left.+\int_{-\infty}^0\exp\left(-\frac{1}{2\tau_{\Omega_j}^r(t)}x_j^2+\left(\frac{\left(r_{\Omega_j}^{(t)}\right)^2}{\tau_{\Omega_j}^r(t)}+\lambda_1\right)x_j\right)dx_j\right)\right]\\
&=\log\left[\lambda_1\kappa_1\left(\kappa_2(\lambda_1)+\kappa_3(\lambda_1)\right)\right]-\log2\,,
\end{split}
\end{align}
where $\kappa_1$ is the same as (\ref{eq:kappa_1_bgm}), $\kappa_2(\lambda_1), \kappa_3(\lambda_1)$ depend on $\lambda_1$. They are as follows:
\begin{subequations}
\begin{align}
\kappa_2(\lambda_1)&=\lambda_1\cdot\sqrt{\frac{\pi\tau_{\Omega_j}^r(t)}{2}}\cdot\erfcx\left(-\frac{r_{\Omega_j}^{(t)}-\lambda_1\cdot\tau_{\Omega_j}^r(t)}{\sqrt{2\tau_{\Omega_j}^r(t)}}\right)\\
\kappa_3(\lambda_1)&=\lambda_1\cdot\sqrt{\frac{\pi\tau_{\Omega_j}^r(t)}{2}}\cdot\erfcx\left(\frac{r_{\Omega_j}^{(t)}+\lambda_1\cdot\tau_{\Omega_j}^r(t)}{\sqrt{2\tau_{\Omega_j}^r(t)}}\right)\,.
\end{align}
\end{subequations}
Taking derivative of (\ref{eq:delta_omega_lambda_simplified}) with respect to $\lambda_1$, we have:
\begin{align}
\begin{split}
%\frac{\partial \Delta_{\Omega_j\rightarrow\lambda_1}^{(t+1)}}{\partial \lambda_1}&=\frac{1}{\lambda_1}+\frac{1}{\kappa_2(\lambda_1)+\kappa_3(\lambda_1)}\left(-\int_0^\infty x_j\exp\left(-\frac{1}{2\tau_{\Omega_j}^r(t)}x_j^2+\left(\frac{r_{\Omega_j}^{(t)}}{\tau_{\Omega_j}^r(t)}-\lambda_1\right)x_j\right)dx_j\right.\\
%&\quad\quad\left.+\int_{-\infty}^0 x_j\exp\left(-\frac{1}{2\tau_{\Omega_j}^r(t)}x_j^2+\left(\frac{r_{\Omega_j}^{(t)}}{\tau_{\Omega_j}^r(t)}+\lambda_1\right)x_j\right)dx_j\right)\\
\frac{\partial \Delta_{\Omega_j\rightarrow\lambda_1}^{(t+1)}}{\partial \lambda_1}=\frac{1}{\lambda_1}+\frac{\kappa_4(\lambda_1)+\kappa_5(\lambda_1)}{\kappa_2(\lambda_1)+\kappa_3(\lambda_1)}\,,
\end{split}
\end{align}
where $\kappa_4(\lambda_1),\kappa_5(\lambda_1)$ depend on $\lambda_1$. They are as follows:
\begin{subequations}
\begin{align}
\kappa_4(\lambda_1)&=\lambda_1\tau_{\Omega_j}^r(t)+\left(r_{\Omega_j}^{(t)}-\lambda_1\tau_{\Omega_j}^r(t)\right)\kappa_2(\lambda_1)\\
\kappa_5(\lambda_1)&=-\lambda_1\tau_{\Omega_j}^r(t)+\left(r_{\Omega_j}^{(t)}+\lambda_1\tau_{\Omega_j}^r(t)\right)\kappa_3(\lambda_1)\,.
\end{align}
\end{subequations}
\smallskip
\item {\bfseries Additive White Gaussian Noise Output Channel:}
The white Gaussian distribution is shown in \ref{eq:adgn_distribution}. We have the following:
\begin{align}
\begin{split}
&\log\int_{-\infty}^{\infty}\Phi(y_i,z_i,\boldsymbol\theta)\exp\left(\frac{-1}{2\tau_{\Phi_i}^q(t)}\left(z_i-q_{\Phi_i}^{(t)}\right)^2\right)\,dz_m\\
&=\frac{1}{2}\log\tau_{\Phi_i}^q(t)-\frac{1}{2}\log\left(\theta_1+\tau_{\Phi_i}^q(t)\right)\\
&\quad\quad-\frac{1}{2}\frac{1}{\left(\theta_1+\tau_{\Phi_i}^q(t)\right)}\left(y_i-q_{\Phi_i}^{(t)}\right)^2\,.
\end{split}
\end{align}
Taking derivative of (\ref{eq:delta_phi_theta_simplified}) with respect to $\theta_1$, we have:
\begin{align}
\begin{split}
\frac{\partial \Delta_{\Phi_i\rightarrow\theta_1}^{(t+1)}}{\partial \theta_1}&=\frac{\left(y_i-q_{\Phi_i}^{(t)}\right)^2}{2\left(\theta_1+\tau_{\Phi_i}^q(t)\right)^2}-\frac{1}{2\left(\theta_1+\tau_{\Phi_i}^q(t)\right)}\,.
\end{split}
\end{align}
\end{enumerate}

\subsection{Max-sum Message Passing}
The parameters of \emph{max-sum} message passing can also be estimated using Algorithm \ref{alg:line_search}. We analyze the channels in this case as follows:
\smallskip
\begin{enumerate}
\item {\bfseries Bernoulli-Gaussian mixture Input Channel:} BGm input channel is not really suited for the \emph{max-sum} message passing. If we compute (\ref{eq:factor_to_variable_x_ms_app}), the maximizing $x_j$ would be $0$, which makes both the parameter estimation and signal recovery impossible.
\smallskip
\item {\bfseries Bernoulli-Exponential mixture Input Channel:} BEx input channel is also not suited for the \emph{max-sum} message passing for the same reason as the BGm input channel. 
\item {\bfseries Laplace Input Channel:} (\ref{eq:factor_to_variable_x_ms_app}) can be written as follows:
\begin{align}
\Delta_{\Omega_j\rightarrow\lambda_1}^{(t+1)}=\max_{x_j}\left[\log\lambda_1-\lambda_1\left|x_j\right|-\frac{\left(x_j-r_{\Omega_j}^{(t)}\right)^2}{2\tau_{\Omega_j}^r(t)}\right]\,.
\end{align}
The maximizing $x_j$ is given by the soft-thresholding method:
\begin{align}
\tilde{x}_n^{(t+1)} = \left(\left|r_{\Omega_j}^{(t)}\right|-\lambda_1\tau_{\Omega_j}^r(t)\right)_+\cdot\textrm{sign}\left(r_{\Omega_j}^{(t)}\right)\,.
\end{align}
If $\left|r_{\Omega_j}^{(t)}\right|>\lambda_1\tau_{\Omega_j}^r(t)$, we have:
\begin{align}
\label{eq:delta_omega_lambda_laplace1_app}
\Delta_{\Omega_j\rightarrow\lambda_1}^{(t+1)}=\log\lambda_1-\lambda_1\left|r_{\Omega_j}^{(t)}\right|+\frac{1}{2}\lambda_1^2\tau_{\Omega_j}^r(t)\,.
\end{align}
If $\left|r_{\Omega_j}^{(t)}\right|\leq\lambda_1\tau_{\Omega_j}^r(t)$, we have:
\begin{align}
\label{eq:delta_omega_lambda_laplace2_app}
\Delta_{\Omega_j\rightarrow\lambda_1}^{(t+1)}=\log\lambda_1-\frac{1}{2\tau_{\Omega_j}^r(t)}\left(r_{\Omega_j}^{(t)}\right)^2\,.
\end{align}
%\begin{align}
%\label{eq:delta_omega_lambda_laplace_app}
%\Delta_{\Omega_j\rightarrow\lambda_1}^{(t+1)}=\left\{\begin{array}{ll}
%\log\lambda_1-\lambda_1\left|r_{\Omega_j}^{(t)}\right|+\frac{1}{2}\lambda_1^2\tau_{\Omega_j}^r(t) &\textrm{ if } \left|r_{\Omega_j}^{(t)}\right|>\lambda_1\tau_{\Omega_j}^r(t)\\
%\log\lambda_1-\frac{1}{2\tau_{\Omega_j}^r(t)}\left(r_{\Omega_j}^{(t)}\right)^2 &\textrm{ if } \left|r_{\Omega_j}^{(t)}\right|\leq\lambda_1\tau_{\Omega_j}^r(t)
%\end{array}\right.
%\end{align}
We can see from (\ref{eq:delta_omega_lambda_laplace1_app}, \ref{eq:delta_omega_lambda_laplace2_app}) that the $\lambda_1$ that maximizes (\ref{eq:ms_pe_map_app_lambda}) is always $\infty$, which makes the estimation of $\lambda_1$ impossible.
\smallskip
\item {\bfseries Additive White Gaussian Noise Output Channel:} (\ref{eq:factor_to_variable_theta_ms_app}) can be written as follows:
\begin{align}
\begin{split}
\Delta_{\Phi_i\rightarrow\theta_1}^{(t+1)}&=\max_{z_i}\left[-\frac{1}{2}\log\theta_1-\frac{1}{2\theta_1}\left(y_i-z_i\right)^2\right.\\
&\quad\quad\left.-\frac{1}{2\tau_{\Phi_i}^q(t)}\left(z_i-q_{\Phi_i}^{(t)}\right)^2\right]\,.
%&=\max_{z_i}-\frac{1}{2}\log\theta_1-\frac{1}{2}\left(\frac{1}{\theta_1}+\frac{1}{\tau_{\Phi_i}^q(t)}\right)\left[\left(z_i-\frac{y_i\tau_{\Phi_i}^q(t)+q_{\Phi_i}^{(t)\theta_1}}{\theta_1+\tau_{\Phi_i}^q(t)}\right)^2-\left(\frac{y_i\tau_{\Phi_i}^q(t)+q_{\Phi_i}^{(t)\theta_1}}{\theta_1+\tau_{\Phi_i}^q(t)}\right)^2\right]
\end{split}
\end{align}
The maximizing $z_i$ is:
\begin{align}
\tilde{z}_m=\frac{y_i\tau_{\Phi_i}^q(t)+q_{\Phi_i}^{(t)}\theta_1}{\theta_1+\tau_{\Phi_i}^q(t)}\,.
\end{align}
We then have:
\begin{align}
\label{eq:delta_phi_theta_awgn_app}
\begin{split}
&\Delta_{\Phi_i\rightarrow\theta_1}^{(t+1)}\\
&=-\frac{1}{2}\log\theta_1+\frac{\left(q_{\Phi_i}^{(t)}\right)^2\theta_1+2q_{\Phi_i}^{(t)}y_i\tau_{\Phi_i}^q(t)-y_i^2\tau_{\Phi_i}^q(t)}{2\tau_{\Phi_i}^q(t)\left(\theta_1+\tau_{\Phi_i}^q(t)\right)}\,.
\end{split}
\end{align}
We can see from (\ref{eq:delta_phi_theta_awgn_app}) that the $\theta_1$ that maximizes (\ref{eq:ms_pe_map_app_theta}) is always $0$, which makes the estimation of $\theta_1$ impossible.
\end{enumerate}

\end{appendices}


\ifCLASSOPTIONcaptionsoff
  \newpage
\fi



%\bibliographystyle{IEEEtran}
\bibliographystyle{IEEEbib}
\bibliography{refs}

\begin{thebibliography}{10}

\bibitem{Decode05}
E.~Cand{\`e}s and T.~Tao,
\newblock ``Decoding by linear programming,''
\newblock {\em IEEE Trans. on Information Theory}, vol. 51(12), pp. 4203--4215,
  2005.

\bibitem{RUP06}
E.J. Candes, J.~Romberg, and T.~Tao,
\newblock ``Robust uncertainty principles: Exact signal reconstruction from
  highly incomplete frequency information,''
\newblock {\em IEEE Trans. on Information Theory}, vol. 52(2), pp. 489--509,
  2006.

\bibitem{CS06}
D.L. Donoho,
\newblock ``Compressed sensing,''
\newblock {\em IEEE Trans. on Information Theory}, vol. 52, no. 4, pp.
  1289--1306, 2006.

\bibitem{SRRP06}
E.J. Candes and T.~Tao,
\newblock ``Near-optimal signal recovery from random projections: Universal
  encoding strategies?,''
\newblock {\em IEEE Trans. on Information Theory}, vol. 52, no. 12, pp.
  5406--5425, 2006.

\bibitem{DL06}
M.~Aharon, M.~Elad, and A.~Bruckstein,
\newblock ``K-svd: An algorithm for designing overcomplete dictionaries for
  sparse representation,''
\newblock {\em IEEE Transactions on Signal Processing}, vol. 54, no. 11, pp.
  4311--4322, 2006.

\bibitem{SRC09}
J.~Wright, A.~Y. Yang, A.~Ganesh, S.~S. Sastry, and Y.~Ma,
\newblock ``Robust face recognition via sparse representation,''
\newblock {\em IEEE Transactions on Pattern Analysis and Machine Intelligence},
  vol. 31, no. 2, pp. 210--227, 2009.

\bibitem{Lasso94}
Robert Tibshirani,
\newblock ``Regression shrinkage and selection via the lasso,''
\newblock {\em Journal of the Royal Statistical Society}, vol. 58, pp.
  267--288, 1994.

\bibitem{GAMP11}
S.~Rangan,
\newblock ``Generalized approximate message passing for estimation with random
  linear mixing,''
\newblock in {\em Information Theory Proceedings (ISIT), 2011 IEEE
  International Symposium on}, July 2011, pp. 2168--2172.

\bibitem{AMP09}
D.L. Donoho, A.~Maleki, and A.~Montanari,
\newblock ``Message-passing algorithms for compressed sensing,''
\newblock {\em Proc. Nat. Acad. Sci.}, vol. 106, no. 45, pp. 18914--18919,
  2009.

\bibitem{AMP10}
D.L. Donoho, A.~Maleki, and A.~Montanari,
\newblock ``Message-passing algorithms for compressed sensing: I. motivation
  and construction,''
\newblock {\em Proc. Inform. Theory Workshop}, Jan 2010.

\bibitem{AMP11}
M.~Bayati and A.~Montanari,
\newblock ``The dynamics of message passing on dense graphs, with applications
  to compressed sensing,''
\newblock {\em IEEE Transactions on Information Theory}, vol. 57, no. 2, pp.
  764--785, Feb 2011.

\bibitem{Graph08}
M.~J. Wainwright and M.~I. Jordan,
\newblock ``Graphical models, exponential families, and variational
  inference,''
\newblock {\em Found. Trends Mach. Learn.}, vol. 1, no. 1-2, Jan. 2008.

\bibitem{SURE_AMP15}
C.~Guo and M.~E. Davies,
\newblock ``Near optimal compressed sensing without priors: Parametric sure
  approximate message passing,''
\newblock {\em IEEE Transactions on Signal Processing}, vol. 63, no. 8, pp.
  2130--2141, April 2015.

\bibitem{D_AMP16}
C.~A. Metzler, A.~Maleki, and R.~G. Baraniuk,
\newblock ``From denoising to compressed sensing,''
\newblock {\em IEEE Transactions on Information Theory}, vol. 62, no. 9, pp.
  5117--5144, Sept 2016.

\bibitem{EM77}
A.~P. Dempster, N.~M. Laird, and D.~B. Rubin,
\newblock ``Maximum likelihood from incomplete data via the em algorithm,''
\newblock {\em Journal of the Royal Statistical Society}, vol. 39, no. 1, pp.
  1--38, 1977.

\bibitem{EMBGAMP11}
J.~Vila and P.~Schniter,
\newblock ``Expectation-maximization bernoulli-gaussian approximate message
  passing,''
\newblock in {\em Conf. Rec. 45th Asilomar Conf. Signals, Syst. Comput}, Nov
  2011, pp. 799--803.

\bibitem{EMGMAMP13}
J.~P. Vila and P.~Schniter,
\newblock ``Expectation-maximization gaussian-mixture approximate message
  passing,''
\newblock {\em IEEE Transactions on Signal Processing}, vol. 61, no. 19, pp.
  4658--4672, Oct 2013.

\bibitem{AMP_CPE14}
U.~S. Kamilov, S.~Rangan, A.~K. Fletcher, and M.~Unser,
\newblock ``Approximate message passing with consistent parameter estimation
  and applications to sparse learning,''
\newblock {\em IEEE Transactions on Information Theory}, vol. 60, no. 5, pp.
  2969--2985, May 2014.

\bibitem{EM_GAMP_Florent12}
Florent Krzakala, Marc Mézard, Francois Sausset, Yifan Sun, and Lenka
  Zdeborová,
\newblock ``Probabilistic reconstruction in compressed sensing: algorithms,
  phase diagrams, and threshold achieving matrices,''
\newblock {\em Journal of Statistical Mechanics: Theory and Experiment}, vol.
  2012, no. 08, pp. P08009, 2012.

\bibitem{SHyperUM11}
M.~D. Iordache, J.~M. Bioucas-Dias, and A.~Plaza,
\newblock ``Sparse unmixing of hyperspectral data,''
\newblock {\em IEEE Transactions on Geoscience and Remote Sensing}, vol. 49,
  no. 6, pp. 2014--2039, June 2011.

\bibitem{ScSPM09}
J.~Yang, K.~Yu, Y.~Gong, and T.~Huang,
\newblock ``Linear spatial pyramid matching using sparse coding for image
  classification,''
\newblock in {\em CVPR}, 2009.

\bibitem{DBWav92}
Ingrid Daubechies,
\newblock {\em Ten Lectures on Wavelets},
\newblock Society for Industrial and Applied Mathematics, Philadelphia, PA,
  USA, 1992.

\bibitem{SIFT99}
D.G. Lowe,
\newblock ``Object recognition from local scale-invariant features,''
\newblock in {\em ICCV}, 1999, vol.~2, pp. 1150--1157.

\bibitem{HOG05}
N.~Dalal and B.~Triggs,
\newblock ``Histograms of oriented gradients for human detection,''
\newblock in {\em CVPR}, 2005, vol.~1, pp. 886--893.

\bibitem{BOF03}
J.~Sivic and A.~Zisserman,
\newblock ``Video google: a text retrieval approach to object matching in
  videos,''
\newblock in {\em ICCV}, 2003, pp. 1470--1477.

\bibitem{BOF04}
G.~Csurka, C.~R. Dance, L.~Fan, J.~Willamowski, and C.~Bray,
\newblock ``Visual categorization with bags of keypoints,''
\newblock in {\em In Workshop on Statistical Learning in Computer Vision,
  ECCV}, 2004, pp. 1--22.

\bibitem{Caltech101}
L.~Fei-Fei, R.~Fergus, and P.~Perona,
\newblock ``Learning generative visual models from few training examples: an
  incremental bayesian approach tested on 101 object categories,''
\newblock {\em IEEE CVPR Workshop on Generative-Model Based Vision}, 2004.

\bibitem{vlfeat}
A.~Vedaldi and B.~Fulkerson,
\newblock ``{VLFeat}: An open and portable library of computer vision
  algorithms,'' \url{http://www.vlfeat.org/}, 2008.

\bibitem{Kmeans02}
Tapas Kanungo, David~M. Mount, Nathan~S. Netanyahu, Christine~D. Piatko, Ruth
  Silverman, and Angela~Y. Wu,
\newblock ``An efficient k-means clustering algorithm: Analysis and
  implementation,''
\newblock {\em IEEE Trans. Pattern Anal. Mach. Intell.}, vol. 24, no. 7, pp.
  881--892, July 2002.

\bibitem{SPM06}
S.~Lazebnik, C.~Schmid, and J.~Ponce,
\newblock ``Beyond bags of features: Spatial pyramid matching for recognizing
  natural scene categories,''
\newblock in {\em CVPR}, 2006, vol.~2, pp. 2169--2178.

\bibitem{SVM95}
C.~Cortes and V.~Vapnik,
\newblock ``Support-vector networks,''
\newblock {\em Machine Learning}, vol. 20, no. 3, pp. 273--297, 1995.

\bibitem{CC01a}
C.-C. Chang and C.-J. Lin,
\newblock ``{LIBSVM}: A library for support vector machines,''
\newblock {\em ACM Transactions on Intelligent Systems and Technology}, vol. 2,
  pp. 27:1--27:27, 2011,
\newblock Software available at \url{http://www.csie.ntu.edu.tw/~cjlin/libsvm}.

\end{thebibliography}

\end{document}